\documentclass{article}

\usepackage{authblk}

\usepackage{amsmath,amssymb,amsthm}
\usepackage{booktabs}
\usepackage{multirow}
\usepackage{makecell}
\usepackage{enumerate}
\usepackage{graphicx}
\graphicspath{{figs/}}
\usepackage{url}
\usepackage{color,subcaption}
\usepackage[top=1in, bottom=1in, left=1.5in, right=1.5in]{geometry}
\usepackage{algorithm}
\usepackage{algpseudocode}
\usepackage{caption}
\usepackage{tikz}
\usepackage[counterclockwise]{rotating}
\usepackage[T1]{fontenc}

\usepackage{thmtools}
\usepackage{thm-restate}

\usepackage{booktabs}

\usepackage[round]{natbib}

\newtheorem{theorem}{Theorem}
\newtheorem{proposition}[theorem]{Proposition}
\newtheorem{lemma}[theorem]{Lemma}
\newtheorem{corollary}[theorem]{Corollary}
\theoremstyle{definition}

\newtheorem{example}{Example}[section]
\newtheorem{remark}{Remark}[section]

\newcommand{\diag}{\text{diag}}
\newcommand{\setcomp}{\mathsf{c}}

\newcommand{\hF}{\widehat{F}}

\newcommand{\hR}{\widehat{R}}
\newcommand{\hc}{\hat{c}}

\newcommand{\bX}{\mathbf{X}}
\newcommand{\bb}{\mathbf{b}}
\newcommand{\br}{\mathbf{r}}

\newcommand{\PP}{\mathbb{P}}
\newcommand{\EE}{\mathbb{E}}
\newcommand{\RR}{\mathbb{R}}

\newcommand{\cG}{\mathcal{G}}
\newcommand{\cH}{\mathcal{H}}
\newcommand{\cM}{\mathcal{M}}

\newcommand{\cO}{\mathcal{O}}
\newcommand{\cP}{\mathcal{P}}
\newcommand{\cR}{\mathcal{R}}

\newcommand{\cK}{\mathcal{K}}

\newcommand{\lb}{\left(}
\newcommand{\rb}{\right)}
\newcommand{\FDP}{\textnormal{FDP}}
\newcommand{\FDR}{\textnormal{FDR}}
\newcommand{\BH}{\textnormal{BH}}
\newcommand{\BY}{\textnormal{BY}}
\newcommand{\SU}{\textnormal{SU}}
\newcommand{\dBH}{\textnormal{dBH}}
\newcommand{\dBY}{\textnormal{dBY}}
\newcommand{\dSU}{\textnormal{dSU}}
\newcommand{\sdBH}{\textnormal{s-dBH}}

\newcommand{\dir}{\textnormal{dir.}}
\newcommand{\sign}{\textnormal{sign}}
\newcommand{\lo}{\textnormal{lo}}
\newcommand{\hi}{\textnormal{hi}}

\newcommand{\pmi}{p^{(i \gets 0)}}
\newcommand{\semic}{\,;}

\newcommand{\td}{\tilde}

\newcommand{\simiid}{\stackrel{\text{i.i.d.}}{\sim}}

\definecolor{ruddybrown}{rgb}{0.73, 0.4, 0.16}

\definecolor{light1}{HTML}{Fee0d2}
\definecolor{medium1}{HTML}{Fc9272}
\definecolor{dark1}{HTML}{de2d26}

\definecolor{light2}{HTML}{deebF7}
\definecolor{medium2}{HTML}{9ecae1}
\definecolor{dark2}{HTML}{3182bd}

\definecolor{overlay12}{HTML}{9c7465}

\definecolor{light0}{HTML}{F2F0F7}
\definecolor{mlight0}{HTML}{9e9ac8}
\definecolor{medium0}{HTML}{6a51a3}
\definecolor{hisat0}{HTML}{4a00F3}
\definecolor{dark0}{HTML}{4a1486}

\usetikzlibrary{math}
\usetikzlibrary{patterns}
\usetikzlibrary{fadings}

\begin{tikzfadingfrompicture}[name=myfade]
  \clip (0,0) rectangle (2,2);
  \shade [left color=transparent!0,
          right color=transparent!0]
         (0,0) rectangle (1.4,1.4);
  \shade [left color=transparent!0,
          right color=transparent!100]
         (1.4,0) rectangle (1.85,1.4);
  \shade [left color=transparent!100,
          right color=transparent!100]
         (1.85,0) rectangle (2,1.4);
  \shade [bottom color=transparent!0,
  top color=transparent!100]
  (0,1.4) rectangle (2,1.85);
  \shade [bottom color=transparent!100,
  top color=transparent!100]
  (0,1.85) rectangle (2,2);
\end{tikzfadingfrompicture}

\title{Conditional calibration for false discovery rate control under dependence}
\author{William Fithian and Lihua Lei}
\date{\today}

\begin{document}

\maketitle

\begin{abstract}
  We introduce a new class of methods for finite-sample false discovery rate (FDR) control in multiple testing problems with dependent test statistics where the dependence is fully or partially known. Our approach separately calibrates a data-dependent $p$-value rejection threshold for each hypothesis, relaxing or tightening the threshold as appropriate to target exact FDR control. In addition to our general framework we propose a concrete algorithm, the dependence-adjusted Benjamini--Hochberg (dBH) procedure, which adaptively thresholds the $q$-value for each hypothesis. Under positive regression dependence the dBH procedure uniformly dominates the standard BH procedure, and in general it uniformly dominates the Benjamini--Yekutieli (BY) procedure (also known as BH with log correction). Simulations and real data examples illustrate power gains over competing approaches to FDR control under dependence.
\end{abstract}


\section{Introduction}

Despite the immense popularity of the false discovery rate (FDR) paradigm and the Benjamini--Hochberg (BH) method for large-scale multiple testing \citep{bh95}, the literature on FDR-controlling methods has long been dogged by their uncertain validity when applied to dependent $p$-values. In particular, the BH procedure is only known to control FDR under restrictive positive dependence assumptions, or after a severe correction to the significance level \citep{benjamini2001control}. Apart from specific supervised learning settings where knockoff methods \citep{barber15, candes2018panning} can be applied, practitioners still commonly default to the uncorrected BH method, choosing to forego theoretical guarantees and hope for the best.

This article introduces new methods for finite-sample FDR control under dependence. Our key technical idea is to decompose the FDR according to the additive contribution of each hypothesis, and use conditional inference to adaptively calibrate a separate rejection rule for each hypothesis to directly control its FDR contribution. Equipped with this tool, we prove finite-sample FDR control for a broad class of multiple testing methods. We also propose a concrete algorithm, the {\em dependence-adjusted Benjamini--Hochberg procedure} ($\dBH$), that operates by adaptively calibrating a separate BH $q$-value cutoff for each hypothesis. Specifically, our method rejects $H_i$ if $q_i \leq \hc_i$, where the calibrated threshold $\hc_i$ may be larger or smaller than $\alpha$. Although the dBH procedure can be applied in a wide variety of discrete and continuous, parametric and nonparametric models, the present work emphasizes multivariate Gaussian and linear regression models. We show empirically that our methods perform similarly to $\BH$, but with provable FDR control.

Because $\hc_i$ can be larger than $\alpha$, the dBH procedure can be, and often is, somewhat more powerful than the usual BH procedure. We show that dBH is uniformly more powerful than BH under positive dependence, in the sense that it makes at least as many rejections, almost surely. In addition, versions of the method are uniformly more powerful than the corrected version of BH (known as the Benjamini--Yekutieli (BY) procedure), usually dramatically so.

\subsection{Multiple testing and the false discovery rate}

In a multiple testing problem, an analyst observes a data set $X \sim P$, and rejects a subset of null hypotheses $H_1, \ldots, H_m$.  We assume $P \in \cP$ for some parametric or non-parametric model $\cP$, and each null hypothesis $H_i \subsetneq \cP$ represents a submodel; without loss of generality, the $i$th alternative hypothesis is $\cP \setminus H_i$. We assume the analyst computes a $p$-value $p_i(X)$ to test each $H_i$, where $p_i$ is marginally {\em super-uniform} (i.e., stochastically larger than $\text{Unif}(0,1)$) under $H_i$. Let $\cH_0(P) = \{i:\; P \in H_i\}$ denote the set of true null hypotheses, and $m_0 = |\cH_0|$. Much of our discussion will treat the parametric setting $\cP = \{P_\theta:\; \theta\in \Theta \subseteq \RR^d \}$, often parameterized so that $H_i$ concerns only $\theta_i$, for example $H_i:\; \theta_i = 0$ or $H_i:\; \theta_i \leq 0$.

A {\em multiple testing procedure} is a decision $\cR(X) \subseteq [m] = \{1,\ldots,m\}$ designating the set of rejected hypotheses. An analyst who rejects $H_i$ for each $i \in \cR(X)$ makes $V = |\cR \cap \cH_0|$ false rejections (sometimes called ``false discoveries''). If $R = |\cR|$ is the number of total rejections, \citet{bh95} define the {\em false discovery proportion} (FDP) as
\[
\FDP(\cR(X); P) = \frac{V}{R \vee 1},
\]
where $a\vee b = \max\{a,b\}$ and $a \wedge b = \min\{a,b\}$. The {\em false discovery rate} (FDR) is defined as the expected FDP:
\[
\FDR_P(\cR) = \EE_P\big[\,\FDP(\cR(X); P)\,\big].
\]
A standard goal in multiple testing is to maximize a procedure's power subject to constraining $\sup_{P\in\cP} \FDR_P(\cR) \leq \alpha$ at a pre-set significance level, typically $5\%$, $10\%$, or $20\%$.

The most widely used method for FDR control is the {\em Benjamini--Hochberg (BH) procedure}, an example of the more general class of {\em step-up procedures}. Let $p_{(1)} \leq \cdots \leq p_{(m)}$ denote the order statistics of the $m$ $p$-values. Then the step-up procedure for an increasing sequence of thresholds $0 \leq \Delta(1) \leq \ldots \leq \Delta(m) \leq 1$ finds the largest index $r$ for which $p_{(r)} \leq \Delta(r)$ and rejects all of the corresponding hypotheses up to that index. That is, we reject the hypotheses with the smallest $R(X)$ $p$-values, where
\begin{equation}\label{eq:stepup}
  R(X) = \max \left\{r:\; p_{(r)}(X) \leq \Delta(r)\right\}.
\end{equation}
The $\BH(\alpha)$ procedure takes $\Delta_\alpha(r) = \alpha r / m$. For a general family of thresholds $\Delta_\alpha(r)$ that are non-decreasing in $\alpha$ and $r$, we denote the generic step-up procedure as $\SU_\Delta(\alpha)$. We denote the corresponding testing procedures as $\cR^{\BH(\alpha)}$ and $\cR^{\SU_\Delta(\alpha)}$ respectively.

As $\alpha$ increases, the $\BH(\alpha)$ procedure becomes more liberal, with nested rejection sets. \citet{storey2003positive} defined the {\em $q$-value} as the level at which $H_i$ is barely rejected:
\begin{equation}\label{eq:q-value}
  q_i(X) = \min \left\{\alpha:\; i \in \cR^{\BH(\alpha)}\right\}.
\end{equation}
The same definition may be extended to any step-up procedure. If $\Delta_\alpha(r)$ is right-continuous in $\alpha$, the rejection sets are right-continuous too, and the minimum is always well-defined.

\citet{bh95} showed that the $\BH(\alpha)$ procedure controls FDR at exactly $\alpha m_0/m$ if the $p$-values are independent, but the picture for dependent $p$-values has been more complex. 

\subsection{FDR under dependence}\label{sec:old-strategies}

We can begin to understand the role of dependence by first making a standard decomposition of the FDR according to the contribution of each true null hypotheses:
\begin{equation}\label{eq:decomp-sum}
  \FDR = \EE\left[\frac{V}{R\vee 1}\right] = \sum_{i \in \cH_0} \EE\left[\frac{V_i}{R\vee 1}\right],
\end{equation}
where $V_i = 1\{H_i \text{ rejected}\}$. Under independence, BH controls each term in the sum at $\alpha/m$, attaining FDR control at level $\alpha m_0/m$.

Positive dependence between $V_i$ and $R$ tends to reduce each term in \eqref{eq:decomp-sum}, making methods like BH conservative. In particular, the BH procedure is known to be conservative under {\em positive regression dependence on a subset} (PRDS): For $a,b\in\RR^d$, we say $a \preceq b$ if $a_i \leq b_i$ for all $i=1$, and a set $A \subseteq \RR^d$ is {\em increasing} if $a \in A$ and $a \preceq b$ implies $b \in A$. We say that $p_{-i} = (p_1,p_2,\ldots, p_{i-1},p_{i+1},\ldots, p_m)$ is {\em positive regression dependent} (PRD) on $p_i$ if $\PP(p_{-i} \in A \mid p_i)$ is increasing in $p_i$ for any increasing set $A$. \citet{benjamini2001control} show that the $\BH(\alpha)$ procedure controls FDR conservatively at $\alpha m_0/m$, provided that $p_{-i}$ is PRD on $p_i$, for every $i \in \cH_0$; this condition is called PRDS. Subsequently, many procedures designed to control FDR under independence have also been shown to control FDR under positive dependence as well. Notable exceptions include the Storey-BH method \citep{storey04}, whose estimate of $m_0/m$ can fail badly under dependence\footnote{\citet{benjamini2006adaptive} propose another adaptive method that behaves better under positive dependence}, and adaptive weighting methods such as AdaPT \citep{lei2018adapt} and SABHA \citep{li2019multiple}, whose finite-sample FDR control may be threatened by local random effects that make a cluster of $p$-values smaller together.\footnote{\citet{li2019multiple} derive an upper bound for the FDR inflation in the multivariate Gaussian case, but the bound depends on unknown aspects of the data distribution.}

Unfortunately, the PRDS condition is quite restrictive. It does hold for one-sided testing with multivariate Gaussian test statistics whose pairwise correlations are  all non-negative, or for one- or two-sided testing of uncorrelated multivariate $t$-test statistics. But $p$-values for one-sided testing with any negative pairwise correlations, or for two-sided testing with any correlations at all, no longer satisfy PRDS. 

For general, unspecified dependence, \citet{benjamini2001control} also showed that the much more conservative $\BH(\alpha/L_m)$ procedure controls FDR at level $\alpha$ under arbitrary dependence, where
\[
L_m = \sum_{i=1}^m \frac{1}{i} = \log m + \cO(1).
\]
This method has become known as the {\em Benjamini--Yekutieli (BY)} procedure, or sometimes the log-corrected BH procedure. The proof technique was subsequently generalized in the {\em shape function} approach of \citet{blanchard2008two} who show that if $\nu$ is any probability measure on $\{1,\ldots,m\}$, then the step-up procedure with
\begin{equation}\label{eq:shape-function}
\Delta_\alpha(r) = \frac{\alpha\beta(r)}{m}, \quad \text{where }  \beta(r) = \sum_{i=1}^r i \nu(\{i\})
\end{equation}
also controls FDR under arbitrary dependence between the $p$-values. Taking $\nu(\{i\}) = (i L_m)^{-1}$ recovers the $\BY(\alpha)$ procedure, but \citet{blanchard2008two} suggest other choices that sometimes improve on the BY procedure's power. 

These methods control FDR under worst-case dependence assumptions, but their generality typically comes at a price of substantial conservatism and diminished power compared to the BH procedure. As a result the BH procedure is often still used in applications where PRDS does not hold. This ``off-label'' use of BH owes in part to a widely held belief that, under dependence typically arising in practice, BH is more often conservative than it is anti-conservative \citep[e.g.][]{farcomeni2006more,kim2008effects}.

A second strategy is to prove asymptotic control in regimes where the limiting problem is simpler. For example, \citet{genovese2004stochastic} and \citet{storey04} study regimes where the empirical distributions of null and non-null $p$-values converge to limiting deterministic functions as $m \to \infty$; this line of analysis was developed further in \citet{ferreira2006benjamini} and \citet{farcomeni2007some}. While these analyses provide valuable insights, the results hold only in the limit where $R\to\infty$; but FDR control is often desired in problems where $R$ may be relatively small, even if $m$ is large. \citet{troendle2000stepwise} and \citet{romano2008control} study resampling-based approaches in a different asymptotic regime where $m$ is fixed but the non-null $p$-values converge in probability to zero; in finite samples this is likely an optimistic assumption. 

Recently discovered knockoff methods \citep{barber15} offer an alternative means of FDR control under dependence for testing coefficients in linear regression models, and have been extended to testing conditional independence in supervised learning settings where a model for the joint distribution of predictor variables is available \citep{candes2018panning}. Knockoff methods, which operate by feeding synthetic noise variables to a supervised learning procedure, represent a sharp methodological departure from classical multiple testing procedures like BH. Knockoffs can be more or less powerful than the BH procedure for context-dependent reasons that are not yet fully understood. We discuss their relative strengths and weaknesses compared to classical procedures like BH in Section~\ref{sec:knockoffs}.

In this work, we propose a new methodological framework for controlling FDR under dependence in a wide variety of discrete and continuous, parametric and nonparametric models. Rather than assume worst-case dependence, we begin with a baseline procedure like BH or BY and calibrate its FDR by exploiting full or partial knowledge of the dependence.


\section{FDR control by conditional calibration}

\subsection{Conditional calibration: a new strategy}\label{sec:new-strategy}

Our method operates by adaptively calibrating a separate rejection threshold for each of the $m$ $p$-values to control each term in \eqref{eq:decomp-sum}, which we will call the {\em FDR contribution} of $H_i$. Let $\tau_i(c; X)$ be some possibly data-dependent rejection threshold for $p_i$, with {\em calibration parameter} $c \geq 0$. We assume $\tau_i$ is non-decreasing in $c$ for all $X$, and $\tau_i(0; X) = 0$ almost surely. We will be primarily interested in the {\em effective $\BH$ threshold}, defined as the $p$-value rejection threshold that is ``estimated'' by the $\BH(c)$ procedure:
\begin{equation}\label{eq:q-value-threshold}
  \tau^{\BH}(c; X) = \frac{c R^{\BH(c)}}{m}.
\end{equation}
Because the BH $q$-value $q_i(X)$ is below $\alpha$ if and only if $p_i \leq \tau^\BH(\alpha)$, we can roughly interpret $\tau^{\BH}$ as an inverse $q$-value transformation, and $c$ as a $q$-value cutoff. More generally, we define the {\em effective $\SU_\Delta$ threshold} as $\tau^{\SU_\Delta}(c) = \Delta_{c}(R^{\SU_\Delta(c)})$.

Taking the decomposition in \eqref{eq:decomp-sum} as our starting point, we will aim to calibrate the threshold for $p_i$, choosing $\hc_i$ to directly control the $i$th term in the sum:
\begin{equation}\label{eq:contrib-i}
  \EE_{H_i}\left[\frac{V_i}{R \vee 1}\right] = \sup_{P \in H_i}\EE_P\left[\frac{1\{p_i \leq \tau_i(\hc_i)\}}{R\vee 1}\right] \leq \frac{\alpha}{m}.
\end{equation}
We will use $\EE_{H_i}[\cdot]$ as a shorthand notation for $\sup_{P \in H_i}\EE_P[\cdot]$ throughout.

There are two main challenges in solving for $\hc_i$ in \eqref{eq:contrib-i}. First, the expectation depends in a possibly complicated way on the entire distribution of $X$, whereas $H_i$ typically only constrains the distribution of $p_i$. Our first idea is to achieve \eqref{eq:contrib-i} by controlling a more tractable conditional expectation, given some {\em conditioning statistic} $S_i$ that blocks most or all of the nuisance parameters from influencing the conditional analysis. Often $S_i$ is independent of $p_i$, but we only require that $p_i$ is conditionally superuniform given $S_i$:
\begin{equation}\label{eq:cond-stat}
\sup_{P \in H_i} \PP_{P}(p_i \leq \alpha \mid S_i) \;\leq\; \alpha, \quad \text{for all } \alpha\in [0,1].
\end{equation}

This style of conditioning is a well-established device for handling nuisance parameters in inference problems, especially for exponential family models and permutation tests, and has seen recent application in approaching complex decision problems like multiple testing \citep[e.g.][]{weinstein2013selection, barber15, candes2018panning} and post-selection inference \citep[e.g.][]{tibshirani2016exact,lee2016exact,fithian2014optimal}.

Under independence, \eqref{eq:cond-stat} is satisfied with $S_i = p_{-i}$. A standard FDR control proof for the BH procedure, introduced in \citet{benjamini2001control}, conditions on $p_{-i}$ and applies the following key lemma, whose proof is given in Appendix~\ref{app:proofs}:

\begin{restatable}{lemma}{lempmi}
\label{lem:p-minus-i}
  Let $\pmi = (p_1,\ldots,p_{i-1},0,p_{i+1},\ldots,p_m)$. If $\cR$ is a step-up procedure with threshold sequence $\Delta(1),\ldots,\Delta(m)$, then the following are equivalent:
  \begin{enumerate}
  \item $p_i \leq \Delta(R(\pmi))$,
  \item $i\in \cR(p)$, and
  \item $\cR(p) = \cR(\pmi)$.
  \end{enumerate}
\end{restatable}

Let $R^0 = R(\pmi)$, which depends only on $p_{-i}$. Then for the standard BH procedure under independence, applying Lemma~\ref{lem:p-minus-i} gives
\begin{align*}
  \EE_{H_i}\left[ \frac{V_i}{R\vee 1} \mid p_{-i}\right] 
  \;=\; \EE_{H_i}\left[ \frac{1\{p_i \leq \alpha R^0/m\}}{R^0} \mid p_{-i}\right] 
  \;\leq\; \frac{1}{R^0} \,\frac{\alpha R^0}{m} \;=\; \frac{\alpha}{m}.
\end{align*}
Marginalizing over $p_{-i}$ and summing over $i \in \cH_0$ yields $\FDR \leq \alpha m_0 / m$.

For dependent $p$-values, $p_{-i}$ will not in general satisfy \eqref{eq:cond-stat}. As a simple example, suppose $Z \sim N_m(\mu, \Sigma)$ where $\Sigma$ is a known covariance matrix with diagonal entries $\Sigma_{ii}=1$, and we wish to test $H_i:\; \mu_i = 0$ against a one- or two-sided alternative. If $\Sigma_{ij}\neq 0$ then $p_i$ and $p_j$ are not independent and the distribution of $p_i$ given $p_{-i}$ depends on $\mu_j$. However, the conditioning statistic $S_i = Z_{-i} - \Sigma_{-i,i} Z_i$ is independent of $Z_i$, and after conditioning on $S_i$ the data distribution depends only on $\mu_i$, which is fixed at zero under the null. 

In this example the data set can be reconstructed from $S_i$ and $Z_i$, and if the $p$-values are one-sided then it can also be reconstructed from $S_i$ and $p_i$. Figure~\ref{fig:cond-comp} illustrates the conditional FDR contribution of $H_i$ for one-sided testing under three conditions: independence ($\Sigma = I_m$), where the contribution is exactly $\alpha/m$; positive-dependence ($\Sigma_{ij} \geq 0$ for all $i,j$), where the contribution is below $\alpha/m$; and worst-case dependence, where the contribution can be as high as $\alpha L_m/m$. In Figure~\ref{fig:cond-comp:posdep}, the red line $R(p_i; S_i)^{-1}$ is increasing in $p_i$ because, fixing $S_i$, the other $p$-values are increasing functions of $p_i$.

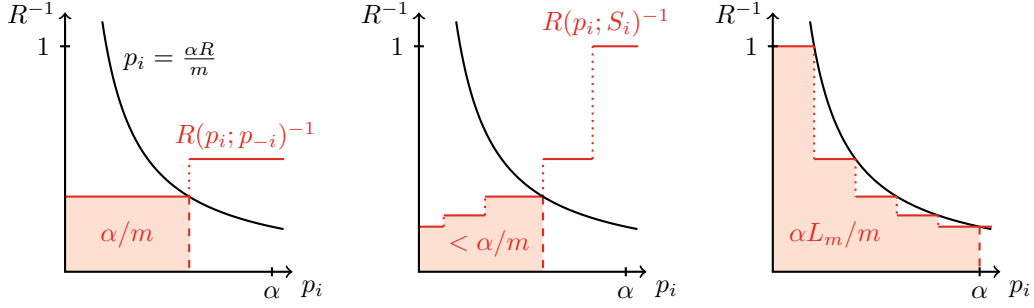
\begin{figure}
  \centering
  \begin{subfigure}{.31\hsize}
    \centering

  \begin{tikzpicture}[xscale=55,yscale=3]
    \path[fill=light1] (0,0) rectangle (.15/5,1/3);

    \draw[thick, domain=0.009:0.053, samples=100] plot (\x, {0.01/\x});

    \draw[thick, dark1] (0,1/3) --  (.15/5,1/3);
    \draw[thick, dotted, dark1] (.15/5,1/3) -- (.15/5, 1/2);
    \draw[thick, dark1] (.15/5, 1/2) -- (.053, 1/2);
    \draw[thick, dashed, dark1] (.15/5,1/3) -- (.15/5,0);

    \node [dark1] at (.015,1/6) {$\alpha/m$};
    \node [above right] at (.01*1.2,1/1.2) {$p_i = \frac{\alpha R}{m}$};    
    \node[dark1, above left] at (.063, 1/2) {$R(p_i;p_{-i})^{-1}$};

    \draw[thick, <->] (0,1.15) -- (0,0) -- (.055,0);
    \node [below right] at (.055,0) {$p_i$};
    \node [left] at (0,1.15) {$R^{-1}$};
    \draw [thick] (.05,-.02) node[below]{$\alpha$} -- (.05,0.02);
    \draw [thick] (-.0012,1) node[left]{$1$} -- (.0012,1);
\end{tikzpicture}
    \caption{Independent case: $R$ is constant on $\{p_i \leq \alpha R/m\}$.}
    \label{fig:cond-comp:indep}
  \end{subfigure}
  \hspace{.01\hsize}
  \begin{subfigure}{.31\hsize}
    \centering

  \begin{tikzpicture}[xscale=55,yscale=3]
    \path [fill=light1] (0,0) -- (0,1/5) --  (.03/5,1/5) --  (.03/5,1/4) --  (.08/5,1/4) --  (.08/5,1/3) --  (.15/5,1/3) -- (.15/5,0) -- cycle;
    \draw[thick, dashed, dark1] (.15/5,1/3) -- (.15/5,0);
    

    \draw[thick, dark1] (0,1/5) --  (.03/5,1/5);
    \draw[thick, dark1] (.03/5,1/4) --  (.08/5,1/4);
    \draw[thick, dark1] (.08/5,1/3) --  (.15/5,1/3);
    \draw[thick, dark1] (.15/5,1/2) -- (.21/5,1/2);
    \draw[thick, dark1] (.21/5, 1) -- (.053, 1);
    
    \draw[thick, dotted, dark1] (.03/5,1/5) --  (.03/5,1/4);
    \draw[thick, dotted, dark1] (.08/5,1/4) --  (.08/5,1/3);
    \draw[thick, dotted, dark1] (.15/5,1/3) --  (.15/5,1/2);
    \draw[thick, dotted, dark1] (.21/5,1/2) -- (.21/5, 1);

    \node [dark1] at (.017,1/8) { $<\alpha/m$};
    \node[dark1, above left] at (.063, 1) {$R(p_i;S_i)^{-1}$};

    \draw[thick, domain=0.009:0.053, samples=100] plot (\x, {0.01/\x});

    \draw[thick, <->] (0,1.15) -- (0,0) -- (.055,0);
    \node [below right] at (.055,0) {$p_i$};
    \node [left] at (0,1.15) {$R^{-1}$};
    \draw [thick] (.05,-.02) node[below]{$\alpha$} -- (.05,0.02);
    \draw [thick] (-.0012,1) node[left]{$1$} -- (.0012,1);
\end{tikzpicture}
    \caption{Positive-dependent case: $R$ is decreasing in $p_i$}
    \label{fig:cond-comp:posdep}
  \end{subfigure}
  \hspace{.01\hsize}
  \begin{subfigure}{.31\hsize}
    \centering

\begin{tikzpicture}[xscale=55,yscale=3]
    \path [fill=light1] 
      (0,1) -- (.05/5,1) --
      (.05/5,.5) -- (.1/5,.5) --
      (.1/5,1/3) -- (.15/5,1/3) --
      (.15/5,.25) -- (.2/5,.25) --
      (.2/5,.2) -- (.25/5,.2) -- 
      (.25/5,0) -- (0,0) -- cycle;

    \draw[thick, domain=0.009:0.053, samples=100] plot (\x, {0.01/\x});

    \draw[thick, dark1] (0,1) -- (.05/5,1);
    \draw[thick, dark1] (.05/5,.5) -- (.1/5,.5);
    \draw[thick, dark1] (.1/5,1/3) -- (.15/5,1/3);
    \draw[thick, dark1] (.15/5,.25) -- (.2/5,.25);
    \draw[thick, dark1] (.2/5,.2) -- (.053,.2);
    
    \draw[thick, dark1, dotted] (.05/5,1) -- (.05/5,.5);
    \draw[thick, dark1, dotted] (.1/5,.5) -- (.1/5,1/3);
    \draw[thick, dark1, dotted] (.15/5,1/3) -- (.15/5,.25);
    \draw[thick, dark1, dotted] (.2/5,.25)  -- (.2/5,.2);

    \node [dark1] at (.015,1/6) {$\alpha L_m / m$};
    \draw[thick, dashed, dark1] (.25/5,1/5) -- (.25/5,0);

    \draw[thick, <->] (0,1.15) -- (0,0) -- (.055,0);
    \node [below right] at (.055,0) {$p_i$};
    \node [left] at (0,1.15) {$R^{-1}$};
    \draw [thick] (.05,-.02) node[below]{$\alpha$} -- (.05,0.02);
    \draw [thick] (-.0012,1) node[left]{$1$} -- (.0012,1);
\end{tikzpicture}
    \caption{Worst-case: $R$ is always just large enough to reject $p_i$.}
    \label{fig:cond-comp:worst}
  \end{subfigure}
  \caption{Visualizing the conditional FDR contribution for $\BH(\alpha)$ under two simplifying assumptions: (i) $p_i$ is uniform under the null, and (ii) the rejection set is a function of $p_i$ and $S_i$. $R^{-1}$ is shown as the broken red line. The $\BH(\alpha)$ procedure rejects $H_i$ when $(p_i, R^{-1})$ is below the hyperbola $p_i = \alpha R / m$ and the conditional FDR contribution is given by $\EE\left[\frac{V_i}{R \vee 1} \mid S_i\right] = \int_{u \leq \alpha R(u)/m} R^{-1}(u)\,du$, the area of the light red region.}
  \label{fig:cond-comp}
\end{figure}

The second main challenge is that the number of rejections $R$ in the denominator depends on all of $\hc_1,\ldots,\hc_m$, so all $m$ calibration problems are coupled to one another. To deal with this, we substitute an ``estimator'' $\hR_i(X) \geq 1$ of the eventual value of $R_i(X) = |\cR(X) \cup \{i\}|$, the number of rejections if we also include $H_i$.  Because $R_i = R$ for $i \in \cR(X)$, we have $V_i / R_i = V_i/(R \vee 1)$ almost surely. Ideally, $\hR_i$ should be an accurate and easily computable lower-bound for $R_i$.

\subsection{Our method}

We now present a generic two-step FDR-controlling method, possibly with a third randomization step to handle cases where $\hR_i$ fails to lower-bound $R_i$:

\paragraph{Step 1: Calibration.} First, we use $\hR_i$ to estimate the conditional FDR contribution $\EE_{H_i}\left[V_i/R_i \mid S_i\right]$ as a function of the calibration parameter, and recalibrate as appropriate to control the conditional expectation at $\alpha/m$:
\begin{equation}\label{eq:def-hc}
  g_i^*(c \semic S_i) = \sup_{P\in H_i} \EE_P\left[\,\frac{1\{p_i \leq \tau_i(c)\}}{\hR_i} \mid S_i \,\right] \;\;\leq \;\; \frac{\alpha}{m},
\end{equation}
suppressing the dependence of $p_i$, $\tau_i(c)$, and $\hR_i$ on $X$ for compactness of notation.

The function $g_i^*$ is almost surely non-decreasing in $c$ with $g_i^*(0) = 0$ because $\tau_i(0)=0$ by assumption. Let $c_i^*(S_i) \leq \infty$ denote the least upper bound for the set of allowed $c$ values, which is either $[0,c_i^*]$ or $[0,c_i^*)$ (note that neither $\tau_i$ nor $g_i^*$ is necessarily continuous in $c$). If $g_i^*(c_i^*) \leq \alpha/m$, we can set $\hc_i(S_i) = c_i^*$; otherwise we can take a non-decreasing sequence $(\hc_{i,t}(S_i))_{t = 1}^\infty$ converging to $c_i^*$ from below, such as $\hc_{i,t} = (c_i^* - 1/t) \vee 0$. We say $\hc_i$ is {\em maximal} if, almost surely, either $c_i^*$ satisfies \eqref{eq:def-hc} and $\hc_i = c_i^*$, or $c_i^*$ does not satisfy \eqref{eq:def-hc} and $\lim_t \hc_{i,t} = c_i^*$. If $c_i^*(S_i)$ is difficult to calculate, it is enough to assume only that $\hc_i(S_i)$ is any value satisfying \eqref{eq:def-hc} almost surely.

\paragraph{Step 2: Initial rejection.}
Next, we initialize the rejection set, via:
\[
\cR_+ = \left\{i:\; p_i \leq \tau_i(\hc_i) \right\}.
\]
If $\hc_i$ is a sequence, the condition $p_i \leq \tau_i(\hc_i)$ is understood to mean $p_i \leq \tau_i(\hc_{i,t})$ for some $t$ (which is not equivalent to $p_i \leq \tau_i(\lim_t \hc_{i,t})$). To limit notational bloat in our prose, we will discuss $\hc_i$ as though it is a single value, but our results all apply to the general case.

Let $R_+  = |\cR_+|$. If $R_+ \geq \hR_i$ for all $i \in \cR_+$, then we can halt the procedure with $\cR = \cR_+$. Otherwise, we may need to prune the rejection set further.

\paragraph{Step 3 (if necessary): Randomized pruning.}
If there is some $i \in \cR_+$ for which $\hR_i > R_+$, then we must prune the rejection set via a secondary BH procedure. For user-generated uniform random variables $u_1,\ldots,u_m \simiid \text{Unif}(0,1)$, let
\begin{equation}\label{eq:bh2}
R(X; u) = \max \left\{r:\; \left|\{i \in \cR_+:\; u_i \leq r/\hR_i\}\right| \geq r\right\},
\end{equation}
and reject $H_i$ for the $R$ indices with $i \in \cR_+$ and $u_i \leq R/\hR_i$. The procedure in \eqref{eq:bh2} is equivalent to the $\BH(1)$ procedure on ``$p$-values'' $\tilde{p}_i = u_i \hR_i / R_+$ for $i \in \cR_+$. We can skip this step if $R_+ \geq \hR_i$ for all $i\in\cR_+$; in that case all $u_i\hR_i/R_+ \leq 1$, so $R = R_+$.

We show next that this procedure controls FDR at the desired level.

\begin{theorem}[FDR control]\label{thm:fdr-control}
  Assume that \eqref{eq:cond-stat} holds, and $\hc_{i,t}(S_i)$ is chosen to guarantee \eqref{eq:def-hc}, for all $i$ and $t$. Then the three-step procedure defined above controls the FDR at or below level $\alpha m_0/m$.
\end{theorem}
\begin{proof}
  It is sufficient to show that $\EE\left[V_i/R_i \right] \leq \alpha/m$ for every $i \in \cH_0$. For $i \in \cR_+$ let $R_i^* = R(X; u^{(i \gets 0)})$, which is independent of $u_i$. By Lemma~\ref{lem:p-minus-i} applied to the secondary BH procedure, $R_i^* = R$ on the event $\{V_i = 1\}$. As a result, we can write
  \begin{align}
    \EE\left[\frac{V_i}{R \vee 1}\right]
    &= \EE\left[
      \frac{1\left\{i \in \cR_+\right\} \cdot 1\left\{u_i \leq R/\hR_i\right\}}{R \vee 1}
      \right]\\
    &= \EE\left[
      \frac{1\left\{i \in \cR_+\right\} \cdot 1\left\{u_i \leq R_i^*/\hR_i\right\}}{R_i^*}
      \right]\\
    &= \EE\left[\;\EE\left[
      \frac{1\left\{i \in \cR_+\right\} \cdot 1\left\{u_i \leq R_i^*/\hR_i\right\}}{R_i^*}
      \mid X, u_{-i}\right] \;\right]\\\label{eq:underestimation-conservatism}
    &\leq \EE\left[
      \frac{1\left\{i \in \cR_+\right\}}{\hR_i}
      \right]\\
    \label{eq:dominated}
    &= \lim_{t\to\infty} \EE\left[
      \frac{1\left\{p_i \leq \tau_i(\hc_{i,t})\right\}}{\hR_i}
      \right]\\
    \label{eq:final-cond-exp}
    &= \lim_{t\to\infty} \EE \left[ \EE\left[ \frac{1\left\{p_i \leq \tau_i(\hc_{i,t})\right\}}{\hR_i} \mid S_i
      \right] \right].
  \end{align}
  We can move the limit outside the integral in \eqref{eq:dominated} by monotone convergence. Under $H_i$, the last expression is no larger than $\alpha/m$, completing the proof.
\end{proof}

We pause to make several further observations:

\begin{remark}\label{rem:rhs}
  If we modify the calibration step replacing $\alpha/m$ by a generic, possibly data-dependent upper bound $\kappa_i(S_i)$ in inequality~\eqref{eq:def-hc}, then the proof of Theorem~\ref{thm:fdr-control} trivially generalizes to guarantee FDR control at level $\sum_{i \in \cH_0} \EE_P \kappa_i(S_i)$. Section~\ref{sec:discussion} discusses two extensions of the procedure where this additional flexibility is useful: adaptive estimation of $m_0/m$ and adaptive hypothesis weighting.
\end{remark}

\begin{remark}\label{rem:conservatism}
  From the proof of Theorem~\ref{thm:fdr-control} we see that there are three possible sources of conservatism in the above procedure. First, we typically have $m_0 < m$. Second, the inequality in \eqref{eq:underestimation-conservatism} is the price we pay in conservatism when $\hR_i$ underestimates $R_i$. Third, the conditional expectation in~\eqref{eq:final-cond-exp} may not attain $\alpha/m$, either because $\hc_i$ is not maximal, because $p_i$ has a discrete distribution, or because $P$ does not attain the supremum in \eqref{eq:def-hc} (for example if $P$ lies in the interior of $H_i$).
\end{remark}

\begin{remark}\label{rem:avoid-randomization}
  When we are forced to carry out the randomized pruning in Step 3, every rejection in $\cR_+$ is at risk, and there is a real danger that we could prune even a hypothesis for which the $p$-value $p_i$ is extremely small. Even if the number of pruned rejections is small in expectation, we believe it is scientifically preferable to use a procedure that obeys the sufficiency principle and avoids randomization, even if the randomization has a negligible effect on power calculations. We call a calibrated procedure {\em safe} if pruning is never necessary.
\end{remark}

To operationalize our method, we must fill in the details of what threshold family $\tau_i$ and estimator $\hR_i$ we use, how we identify the conditioning statistic $S_i$, and how we calibrate the threshold in practice. The next sections address these issues in turn.

\subsection{The dependence-adjusted BH and BY procedure}\label{sec:dbh-procedure}

While Theorem~\ref{thm:fdr-control} proves FDR control for a broad class of procedures, our empirical results focus on special cases of our method that are designed to couple tightly with the BH and BY procedures. If we use the three-step method of the previous section with the effective BH threshold $\tau_i = \tau^{\BH}$ and estimator $\hR_i = R_i^{\BH(\gamma \alpha)} = |\cR^{\BH(\gamma \alpha)} \cup\{i\}|$, we call the resulting method the {\em dependence-adjusted BH procedure}, which we denote $\dBH_\gamma(\alpha)$. In the special case where we take $\gamma = 1/L_m$, we call the resulting method the {\em dependence-adjusted BY procedure}, denoted $\dBY(\alpha)$. More generally, let the $\dSU_{\gamma, \Delta}(\alpha)$ procedure use threshold $\tau_i = \tau^{\SU_\Delta}$ and estimator $\hR_i = R_i^{\SU_{\Delta}(\gamma\alpha)}$. All of these definitionally use maximal $\hc_i$.

We can interpret the $\dBH_\gamma(\alpha)$ calibration parameters $\hc_1,\ldots,\hc_m$ as calibrated $q$-value rejection thresholds for the respective hypotheses, since $q_i \leq c$ if and only if $p_i \leq \tau^{\BH}(c)$. As a result, we have
\begin{equation}\label{eq:R-sandwich}
\cR^{\BH(\min_i \hc_i)} \,\subseteq\, \cR_+ \,\subseteq\, \cR^{\BH(\max_i \hc_i)}.
\end{equation}
Provided that $\min_i\hc_i \geq \gamma \alpha$, the randomization step is avoided and we can replace $\cR_+$ with $\cR$ in \eqref{eq:R-sandwich}. In our simulations we take $\gamma = 0.9$ as a conservative choice except in the positive-dependent case, where we take $\gamma = 1$. Figure~\ref{fig:dBH} illustrates how conditional calibration operates for the $\dBH_1$ and $\dBY$ procedures.

Under a slight strengthening of the PRDS condition, the $\dBH_1(\alpha)$ procedure is {\em uniformly more powerful} than the BH procedure, meaning $\cR^{\dBH_1} \supseteq \cR^{\BH}$ almost surely. For a given conditioning statistic $S_i$ we say the $p$-values $p_{-i}$ are {\em conditionally positive regression dependent} (CPRD) if $\PP(p_{-i} \in A \mid p_i, S_i)$ is almost surely increasing in $p_i$ for any increasing set $A$. If $p_{-i}$ is CPRD on $p_i$ for all $i \in \cH_0$ we say the $p$-values are {\em CPRD on a subset} (CPRDS). Note CPRDS implies PRDS after marginalizing over $S_i$, but PRDS does not necessarily imply CPRDS.

\begin{restatable}{theorem}{thmsafe}
\label{thm:safe}
  Assume $\hc_1,\ldots,\hc_m$ are maximal. Then
  \begin{enumerate}
  \item If the $p$-values are independent with $p_i$ uniform under $H_i$, then the $\dBH_1(\alpha)$ procedure with $S_i = p_{-i}$ is identical to the $\BH(\alpha)$ procedure.
  \item If the $p$-values are CPRDS for all $P \in \cP$, then the $\dBH_1(\alpha)$ procedure is safe, and uniformly more powerful than the $\BH(\alpha)$ procedure. 
  \item For arbitrary dependence, the $\dBY(\alpha)$ procedure is safe, and uniformly more powerful than the $\BY(\alpha)$ procedure.
  \item Assume the thresholds $\Delta_\alpha$ are of the form \eqref{eq:shape-function}. Then for arbitrary dependence, the $\dSU_\Delta(\alpha)$ procedure is safe, and uniformly more powerful than the $\SU_\Delta(\alpha)$ procedure.
  \end{enumerate}
\end{restatable}

The proof of Theorem~\ref{thm:safe} adapts and extends proofs in \citet{benjamini2001control} and \citet{blanchard2008two}, and may be found in Appendix~\ref{app:proofs}.

\begin{remark}\label{rem:cprd}
  The proof of Theorem~\ref{thm:safe} only involves values of $p_i$ in the interval $p_i \in [0,\alpha]$, and the only increasing sets that appear in the proof are of the form $\{\hR_i \leq r\}$. For the purpose of applying Theorem~\ref{thm:safe}, then, we could relax the definition of CPRD to require only that $\PP(\hR_i \leq r \mid p_i, S_i)$ is increasing for $p_i\in [0,\alpha]$. As a result, the second conclusion of Theorem~\ref{thm:safe} applies to one-sided testing with uncorrelated multivariate $t$-statistics even though the $p$-values are neither PRDS nor CPRDS, as we show in Section~\ref{sec:tstats}.
\end{remark}
\begin{remark}
  In many of our examples, including one-sided multivariate Gaussian testing, the entire data set can be reconstructed from $p_i$ and $S_i$. In that case, the conditional probability is always 0 or 1, and a sufficient condition for CPRD is that, fixing $S_i$, every other $p_{j}$ is an increasing function of $p_i$.
\end{remark}

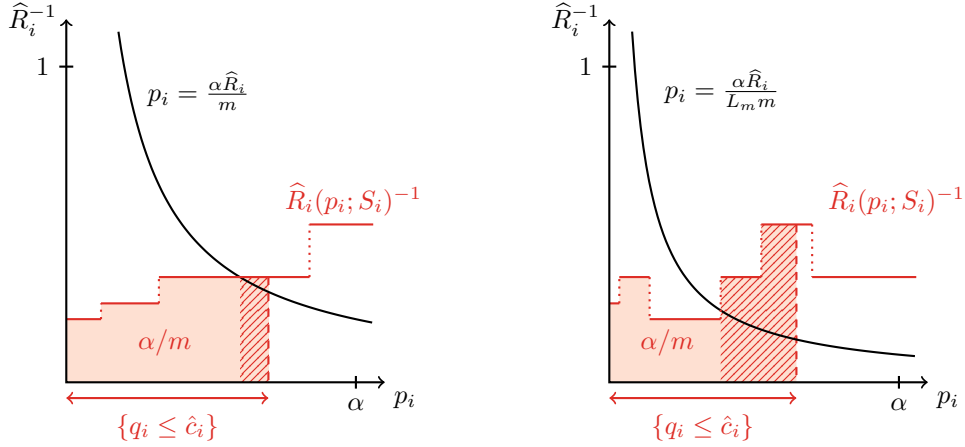
\begin{figure}
  \centering
  \begin{subfigure}{.46\hsize}
    \centering

  \begin{tikzpicture}[xscale=55*1.4,yscale=3*1.4]
    \path [fill=light1] (0,0) -- (0,1/5) --  (.03/5,1/5) --  (.03/5,1/4) --  (.08/5,1/4) --  (.08/5,1/3) --  (.15/5,1/3) -- (.15/5,0) -- cycle;
    
    \path [fill=light1] (.15/5,0) rectangle (0.1745/5,1/3);
    \path [pattern = north east lines, pattern color=dark1] (.15/5,0) rectangle (0.1745/5,1/3);
    \draw[thick, dashed, dark1] (0.1745/5,1/3) -- (0.1745/5,0);

    \draw[thick, domain=0.009:0.053, samples=100] plot (\x, {0.01/\x});

    \draw[thick, dark1] (0,1/5) --  (.03/5,1/5);
    \draw[thick, dark1] (.03/5,1/4) --  (.08/5,1/4);
    \draw[thick, dark1] (.08/5,1/3) --  (.15/5,1/3);
    \draw[thick, dark1] (.15/5,1/3) -- (.21/5,1/3);
    \draw[thick, dark1] (.21/5, 1/2) -- (.053, 1/2);
    
    \draw[thick, dotted, dark1] (.03/5,1/5) --  (.03/5,1/4);
    \draw[thick, dotted, dark1] (.08/5,1/4) --  (.08/5,1/3);
    \draw[thick, dotted, dark1] (.21/5,1/3) -- (.21/5, 1/2);

    \node [dark1] at (.017,1/8) { $\alpha/m$};
    \node [above right] at (.01*1.25,1/1.2) {$p_i = \frac{\alpha \hR_i}{m}$};    
    \node[dark1, above left] at (.063, 1/2) {$\hR_i(p_i;S_i)^{-1}$};
    \draw [thick, dark1, <->] (0, -.05) -- (0.1745/5,-.05);
    \node [below, dark1] at (0.1745/5/2,-.07) {$\{q_i \leq \hc_i\}$};

    \draw[thick, <->] (0,1.15) -- (0,0) -- (.055,0);
    \node [below right] at (.055,0) {$p_i$};
    \node [left] at (0,1.15) {$\hR_i^{-1}$};
    \draw [thick] (.05,-.02) node[below]{$\alpha$} -- (.05,0.02);
    \draw [thick] (-.0012,1) node[left]{$1$} -- (.0012,1);
\end{tikzpicture}
    \caption{The $\dBH_1(\alpha)$ procedure in the positive-dependent case.}
    \label{fig:dBH:pos-dep}
  \end{subfigure}
  \hspace{.04\hsize}
  \begin{subfigure}{.46\hsize}
    \centering

\begin{tikzpicture}[xscale=55*1.4,yscale=3*1.4]
    \path [fill=light1] 
      (0,1/4) -- (.02/5/2.2833333,1/4) --
      (.02/5/2.2833333,1/3) -- (.08/5/2.2833333,1/3) --
      (.08/5/2.2833333,1/5) -- (.22/5/2.2833333,1/5) --
      (.22/5/2.2833333,1/3) -- (.3/5/2.2833333,1/3) --
      (.3/5/2.2833333,1/2) -- (.369/5/2.2833333,1/2) -- 
      (.369/5/2.2833333,0) -- (0,0) -- cycle;
    \path [pattern = north east lines, pattern color=dark1]
      (.22/5/2.2833333,1/3) -- (.3/5/2.2833333,1/3) --
      (.3/5/2.2833333,1/2) -- (.369/5/2.2833333,1/2) -- 
      (.369/5/2.2833333,0) -- (.22/5/2.2833333,0) -- cycle;
    \draw[thick, dashed, dark1] (.369/5/2.2833333,1/2) -- (.369/5/2.2833333,0);

    \draw[thick, domain=0.0039416:0.053, samples=100] plot (\x, {0.01/\x/2.2833333});

    \draw[thick, dark1] (0,1/4) -- (.02/5/2.2833333,1/4);
    \draw[thick, dark1] (.02/5/2.2833333,1/3) -- (.08/5/2.2833333,1/3);
    \draw[thick, dark1] (.08/5/2.2833333,1/5) -- (.22/5/2.2833333,1/5);
    \draw[thick, dark1] (.22/5/2.2833333,1/3) -- (.3/5/2.2833333,1/3);
    \draw[thick, dark1] (.3/5/2.2833333,1/2) -- (.4/5/2.2833333,1/2);
    \draw[thick, dark1] (.4/5/2.2833333,1/3) -- (.053,1/3);

    \draw[thick, dark1, dotted] (.02/5/2.2833333,1/4) --  (.02/5/2.2833333,1/3);
    \draw[thick, dark1, dotted] (.08/5/2.2833333,1/3) --  (.08/5/2.2833333,1/5);
    \draw[thick, dark1, dotted] (.22/5/2.2833333,1/5) --  (.22/5/2.2833333,1/3);
    \draw[thick, dark1, dotted] (.3/5/2.2833333,1/3) --  (.3/5/2.2833333,1/2);
    \draw[thick, dark1, dotted] (.4/5/2.2833333,1/2) --  (.4/5/2.2833333,1/3);

    

    \node [dark1] at (.01,1/8) {$\alpha / m$};
    \node [above right] at (.01*.8,1/1.2) {$p_i = \frac{\alpha \hR_i}{L_m m}$};    
    \node[dark1, above left] at (.063, 1/2) {$\hR_i(p_i;S_i)^{-1}$};
    \draw [thick, dark1, <->] (0, -.05) -- (.369/5/2.2833333,-.05);
    \node [below, dark1] at (.369/5/2.2833333/2,-.07) {$\{q_i \leq \hc_i\}$};

    \draw[thick, <->] (0,1.15) -- (0,0) -- (.055,0);
    \node [below right] at (.055,0) {$p_i$};
    \node [left] at (0,1.15) {$\hR_i^{-1}$};
    \draw [thick] (.05,-.02) node[below]{$\alpha$} -- (.05,0.02);
    \draw [thick] (-.0012,1) node[left]{$1$} -- (.0012,1);
\end{tikzpicture}
    \caption{The $\dBY(\alpha)$ procedure, which is often much more powerful than $\BY(\alpha)$.}
    \label{fig:dBH:dBY}
  \end{subfigure}
  \caption{Illustration of conditional calibration for the $\dBH_1$ and $\dBY$ procedures.  The red line in each plot represents the inverse of $\hR_i = |\cR^{\BH(\gamma \alpha)} \cup \{i\}|$ (assumed here as in Figure~\ref{fig:cond-comp} to be a deterministic function of $p_i$ and $S_i$) for $\gamma=1$ (left) and $\gamma = 1/L_m$ (right). If $p_i$ is uniform under the null, the conditional FDR contribution is estimated as $\int_{u:\; q_i(u) \leq \hc_i} \hR_i^{-1}(u)\,du$. The hatched region is added in the calibration step to ``top up'' the FDR contribution to $\alpha/m$. The set $\{u:\; q_i(u) \leq \hc_i\}$ is not always a contiguous interval.}
  \label{fig:dBH}
\end{figure}

Recognizing $\hc_i$ as a $q$-value threshold has a convenient computational interpretation as well, because $i \in \cR_+$ if and only if $g_i^*(q_i \semic S_i) \leq \alpha / m$. As a result, we usually do not need to calculate $\hc_i$ explicitly, but can simply evaluate the conditional expectation plugging in $c = q_i(X)$. Section~\ref{sec:computation} discusses computational considerations in more detail.

\subsection{Identifying the conditioning statistic $S_i$}\label{sec:find-s}

There are several desiderata for a good conditioning statistic. Most importantly, recall that our method's validity depends on $p_i$ being a valid conditional $p$-value, so that \eqref{eq:cond-stat} holds. To facilitate calibration, $S_i$ should also eliminate or mitigate the influence of nuisance parameters on the conditional distribution of $X$. Finally, the conditional distribution under $H_i$ should be analytically and/or computationally tractable. Calibration is conceptually simplified if $S_i$ is a sufficient statistic for the null submodel $H_i$, so that the conditional distribution of $X$ is known under $H_i$. In that case, we say $H_i$ is {\em conditionally simple}, and $g_i^*(c \semic S_i)$ is an integral we can directly evaluate. Otherwise, we say $H_i$ is {\em conditionally composite}, and $P_i^*$ is {\em least favorable} for calibrating $\hc_i$ if it almost surely attains the supremum:
\begin{equation}\label{eq:least-fav}
g_i^*(c \semic S_i) = \EE_{P_i^*}\left[\,\frac{1\{p_i \leq \tau_i(c)\}}{\hR_i} \mid S_i \,\right], \quad \text{for all } c \geq 0.
\end{equation}

We discuss two primary examples where the choice of conditioning statistic is fairly natural: parametric exponential family models and nonparametric models with constraints on the dependence graph.

\begin{example}[Exponential families]
In exponential family models, there is a natural choice of $S_i$ which follows from the classical theory of conditional testing in the style of \citet{lehmann1955completeness}. Suppose our model arises from a full-rank exponential family in canonical form:
\begin{equation}\label{eq:exp-fam}
  X \sim f_{\theta}(x) = e^{\theta'T(x) - A(\theta)}\, f_0(x), \quad \theta \in \Theta \subset \RR^d,
\end{equation}
and for $i = 1,\ldots, m \leq d$, $H_i$ takes the form $H_i:\theta_i = 0$ or $H_i:\;\theta_i\leq 0$. In this setting, the uniformly most powerful unbiased (UMPU) test rejects $H_i$ when $T_i(X)$ is extreme, conditional on the value of $S_i = T_{-i}(X)$; see e.g. \citet{lehmann2005testing}. As a result, $T_{-i}(X)$ makes a natural choice of test statistic, because it eliminates nuisance parameters and because $p_i$ is conditionally valid by construction, so the hypothesis $H_i$ is conditionally simple and we can evaluate $g_i^*(c)$ directly. 
For one-sided testing, the null hypothesis $H_i:\; \theta_i \leq 0$ is conditionally composite, but we will see in Section~\ref{sec:one-sided} that under mild conditions setting $\theta_i=0$ is least favorable.

Section~\ref{sec:examples} and Appendix~\ref{app:examples} discuss a variety of multiple testing problems in this vein, including multivariate Gaussian test statistics, testing in linear models, multiple comparisons with binary responses, and conditional independence testing in Gaussian graphical models.
\end{example}

\begin{example}[Nonparametric models with dependence constraints]\label{ex:nonpar-graph}
As a nonparametric example, consider observing only $p$-values $p \in [0,1]^m$, where the dependence is not parametrically specified but there are constraints on the dependence graph between $p$-values. Specifically, we assume that each $p$-value has a neighborhood $\cM_i\subseteq [m]$ for which $S_i = p_{\cM_i^\setcomp}$ is independent of $p_i$. We assume nothing else about $p_i$ except that it is uniform or superuniform under $H_i$. In this case the null hypotheses are not conditionally simple since we cannot directly sample from the distribution of $p$ given $S_i$. Instead, we must assume worst-case dependence of $p_{\cM_i}$ on $p_i$: the least favorable distribution $P_i^*$ adversarially configures $p_{\cM_i\setminus\{i\}}$ to maximize the integrand in \eqref{eq:def-hc}, conditional on $p_i$ and $S_i = p_{\cM_i^\setcomp}$. This may nevertheless allow for considerable improvement on the BY method or the shape function approach of \citet{blanchard2008two}, which effectively assume worst-case dependence of {\em all} other $p$-values on $p_i$, corresponding to the special case where all $\cM_i = [m]$. We defer further exploration of this example to future work.
\end{example}

\subsection{Recursive refinement of $\hR_i$}\label{sec:recursive}

The random variable $\hR_i(X)$ is an unusual estimator in that, by the time the procedure terminates, we will have computed the estimand $R_i(X)$. If they differ, it seems natural to re-run the procedure substituting the ``correct value'' for the inaccurate estimate. While we can start again at Step 1 with the same threshold family, we will still not have a perfect estimator because changing $(\hR_1,\ldots,\hR_m)$ will affect the entire procedure and change $(R_1,\ldots,R_m)$ in turn. Nevertheless, we may obtain a better procedure if the new $\hR_i$ is a better estimator of the new $R_i$. We call this process {\em recursive refinement} of the estimator.

We will denote the original estimator as $\hR_i^{(1)}$, which leads to original calibration parameter $\hc_i^{(1)}$ and initial rejection set $\cR_+^{(1)}$. We define the  recursively refined estimator as
\begin{equation}
  \hR_i^{(2)}(X) = \left|\cR_+^{(1)}(X) \cup \{i\}\right|.
\end{equation}

We can then calibrate new thresholds $\hc_i^{(2)}$ solving \eqref{eq:def-hc} with respect to $\hR_i^{(2)}$, and proceed as before. In principle, we can repeat this refinement as many times as we want, defining
$\hR_i^{(k)}(X) = \left|\cR_+^{(k-1)}(X) \cup \{i\}\right|$ for all $k > 1$, but in most problems of moderate size the computational cost is prohibitive for $k > 2$, for reasons we explain in Section~\ref{sec:computation}.

Recursive refinement is especially useful when we begin with a very conservative estimator, as we do when we use the $\dBY$ procedure. If we use the effective $\BH$ threshold with the $\dBH_\gamma(\alpha)$ estimator, we call the resulting procedure the $\dBH_\gamma^2(\alpha)$ procedure, or the $\dBY^2(\alpha)$ procedure if $\gamma = 1/L_m$. When the baseline procedure is safe, recursive refinement always yields another safe procedure that is uniformly more powerful:

\begin{theorem}\label{thm:recursive}
    Assume $\hc_1^{(k)},\ldots,\hc_m^{(k)}$ are maximal for all $i$ and $k$. If $\cR^{(1)}$ is safe, then for every $k \geq 1$, $\cR^{(k+1)}$ is safe and uniformly more powerful than $\cR^{(k)}$.
\end{theorem}

\begin{proof}
  It is sufficient to prove the result for $k = 1$, since for any $k > 1$, $\cR^{(k+1)}$ bears the same relationship to $\cR^{(k)}$ that $\cR^{(2)}$ bears to $\cR^{(1)}$.

The safeness of $\hR_i^{(1)}$ implies that, almost surely, $\cR_+^{(1)} = \cR^{(1)}$ and
\begin{equation}
  R_i^{(1)} = R_+^{(1)} \geq \hR_i^{(1)}, \text{ for all } i \in \cR^{(1)}.
\end{equation}

As a result, we have
\begin{equation}
  \hR_i^{(2)} = |\cR^{(1)} \cup \{i\}| = R_i^{(1)} \geq \hR_i^{(1)}, \text{ for all } i \in \cR^{(1)}, \text{ almost surely.}
\end{equation}

Because $\hR_i^{(2)} \geq \hR_i^{(1)}$, the integrand in \eqref{eq:def-hc} for $k=2$ is almost surely smaller than the integrand for $k=1$, so $\hc_i^{(2)} \geq \hc_i^{(1)}$ almost surely. If all the calibration parameters increase, then $\cR_+^{(2)} \supseteq \cR_+^{(1)}$ as well, so for all $i \in \cR^{(2)}$,
\begin{equation}
  \cR_+^{(2)} \supseteq (\cR^{(1)} \cup \{i\}) \;\Longrightarrow\; R_+^{(2)} \geq \hR_i^{(2)},
\end{equation}
completing the proof.
\end{proof}

\subsection{One-sided testing in exponential family models}\label{sec:one-sided}

In this section we consider how to test the one-sided hypotheses $H_i:\; \theta_i \leq 0$ in exponential family models of the form \eqref{eq:exp-fam}. We assume throughout that $0 \in \Theta\cup \RR^d$, and that the tests are right-tailed; otherwise we can reparameterize the family (possibly with different reparameterizations for each $i$).

If we were testing a single hypothesis, the UMPU test for $H_i$ would reject for large values of $T_i$, conditional on the value of $S_i = T_{-i}$; these include the $z$-, $t$-, and Fisher exact tests we discuss in Section~\ref{sec:examples} and Appendix~\ref{app:examples} \citep{lehmann2005testing}. That is, the conditional test $p$-value is given by 
\[
p_i = 1 - \PP_{\theta_i=0}(T_i < t \mid S_i) = \lim_{t \rightarrow T_i^-} 1 - F_{i, 0}(t \mid S_i),
\]
where the conditional distribution $F_{i,\theta_i}(t \mid S_i) = \PP_{\theta_i}(T_i \leq t \mid S_i)$ depends only on $\theta_i$.

Because the distribution of $T_i$ is stochastically increasing in $\theta_i$, only the boundary case $\theta_i = 0$ is relevant for calculating the $p$-value, but we cannot necessarily restrict our attention to the boundary when we calibrate $\hc_i$, unless $\theta_i=0$ is least favorable in the sense of \eqref{eq:least-fav}. The next result gives a sufficient condition for least favorability:

\begin{proposition}\label{prop:worst-case-boundary}
  Consider testing $H_i:\; \theta_i \leq 0$ for $i = 1,\ldots,m \leq d$ in an exponential family model of the form \eqref{eq:exp-fam} with $0 \in \Theta^{\mathrm{o}} \subseteq \RR^d$. Assume for all $i$ that $p_i$ is given by the standard one-sided UMPU test, and that we have for some $\alpha'$ which is not necessarily the target FDR level, almost surely,
  \begin{enumerate}[(i)]
    \item $\tau_i(c; X) \leq \alpha'$, and
    \item Under $\theta_i=0$, the conditional upper-$\alpha'$ quantile of $T_i$ is above its conditional mean:
      \[
      F_{i,0}(t \mid S_i) < 1-\alpha', \quad \text{for all} \;\; t < \mu_{i,0}(S_i),
      \]
      where $\mu_{i,\theta_i}(S_i) = \EE_{\theta_i}\left[T_i \mid S_i\right]$.
  \end{enumerate}
  Then, $\theta_i = 0$ is least favorable for $H_i$, for purposes of calibrating $\hc_i$.
\end{proposition}

\begin{proof}
  Let $f_{\theta_i}(t \mid s)$ denote the conditional density of $T_i(X)$ given $S_i(X) = s$, which is a one-parameter exponential family with respect to some base measure $P_0$:
  \[
  f_{\theta_i}(t \mid s) = e^{\theta_i t - B(\theta_i \mid s)} f_0(t \mid s),
  \]
  and let
  \[
  h(t,s) = \EE\left[\,\frac{1\{p_i \leq \tau_i(c)\}}{\hR_i} \mid S_i =s, T_i = t\,\right],
  \]
  which does not depend on $\theta$ by sufficiency of $(S_i, T_i)$. 
  
  Finally, define
  \begin{align*}
    g_{i,\theta_i}(c \semic s)
    &= \EE_{\theta_i} \left[\,\frac{1\{p_i \leq \tau_i(c)\}}{\hR_i} \mid S_i = s\,\right]\\
    &= \int h(t, s) f_{\theta_i}(t \mid s)\,dP_0(t).
  \end{align*}
  
  Because $h(t,s)$ takes values in $[0,1]$, by dominated convergence we can differentiate under the integral sign, giving
  \[
  \frac{\partial}{\partial \theta_i} g_{i,\theta_i}(c\semic s) = \int h(t, s) (t - \mu_{i,\theta_i}(s)) f_{\theta_i}(t \mid s)\,dP_0(t),
  \]
  where $\mu_{i,\theta_i}(s) = \frac{d}{d \theta_i} B(\theta_i \mid s) = \EE_{\theta_i}[ T_i \mid S_i = s]$ is increasing in $\theta_i$.

  If $\theta_i \leq 0$ then our assumption (ii) ensures that $p_i > \alpha' \geq \tau_i(c)$, and therefore $h(T_i \mid S_i)=0$, whenever $T_i < \mu_{i,\theta_i}(S_i) \leq \mu_{i,0}(S_i)$, so the integrand is non-negative for all $t$ and $s$. As a result $g_{i,\theta_i}(c \semic s)$ is non-decreasing in $\theta_i$ for any $\theta_i \leq 0$, so it attains its maximum at 0.
\end{proof}

Although there is no universal cap on the $\tau^{\BH}$ threshold, in practice it very rarely exceeds $\alpha$, and we can choose to modify it by capping it manually at some $\alpha'$. In our implementation of $\dBH$, we cap $c$ at $2\alpha$, effectively capping $\tau^{\BH}$ at $\alpha'=2\alpha$, as discussed in Appendix \ref{app:qcap}. In all examples discussed in Section \ref{sec:examples}, $T_i$ is symmetrically distributed given $S_i$ with $\mu_{i, 0}(S_i) = 0$, so the assumptions of Proposition~\ref{prop:worst-case-boundary} hold for all $\alpha \leq 0.25$.

\subsection{Two-sided testing and directional error control}\label{sec:two-sided}

We say a hypothesis is {\em two-sided} if it can be written as $H_i:\; \theta_i(P) = 0$, for some parameter $\theta_i$ mapping $\cP$ to $\RR$, where the range includes both positive and negative values. Because a two-sided hypothesis frequently represents a ``measure-zero'' set in the model $\cP$, rejecting $H_i$ is more meaningful when we can also draw an inference about the sign of $\theta_i$. If we write $H_i$ as the intersection of the one-sided hypotheses $H_i^\leq:\; \theta_i(P) \leq 0$ and $H_i^\geq:\; \theta_i(P) \geq 0$, a {\em directional inference} is one that rejects exactly one of $H_i^{\leq}$ and $H_i^\geq$ along with $H_i$.

For multiple testing of two-sided hypotheses with directional inferences, let $\cR^+$ denote the set of indices for which we declare $\theta_i > 0$ (reject $H_i^{\leq}$) and $\cR^-$ the set for which we declare $\theta_i < 0$ (reject $H_i^{\geq}$), with $\cR$ the disjoint union of both sets. Let $p_i^+$ and $p_i^-$ denote $p$-values for each of the two tests, which we assume are conditionally valid:
\begin{equation}\label{eq:two-sided-p}
\sup_{P \in H_i^{\leq}} \PP_P(p_i^+ \leq \alpha \mid S_i) \leq \alpha \quad\text{and}\;\; \sup_{P \in H_i^{\geq}}\PP_P(p_i^- \leq \alpha \mid S_i) \leq \alpha, \quad \forall \alpha\in[0,1].
\end{equation}

We assume the two one-sided tests are based on a common test statistic $T_i$, with $H_i^{\leq}$ rejected when $T_i$ is large and $H_i^{\geq}$ rejected when $T_i$ is small, where the critical thresholds possibly depend on $S_i$. For the sake of simplicity we also assume the two-sided test is {\em equal-tailed}, in the sense that $p_i = 2\min\{p_i^+, p_i^-\}$.

In general, multiple testing procedures that are valid for two-sided hypotheses do not necessarily justify directional conclusions even if the constituent single hypothesis tests do \citep{shaffer1980control, finner1999stepwise}. Testing two-sided hypotheses with directional inferences creates more opportunities to make errors: defining $\cH_i^{\leq} = \{i:\; \theta_i \leq 0\}$ and likewise $\cH_i^\geq = \{i:\; \theta_i \geq 0\}$, the number of directional errors is $V^{\dir} = V^+ + V^-$, where
\[
V^+ = \left|\cH^{\leq} \cap \cR^+\right|, \quad \text{ and } \;\;V^- = \left|\cH^{\geq} \cap \cR^-\right|.
\]
The {\em directional FDP} is defined as $\FDP^\dir = V^\dir / (R \vee 1)$, and its expectation $\FDR^\dir$ is the {\em directional FDR.} The next result gives a natural sufficient condition guaranteeing that our method with directional inferences controls the directional FDR:

\begin{lemma}\label{lem:dirFDR}
  Assume $p_i^+$ and $p_i^-$ are valid in the sense of \eqref{eq:two-sided-p}, and that the assumptions of Theorem~\ref{thm:fdr-control} are satisfied for $S_i, \tau_i, \hc_{i,t}$ and $\hR_i$, for $i=1,\ldots,m$. Define
  \[
  g_{i,\theta}^a(c \semic S_i) \;=\; \sup_{P: \,\theta_i(P) = \theta} \EE_P\left[\,\frac{1\{p_i^a \leq \tau_i(c)/2\}}{\hR_i} \mid S_i\,\right], \quad
  \text{for} \;\; a = +, -.
  \]
  If $g_{i,\theta}^+(c \semic S_i)$ is almost surely non-decreasing in $\theta$ for $\theta \leq 0$, and $g_{i,\theta}^-$ is non-increasing in $\theta$ for $\theta \geq 0$, then our three-step method controls the directional FDR.
\end{lemma}
\begin{proof}
Define $V_i^a = 1\{p_i^a \leq \tau_i(c)/2\}$, so that $V_i = 1\{p_i \leq \tau_i(c)\} = V_i^+ + V_i^-$.
Then
\begin{align*}
  \FDR^\dir &= \sum_{i \in \cH_0^\leq} \EE\left[\frac{V_i^+}{R \vee 1}\right]
              + \sum_{i \in \cH_0^\geq} \EE\left[\frac{V_i^-}{R \vee 1}\right]\\
            &= \sum_{i:\; \theta_i = 0} \EE\left[\frac{V_i}{R \vee 1}\right] 
              + \sum_{i:\; \theta_i < 0} \EE\left[\frac{V_i^+}{R \vee 1}\right]
              + \sum_{i:\; \theta_i > 0} \EE\left[\frac{V_i^-}{R \vee 1}\right].
\end{align*}
Each term in the first sum is no larger than $\alpha/m$, by the argument in Theorem~\ref{thm:fdr-control}. For a generic term in the second sum, we have almost surely
\[
g_{i,\theta_i}^+(c \semic S_i) \;\leq\; g_{i,0}^+(c \semic S_i) \;\leq\; g_i^*(c \semic S_i).
\]
Then we can likewise repeat the argument of Theorem~\ref{thm:fdr-control} to obtain
\begin{align*}
  \EE\left[\frac{V_i^+}{R \vee 1}\right] 
  &\leq \lim_{t \to \infty} \EE \left[ \EE\left[ \frac{1\left\{p_i^+ \leq \tau_i(\hc_{i,t})/2\right\}}{\hR_i} \mid S_i
    \right] \right]\\
  &\leq \lim_{t \to \infty} \EE \left[\, g_{i,\theta_i}^+(\hc_{i,t} \semic S_i) \,\right]\\
  &\leq \lim_{t \to \infty} \EE \left[\, g_i^*(\hc_{i,t} \semic S_i) \,\right],
\end{align*}
which is no larger than $\alpha/m$ by assumption. Likewise, each term in the third sum is no larger than $\alpha/m$, and there are $m$ total terms among the three sums.
\end{proof}

As an immediate consequence of Lemma~\ref{lem:dirFDR} and Proposition~\ref{prop:worst-case-boundary}, we see that we can draw directional conclusions for two-sided multiple testing in exponential family models:

\begin{corollary}
  Consider testing $H_i:\; \theta_i = 0$ for $i = 1,\ldots,m \leq d$ in an exponential family model of the form \eqref{eq:exp-fam} with $0 \in \Theta^{\mathrm{o}} \subseteq \RR^d$. Assume for all $i$ that $p_i = 2\min\{p_i^+,p_i^-\}$ where $p_i^\pm$ are given by the standard one-sided UMPU tests, and that we have, almost surely,
  \begin{enumerate}[(i)]
    \item $\tau_i(c; X) \leq \alpha$, and
    \item Under $\theta_i=0$, the conditional mean is between the conditional lower- and upper-$\alpha/2$ quantiles:
      \begin{align*}
        F_{i,0}(t \mid S_i) &< 1-\alpha/2, \quad \text{for all} \;\; t < \mu_{i,0}(S_i), \quad \text{ and }\\
        F_{i,0}(t \mid S_i) &> \alpha/2, \qquad \text{for all} \;\; t > \mu_{i,0}(S_i),
      \end{align*}
      where $\mu_{i,\theta_i}(S_i) = \EE_{\theta_i}\left[T_i \mid S_i\right]$.
  \end{enumerate}
  Then, our three-step method controls the directional FDR.
\end{corollary}

\begin{proof}
  Applying Proposition~\ref{prop:worst-case-boundary} to the modified threshold $\tilde{\tau}_i(c) = \tau_i(c)/2$, we have the hypotheses of Lemma~\ref{lem:dirFDR}.  
\end{proof}



\section{Examples}\label{sec:examples}

In this section we give additional details about several parametric examples arising from the multivariate Gaussian family. Appendix~\ref{app:examples} discusses three further parametric examples --- edge testing in Gaussian graphical models, post-selection $z$- and $t$-testing, and multiple comparisons to control for a one-way layout with binary outcomes --- as well as a nonparametric example, multiple comparisons to control with in a one-way layout with generic responses and $p$-values arising from permutation tests.

Let $Z \sim N_d(\mu, \Sigma)$ with $\Sigma \succ 0$, an exponential family model with density
\begin{align}
  f_{\mu, \Sigma}(z) 
  &= \frac{1}{(2\pi)^{n/2}|\Sigma|} \,\exp\left\{ - \frac{1}{2} (z-\mu)'\Sigma^{-1} (z-\mu)\right\} \\
  \label{eq:gauss-exp-fam}
  &= \frac{1}{(2\pi)^{d/2}|\Sigma|} \,\exp\left\{ \mu'\Sigma^{-1}z - \frac{1}{2} z'\Sigma^{-1}z - \frac{1}{2} \mu'\Sigma^{-1}\mu\right\}.
\end{align}

Defining $A^i = (\Sigma^{-1})_{-i,-i}$, it will be useful to recall the formula 
\begin{equation}\label{eq:Ai}
(\Sigma^{-1})_{-i,i} = A^i \Sigma_{-i,i}\Sigma_{i,i}^{-1}.
\end{equation}

\subsection{Multivariate $z$-statistics}\label{sec:zstats}

First assume $\Sigma$ is known, with $\Sigma_{i,i} = 1$, so that each $Z_i$ is a $z$-statistic for testing $H_i:\; \mu_i = 0$ or $H_i:\;\mu_i \leq 0$, for $i=1,\ldots, m \leq d$, and $p_i$ is the resulting one- or two-sided $p$-value. We can rewrite \eqref{eq:gauss-exp-fam} as a full-rank $d$-parameter exponential family:
\begin{align}
f_{\mu, \Sigma}(z) 
  &= \frac{1}{(2\pi)^{d/2}|\Sigma|} \,\exp\left\{ \mu_i(\Sigma^{-1}z)_i + \mu_{-i}'(\Sigma^{-1}z)_{-i} - \frac{1}{2} z'\Sigma^{-1}z - \frac{1}{2} \mu'\Sigma^{-1}\mu\right\}\\
  \label{eq:mvgauss-exp-fam-i}
  &= \frac{1}{(2\pi)^{d/2}|\Sigma|} \,\exp\left\{ \mu_i(\Sigma^{-1}z)_i + (A^i\mu_{-i})'(z_{-i} - \Sigma_{-i,i}z_i) - \frac{1}{2} z'\Sigma^{-1}z - \frac{1}{2} \mu'\Sigma^{-1}\mu\right\}.
\end{align}

To test $H_i$, the general proposal in Section~\ref{sec:find-s} leads to the conditioning statistic $S_i = Z_{-i} - \Sigma_{-i,i}Z_i$. $S_i$ is independent of $Z_i$ and $p_i$ since $\text{Cov}(Z_i, S_i) = 0.$

To carry out the $\dBH_\gamma(\alpha)$ procedure, we must evaluate for each $i$ whether
\begin{equation}\label{eq:mvgauss-int}
\EE_0\left[ \frac{1\{q_i \leq c\}}{R_i^{\BH(\gamma \alpha)}} \mid S_i \right] \leq \frac{\alpha}{m},
\end{equation}
plugging in $c = q_i(Z)$, the observed BH $q$-value. Because $Z_{-i} = S_i + \Sigma_{-i,i}Z_i$, it is straightforward to evaluate the expectation in \eqref{eq:mvgauss-int} by integrating over the set $\{(z, S_i + \Sigma_{i,-i} z): \; z \in \RR\}$. For one-sided $p$-values $p_i = 1 - \Phi(Z_i)$, $p_{-i}$ is CPRD on $p_i$ if and only if $\Sigma_{ij} \geq 0$ for all $j \leq m$, mirroring the condition for marginal PRD in \citet{benjamini2001control}.

\subsection{Multivariate $t$-statistics}\label{sec:tstats}

A slightly harder case is to assume that $\Sigma = \sigma^2 \Psi$ where $\Psi \succ 0$ is known but $\sigma^2>0$ is unknown. Assume we are still testing $H_i:\; \mu_i = 0$ for $i = 1,\ldots, m \leq d$, with an additional independent vector $W \sim N_{n-d}(0, \sigma^2 I_{n-d})$ available for estimating $\sigma^2$. Then the usual $t$-statistic for testing $H_i$ is
\[
T_i = \frac{Z_i}{\sqrt{\Psi_{i,i} \hat\sigma^2}} \stackrel{H_i}{\sim} t_{n-d}, \qquad \text{ where } \quad(n-d)\hat\sigma^2 = \|W\|^2 \sim \sigma^2\chi_{n-d}^2.
\]

Extending the density in \eqref{eq:mvgauss-exp-fam-i} to include $W$, we obtain the $d+1$-parameter exponential family form
\begin{align*}
  f_{\mu, \Psi, \sigma^2}(z, w) = \frac{1}{(2\pi\sigma^2)^{n/2}|\Psi|} \,\exp\Bigg\{ 
  \frac{\mu_i}{\sigma^2}(\Psi^{-1}z)_i &+ (A^i\mu_{-i})'(z_{-i}-\Psi_{-i,i}\Psi_{i,i}^{-1}z_i)) \\
  &- \frac{1}{2\sigma^2} (w'w + z'\Psi^{-1}z) - \frac{\mu'\Psi^{-1}\mu}{2\sigma^2} \Bigg\},
\end{align*}

Letting $U_i = Z_{-i} - \Psi_{-i,i}\Psi_{i,i}^{-1}Z_i$, the general proposal in Section~\ref{sec:find-s} leads to the conditioning statistic $\left(U_i, \;\|W\|^2 + Z'\Psi^{-1}Z\right)$, or equivalently
\[
S_i(Z,W) = (U_i, V_i), \quad \text{ where } V_i = \|W\|^2 + \frac{Z_i^2}{\Psi_{i,i}}.
\]
since
\[
Z'\Psi^{-1}Z = \frac{Z_i^2}{\Psi_{i,i}} + U_i'(\Psi_{-i,-i} - \Psi_{-i,i}\Psi_{i,i}^{-1}\Psi_{i,-i})^{-1} U_i.
\]

$T_i$, $U_i$, and $V_i$ are mutually independent under $H_i$, and we can reconstruct the other $t$-statistics from their values:
\[
  V_i = \|W\|^2 + \frac{Z_i^2}{\Psi_{ii}}
      = \|W\|^2 \left(1 + \frac{T_i^2}{n-d}\right),
\]
and
\[
  T_j \;\;=\;\; \frac{Z_j}{\sqrt{\Psi_{j,j}\hat\sigma^2}} 
      \;\;=\;\; \frac{U_{ij}}{\sqrt{\Psi_{j,j}\hat\sigma^2}}  + \frac{\Psi_{j,i} Z_i}{\Psi_{i,i}\sqrt{\Psi_{j,j}\hat\sigma^2}} 
      \;\;=\;\; U_{ij}\sqrt{\frac{n-d+T_i^2}{\Psi_{j,j}V_i}} + \frac{\Psi_{j,i}}{\Psi_{i,i}} T_i
\]
Hence, just as in the previous section we can evaluate $g_i^*(q_i)$ by integrating over reconstructed $t$-statistics: 
\[
\left\{\left(t, \;\;\sqrt{\frac{n-d+t^2}{V_i}}\, \diag(\Psi_{-i,-i})^{-1/2}\,U_i \;+\; \frac{\Psi_{-i,i}}{\Psi_{i,i}}\, t\right): t \in \cR\right\},
\]
where $t \sim t_{n-d}$. For two-sided testing with uncorrelated test statistics (diagonal $\Psi$), we see that $T_j^2$ is non-decreasing in $t^2$ for all values of $t$, so the test statistics are CPRDS. For one-sided right-tailed testing with uncorrelated test statistics, in light of Remark~\ref{rem:cprd} we only need to consider values with $t > 0$ and the conditional probability of the event $R_{i}^{\BH(\alpha)}\le r$, which only depends on p-values that are below $\alpha < 0.5$ and thus positive $T_j$ values. Then, $T_j$ is non-decreasing in $t$ and positive provided that $\Psi_{j,i} = 0$ and $U_{ij} = Z_j > 0$. As a result, even though the one-sided test statistics are not CPRDS, the second claim of Theorem~\ref{thm:safe} holds. When the test statistics are correlated, even this relaxed version of the CPRDS condition does not hold for either one- or two-sided testing.

\subsection{Testing coefficients in linear models}\label{sec:linearmodels}

A third example is the Gaussian linear model in which we observe covariates $x_i \in \RR^d$, with response
\[
Y_i \sim x_i'\beta + \epsilon_i, \quad \text{ for } \epsilon_i \simiid N(0,\sigma^2), \quad i = 1,\ldots,n,
\]
with $\beta \in \RR^d$ and $\sigma^2 > 0$ unknown. Typically we wish to test $H_j:\; \beta_j = 0$ (or analogous one-sided hypotheses), for $j=1,\ldots, m \leq d$.

If the design matrix $\bX = (x_1, \ldots, x_n)' \in \RR^{d \times n}$ has full column rank then a sufficiency reduction boils the data set down to ordinary least squares coefficients and residual sum of squares:
\[
\hat\beta = (\bX'\bX)^{-1}\bX'Y \sim N_d(\beta,\; \sigma^2 (\bX'\bX)^{-1}), \quad \text{ and } \quad \|Y - \bX\hat\beta\|^2 \sim \sigma^2 \chi_{n-d}^2.
\]

Thus we can reduce the linear model case to multivariate $t$-testing problem we discuss above, with $Z = \hat\beta$, $\Psi = (\bX'\bX)^{-1}$ and $\|W\|^2 = \|Y - \bX\hat\beta\|^2 = RSS$, the residual sum of squares.

Defining $\bX_j$ as the $j$th column of $\bX$, and $\bX_{-j}$ as the remaining columns, let 
\[
\bb^j(\bX) \;=\; (\bX_{-j}'\bX_{-j})^{-1}\bX_{-j}'\bX_j \;=\; \Psi_{-j,j}\Psi_{j,j}^{-1},\quad \text{and} \quad \br^j(\bX) \;=\; \bX_j - \bX_{-j}\bb^j
\]
denote the coefficients and residuals from an OLS regression of $\bX_i$ on $\bX_{-i}$. Then we can write
\[ 
U_j  \;=\; \hat\beta_{-j} -  \hat\beta_j \bb^j, \quad \text{and} \quad V_j \;=\; RSS + \frac{\hat\beta_j^2}{((\bX'\bX)^{-1})_{j,j}} \;=\; RSS + \hat\beta_j^2 \,\|\br^j\|^2.
\]
Applying the logic of the previous section, the one- and two-sided $t$-statistics are CPRDS if $\bX'\bX$ is diagonal. In many cases, only a subset of regression coefficients are of interest; for example we would rarely test for an intercept term. Assuming the coefficients are defined so that only the first $m < d$ are of interest, we use $\hat{\beta}_{[m]} \sim N_{m}(\beta_{[m]},\; \sigma^2 [(\bX'\bX)^{-1}]_{[m], [m]})$ and $\|Y - \bX\hat\beta\|^2 \sim \sigma^2 \chi_{n-d}^2$, which has the same structure as testing the full set of coefficients.



\section{Computation}\label{sec:computation}

\subsection{An exact homotopy algorithm for $\dBH_\gamma(\alpha)$}\label{sec:homotopy}
In this section we discuss an exact homotopy algorithm for $\dBH_\gamma(\alpha)$. It can be easily generalized to $\dSU_{\gamma, \Delta}(\alpha)$ at the cost of more complex notation; see Appendix \ref{app:computation} for details. Recalling the definition of $\dBH_\gamma(\alpha)$ in Section \ref{sec:dbh-procedure} and that of $\tau^{\BH}(c; X)$ in \eqref{eq:q-value-threshold}, $g_i^*$ can be equivalently formulated as
\[g_i^*(c \semic S_i) = \sup_{P\in H_i}\EE\left[\frac{1\{p_{i}\le c R^{\BH(c)}(X) / m\}}{R^{\BH(\gamma \alpha)}(X)} \mid S_i\right].\]
For all examples discussed in Section \ref{sec:examples}, $H_i$ is conditionally simple and $p_i = \eta_i(T_i)$ for some univariate transformation $\eta_i$ of test statistic $T_i$, and there exists a bijective mapping $\xi_i$ from $(T_i, S_i)$ to $(T_1, T_2, \ldots, T_m)$. For the multivariate Gaussian case, $T_i = Z_i$, $\eta_i(t) = 1 - \Phi(t)$ for one-sided testing and $\eta_i(t) = 2(1 - \Phi(|t|))$ for two-sided testing, and $\xi_i(t; s) = (\xi_{i1}(t), \ldots, \xi_{im}(t))$ where $\xi_{ii}(t) = t$ and $\xi_{ij}(t) = s_j + \Sigma_{j, i}t$. Since $T_i$ is independent of $S_i$, $g_i^*(c \semic S_i)$ can be equivalently formulated as the following univariate integral:
\begin{equation}
  \label{eq:integral}
  g_i^*(c \semic S_i) = \int_{\RR} \frac{1\{\eta_i(t)\le c R_i^{(c)}(t) / m\}}{R_i^{(\gamma \alpha)}(t)}dP_i(t)
\end{equation}
where $P_i(t)$ is the marginal distribution of $T_i$ under $H_i$, and $R_i^{(c)}(t) = R^{\BH(c)}(\xi_i(t; S_i))$ denotes the number of rejections by $\BH(c)$ if the observed test statistics $\xi_i(T_i; S_i)$ are replaced by $\xi_i(t; S_i)$. As remarked at the end of Section \ref{sec:dbh-procedure}, if $q_i$ is the observed $q$-value then
\[\cR_{+} = \left\{i: g_i^*(q_i \semic S_i)\le \frac{\alpha}{m}\right\}.\]
Therefore, it is left to compute $R_i^{(c)}(t)$ with $c\in \{q_i, \gamma\alpha\}$.

The test statistics $(T_1, \ldots, T_m)$ as well as the induced $p$-values $(p_1, \ldots, p_m)$ from the new dataset $\xi_i(t; S_i)$ are both functions of $t$: $T_j(t) = \xi_{ij}(t)$ and $p_j(t) = \eta_j(T_j(t)) = \eta_j(\xi_{ij}(t))$. Since $R_i^{(c)}(t)$ is an integer-valued function, it must be piecewise constant. We call each point at which $R_i^{(c)}(t)$ changes the value a \emph{knot}. Then a point is a knot only if $p_j(t)$ crosses $c r / m$ for some $j, r\in [m]$. 

For all examples discussed in Section \ref{sec:examples}, the domain of $T_i$ is $\RR$, and $\eta_i, \xi_{ij}$ are differentiable. Then the set of potential knots of $R_i^{(c)}(t)$ is
\begin{equation}
  \label{eq:cK}
  \cK_i = \bigcup_{j, r\in [m]}\cK_{i, j, r}, \quad \mbox{where }\cK_{i, j, r} = \left\{t: \eta_j(\xi_{ij}(t)) = \frac{cr}{m}\right\},
\end{equation}
For the one-sided multivariate Gaussian testing problem, 
\[\cK_{i, i, r} = \left\{t: 1 - \Phi(t) = \frac{cr}{m}\right\} = \left\{\Phi^{-1}\lb 1 - \frac{cr}{m}\rb\right\},\]
and for each $j\not = i$,
\begin{align*}
  \cK_{i, j, r} &= \left\{t: 1 - \Phi(S_{ij} + \Sigma_{j,i}t) = \frac{cr}{m}\right\}\\
  & = \left\{
    \begin{array}{ll}
      \left\{(\Phi^{-1}\lb 1 - \frac{cr}{m}\rb - S_{ij}) / \Sigma_{j, i}\right\} & (\mbox{if }\Sigma_{j, i}\not = 0)\\
      \emptyset & (\mbox{if }\Sigma_{j, i} = 0 \mbox{ and }1 - \Phi(S_{ij}) \not = \frac{cr}{m})\\
      \RR & (\mbox{otherwise})
    \end{array}
    \right..
\end{align*}
Note that $S_{ij}$ has an absolutely continuous density, $\PP(\cK_{i, j, r} = \RR\mbox{ for any }j, r) = 0$. Thus, with probability $1$,
\[\cK_i = \bigcup_{r\in [m]}\bigcup_{j: \Sigma_{j, i}\not = 0}\cK_{i, j, r}\]
where $\cK_{i, j, r}$ is a singleton. Similarly, for the two-sided multivariate Gaussian testing problem, it is easy to verify that $\cK_i$ has the same form as above except that each $\cK_{i, j, r}$ has two elements. For multivariate t-statistics, $\cK_{i, j, r}$ has a more complicated structure though it can still be computed efficiently; see Appendix \ref{app:knots_t} for details.

Let $t_1 < t_2 < \ldots < t_N$ denote the elements of $\cK$ with $(j_1, r_1), (j_2, r_2), \ldots, (j_N, r_N)$ denoting the indices such that $t_k \in \cK_{i, j_k, r_k}$. Let 
\[B^{(c)}_\ell(t) = \bigg|\left\{\ell: p_\ell(t) \le \frac{c\ell}{m}\right\}\bigg| - \ell, \quad (\ell = 0, \ldots, m).\]
By definition of $\BH(c)$,
\begin{equation}
  \label{eq:Rit_Bt}
  R_i^{(c)}(t) = \max\left\{\ell: B^{(c)}_\ell(t) = 0\right\}.
\end{equation}
As $t$ moves from $t_{k-1}$ to $t_k$, $B_{\ell}^{(c)}(t)$ remain the same for $\ell \not = r_k$ while $B_{r_k}^{(c)}(t_k)$ is incremented by $+1$ or $-1$, depending on whether $p_{j_k}(t)$ is increasing or decreasing at $t_k$. For all examples discussed in Section \ref{sec:examples}, $p_j(t)$ is differentiable, and thus
\begin{equation}
  \label{eq:update_B}
  B_{r_k}^{(c)}(t_k) - B_{r_k}^{(c)}(t_{k-1}) = \sign\lb p'_{j_k}(t_k)\rb.
\end{equation}
For one-sided multivariate Gaussian testing problems, $\eta_{j_k}'(z) = -\phi(z) < 0$ for any $z\in \RR$ and $\xi_{ij_k}'(t) = \Sigma_{j_k, i}$. As a result, 
\[B_{r_k}^{(c)}(t_k) - B_{r_k}^{(c)}(t_{k-1}) = \sign(\Sigma_{j_k, i}).\]

This motivates a homotopy algorithm to calculate $B_\ell^{(c)}(t)$ sequentially based on \eqref{eq:update_B} and $R_i^{(c)}(t)$ based on \eqref{eq:Rit_Bt}. It is not hard to see he computational cost of the homotopy algorithm for a single hypothesis $H_i$ is $\cO(|\cK_i|)$. Therefore, the total cost for $\dBH_\gamma(\alpha)$ is of order
\begin{equation}
  \label{eq:sumijr}
  \sum_{i}|\cK_i| \le \sum_{i}\sum_{j}\sum_{r}|\cK_{i, j, r}|.
\end{equation}
Naively, it requires $\cO(m^3)$ computation since there are $m$ summands for $i$ and $j$, corresponding to the hypotheses, and $m$ summands for $r$, corresponding to the thresholds. Nonetheless, we can significantly reduce the size of each sum above by using a step-up method similar to $\BH$, but with sparse increments so that there are only $\log_2 m$ distinct threshold values:
\begin{equation}
  \label{eq:sparse_dSU}
  \Delta_{\alpha}(r) = \frac{\alpha \beta(r)}{m}, \quad \text{with } \beta(r) = 2^{\lfloor \log_2 r\rfloor}.
\end{equation}
We define the {\em sparse $\dBH_\gamma$} ($\sdBH_\gamma$) method as the $\dSU_{\gamma,\Delta}$ method with thresholds given in \eqref{eq:sparse_dSU}. In this case, even the naive method only requires $\cO(m^2 \log m)$ computation. 

With all the tricks that are detailed in Appendix \ref{app:computation}, the number of summands in all of the three sums can be further drastically reduced. For multivariate Gaussian testing problems, the number of $i$ has the same order of $R_\BH(2\alpha)$, the number of $j$ given $i$ has the same order as the range of non-negligible correlation, and the number of $r$ given $i$ and $j$ may be far lower than the total number of thresholds when $\Sigma_{j, i}$ is small. For the case with short-ranged dependence like in the autoregressive (AR) process, and a bounded number of signals, $R_{\BH}(2\alpha)$ is bounded with high probability and thus the computation cost of $\dBH_\gamma(\alpha)$ is at most $\cO(m)$. Thus, although the worst-case performance is poor, the cost is highly instance-specific and we find that the algorithm is reasonably fast in many cases.

\subsection{An approximate numerical integration for $\dBH^2_\gamma(\alpha)$}\label{sec:numerical_integration}

Similar to \eqref{eq:integral}, the conditional expectation $g_i^*(c\mid S_i)$ in $\dBH^2_\gamma(\alpha)$ can be formulated as
\begin{equation}
  \label{eq:integral2}
  g_i^*(c \semic S_i) = \int_{\RR} \frac{1\{\eta_i(t)\le c R_i^{(c)}(t) / m\}}{\td{R}_i(t)}dP_i(t),
\end{equation}
where $\td{R}_i(t) = R^{\dBH_\gamma(\alpha)}(\xi_i(t; S_i))$ denotes the number of rejections by $\dBH_\gamma(\alpha)$ if the test statistics shift from $X = (T_i, S_i)$ to $\xi_i(t; S_i)$. Unlike $\dBH_\gamma(\alpha)$, the denominator $\td{R}_i(t)$ has a much more complicated structure and we do not have an efficient homotopy algorithm to calculate the whole path. In principle, Monte-Carlo integration can guarantee almost sure convergence as the number of random samples grows to infinity because the integrand is bounded. However, it introduces extra randomness to the procedure which is undesirable. For this reason, we approximate \eqref{eq:integral2} via a heuristic numerical integration method that has no guarantee in theory but works well in practice. For illustration, we focus on the one-sided multivariate Gaussian testing problem. 

The first step is to reduce \eqref{eq:integral2} to a finite-range integral. Since $R_i^{(c)}(t)\le m$, the integrand is $0$ whenever $\eta_i(t) > c$, or equivalently $t < t_\lo \triangleq \Phi^{-1}(1 - c)$. On the other hand, let $t_\hi \triangleq \Phi^{-1}(1 - \alpha \epsilon / m)$ for some $\epsilon < 1$, then the integral \eqref{eq:integral2} from $t_\hi$ to $\infty$ is upper bounded by $\alpha \epsilon / m$ because the integrand is bounded by $1$ and $P_i$ is the standard Gaussian distribution. As a consequence,
$\int_{t_\lo}^{t_\hi} r_i(t)dP_i(t)\le g_i^*(c \semic S_i) \le \int_{t_\lo}^{t_\hi} r_i(t)dP_i(t) + \frac{\alpha\epsilon}{m}$, where $r_i(t)$ denotes the integrand. If we take $\epsilon$ to be small, e.g. $\epsilon = 0.01$, then the approximation error of $\int_{t_\lo}^{t_\hi}r_i(t)dP_i(t)$ is negligible.

To compute $\int_{t_\lo}^{t_\hi}r_i(t)dP_i(t)$, a naive method is to approximate $r_i(t)$ by a piecewise constant function evaluated on an equi-spaced grid of $[t_\lo, t_\hi]$. However, it may be inefficient since $r_i(t) = 0$ whenever $\eta_i(t) > c R_i^{(c)}(t) / m$. A simple improved version is to find the region of $t$ in which $\eta_i(t) \le c R_i^{(c)}(t) / m$ using the exact homotopy algorithm for $R_i^{(c)}(t)$, and then discretize the resulting region to approximate $\int_{t_\lo}^{t_\hi}r_i(t)dP_i(t)$.

Naively, the computational cost is the product of the number of hypotheses $m$, the grid size and the cost of the homotopy algorithm to calculate a single $\td{R}_i(t)$. However, as with the homotopy algorithm, we discussed a few tricks in Appendix \ref{app:computation} that can drastically reduce the number of hypotheses for which the integral $g_i^*(q_i \semic S_i)$ needs to be computed. For instance, for a safe procedure, Theorem \ref{thm:recursive} guarantees that the hypotheses rejected by $\dBH_\gamma(\alpha)$ are also rejected by $\dBH^2_\gamma(\alpha)$, for which the computation of $g_i^*(q_i \semic S_i)$ can be avoided. With all tricks discussed in Appendix \ref{app:computation}, it is even possible that no integral needs to be evaluated, in which case the computational cost of $\dBH^2_\gamma(\alpha)$ reduces to that of $\dBH_\gamma(\alpha)$. In a nutshell, the computational cost of the above algorithm is highly instance-specific. 

\subsection{Illustration of scalability}
With all the tricks discussed in Appendix \ref{app:computation}, both algorithms are efficient and scalable to problems of reasonably large size. We illustrate it using a simple simulation study on multivariate z-statistics with an autoregressive covariance structure with autocorrelation $0.8$. We consider the number of hypotheses $m \in \{10^2, 10^3, 10^4, 10^5, 10^6\}$. For each size $m$, we consider both one- and two-sided tests, with either $10$ or $30$ non-nulls in the front of the list with mean $\sqrt{2\log m}$. In each case, we implement $\dBH_{1}(\alpha)$ and $\sdBH_{1}(\alpha)$ via the homotopy algorithm and implemented $\dBH_{1}^{2}(\alpha)$ and $\sdBH_{1}^{2}(\alpha)$ via the approximate numerical integration with $20$ knots and $40$ knots for one- and two-sided tests, respectively. Figure \ref{fig:runtime} presents the median running time over $100$ simulations of each method. For $\dBH_{1}(\alpha)$ and $\sdBH_{1}(\alpha)$, the homotopy algorithm can handle $10^6$ hypotheses in a few minutes, while for $\dBH_{1}^{2}(\alpha)$ and $\sdBH_{1}^{2}(\alpha)$, the approximate numerical integration can handle $10^5$ hypotheses in $20$ minutes. The results corroborate the intuition in previous sections that (a) using a sparse threshold collection yields faster algorithms, and (b) both algorithms are more efficient for sparser problems.

\begin{figure}[H]
  \centering
  \includegraphics[width = 0.8\textwidth]{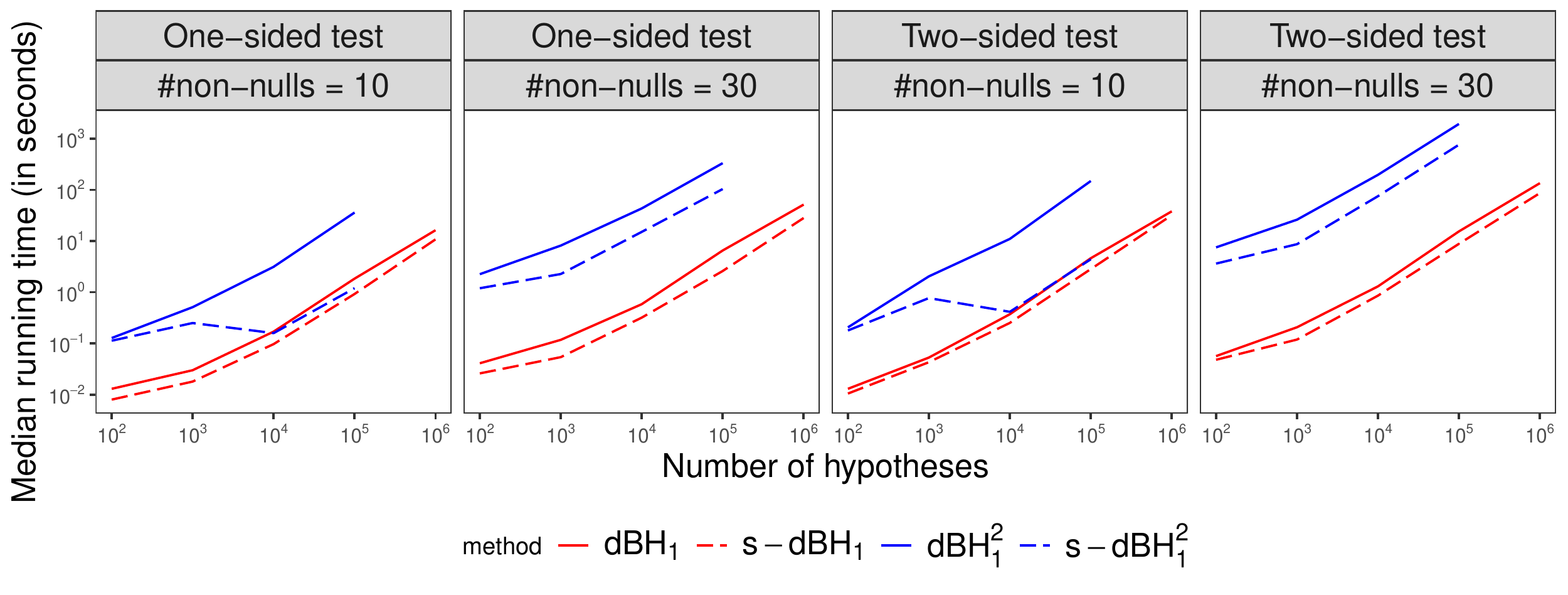}
  \caption{Median running time (in seconds) over $100$ independent simulations.}\label{fig:runtime}
\end{figure}



\section{Selected simulations}\label{sec:simulations}

In Appendix~\ref{app:full-simulations} we provide extensive simulations to compare the power of our approach with the power of several competing procedures including the $\BH(\alpha)$ and $\BH(\alpha/L_m)$ procedures as well as the fixed-X knockoffs \citep{barber15}, where appropriate. This section includes some highlights from our simulation results.

We start from a multivariate Gaussian case with $m = 1000$ and $z\sim N_m(\mu, \Sigma)$ where 
\[\mu_1 = \cdots = \mu_{10} = \mu^*, \quad \mu_{11} = \cdots = \mu_{1000} = 0,\]
We consider two types of covariance structures: (1) an autoregressive structure with $\Sigma_{ij} = (0.8)^{|i - j|}$; and (2) a block dependence structure with $\Sigma_{ii} = 1$ and $\Sigma_{ij} = 0.5 \cdot 1(\lceil i / 20\rceil = \lceil j / 20\rceil)$. We perform both one- and two-sided testing using $\BH(\alpha)$, $\dBH_\gamma(\alpha)$, $\dBH^2_\gamma(\alpha)$, $\BY(\alpha)$, $\dBY(\alpha)$ and $\dBY^2(\alpha)$. All these methods are implemented in the \texttt{R} package \texttt{dbh}. For one-sided testing, we choose $\gamma = 1$ because the $p$-values are CPRD, as shown in Section \ref{sec:zstats}. For two-sided testing, we choose $\gamma = 0.9$. We set the level $\alpha=0.05$ and tune the signal strength $\mu^*$ such that $\BH(\alpha)$ has approximately $30\%$ power in a separate Monte-Carlo simulation. We run each of the above 12 methods on $1000$ independent samples of z-values and estimate the \FDR ~and power, presented in Figure \ref{fig:mvgauss}.
\begin{figure}[H]
  \centering
  \begin{subfigure}[t]{0.48\textwidth}
    \includegraphics[width = 1\textwidth]{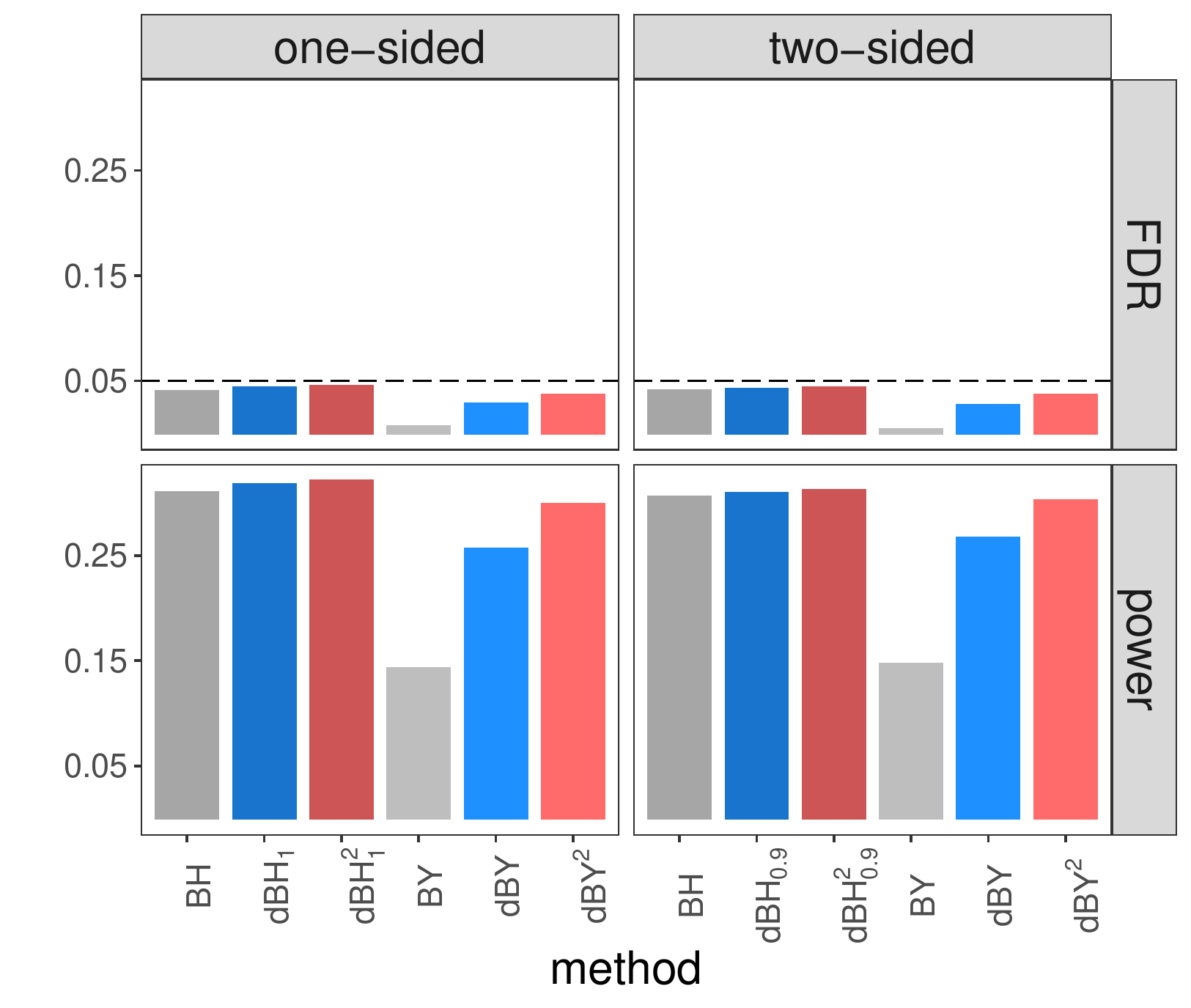}
    \caption{AR$(0.8)$ process}
  \end{subfigure}
  \begin{subfigure}[t]{0.48\textwidth}
    \includegraphics[width = 1\textwidth]{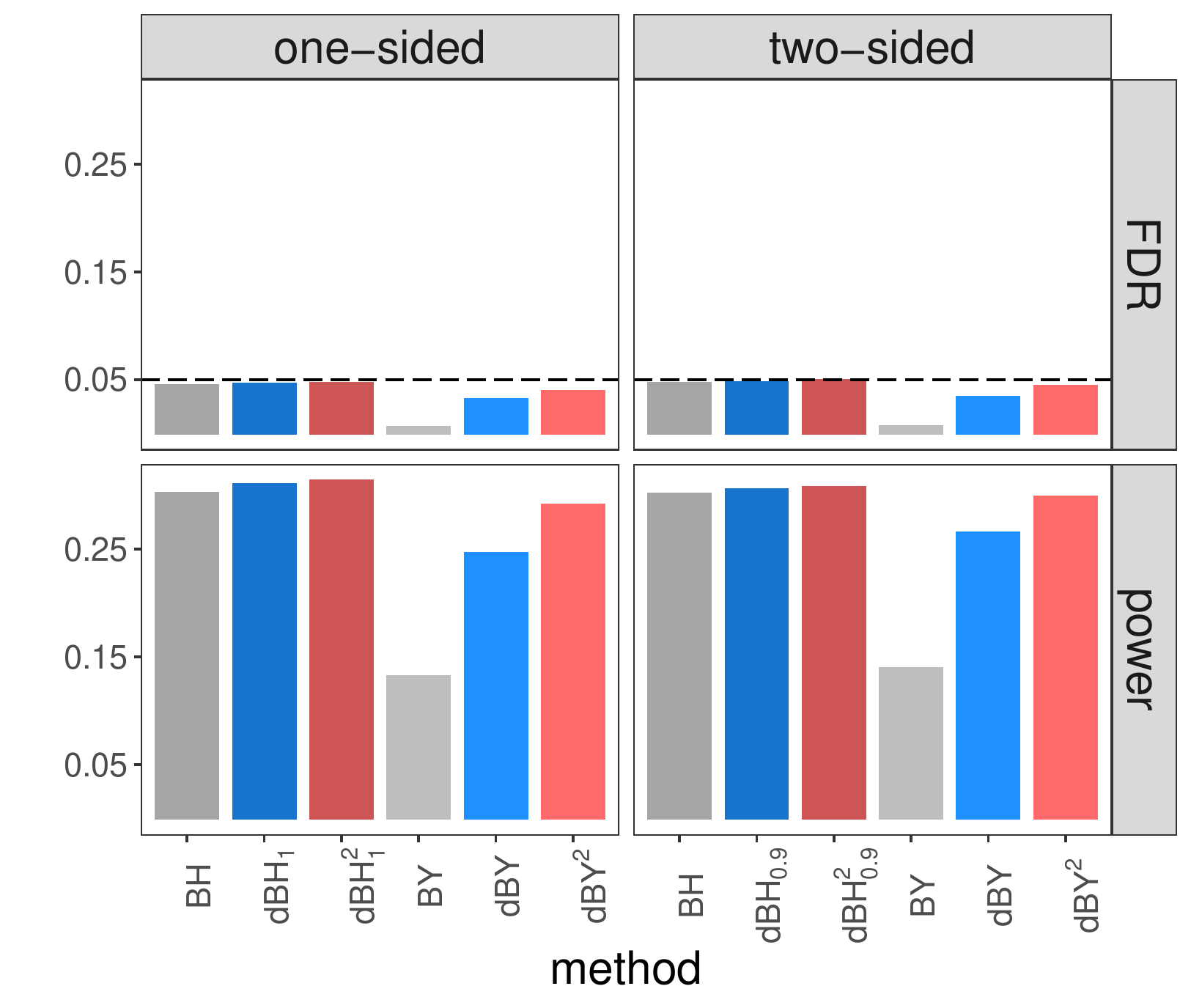}
    \caption{Block dependence}
  \end{subfigure}
  \caption{Estimated \FDR ~and power for multivariate z-statistics.}\label{fig:mvgauss}
\end{figure}

We observe that $\dBH$ and $\dBH^2$ slightly improve the power of $\BH$, while $\dBY$ and $\dBY^2$ significantly improve the power of $\BY$, for one- and two-sided testing with both covariance structures. For one-sided testing, the $p$-values are CPRD, so Theorem~\ref{thm:recursive} guarantees that $\BH, \dBH_1,$ and $\dBH^2_1$ are all safe procedures with nested rejection sets. For two-sided testing, $\BH$ does not provably control the FDR, unlike the other five methods.  

For two-sided testing, $\BH(\alpha)$ does not provably control FDR, but the other five methods do. Although the $\dBY_{0.9}^2(\alpha)$ procedure is not safe, the randomized pruning is never invoked over $1000$ realizations of each simulation scenario. In all four scenarios, the power of $\dBY^2$ is comparable to that of $\BH$.

\begin{figure}[H]
  \centering
  \begin{subfigure}[t]{0.48\textwidth}
    \includegraphics[width = 1\textwidth]{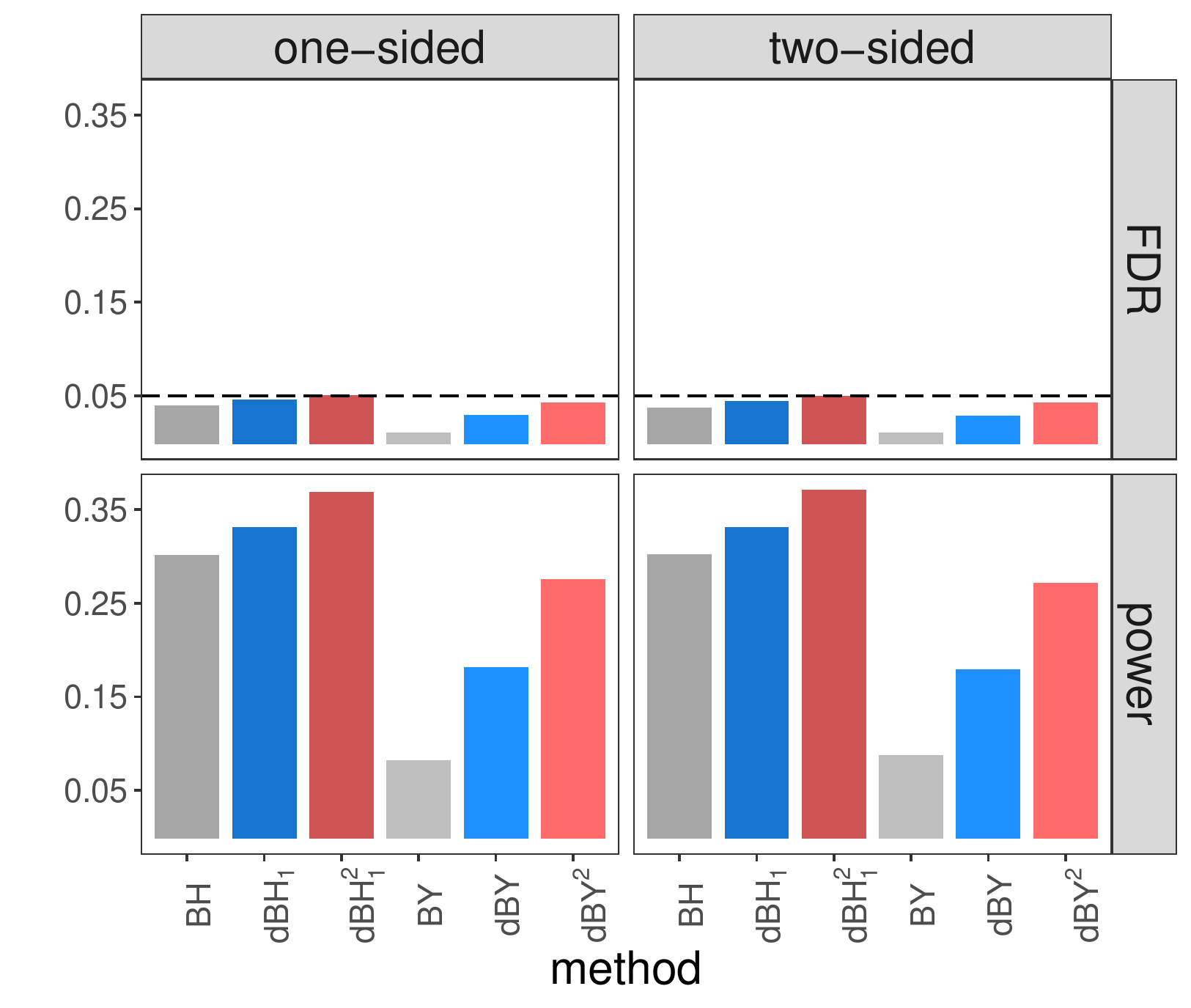}
    \caption{$m = 100, \,\, n - d = 5$}\label{fig:mvt:small}
  \end{subfigure}
  \begin{subfigure}[t]{0.48\textwidth}
    \includegraphics[width = 1\textwidth]{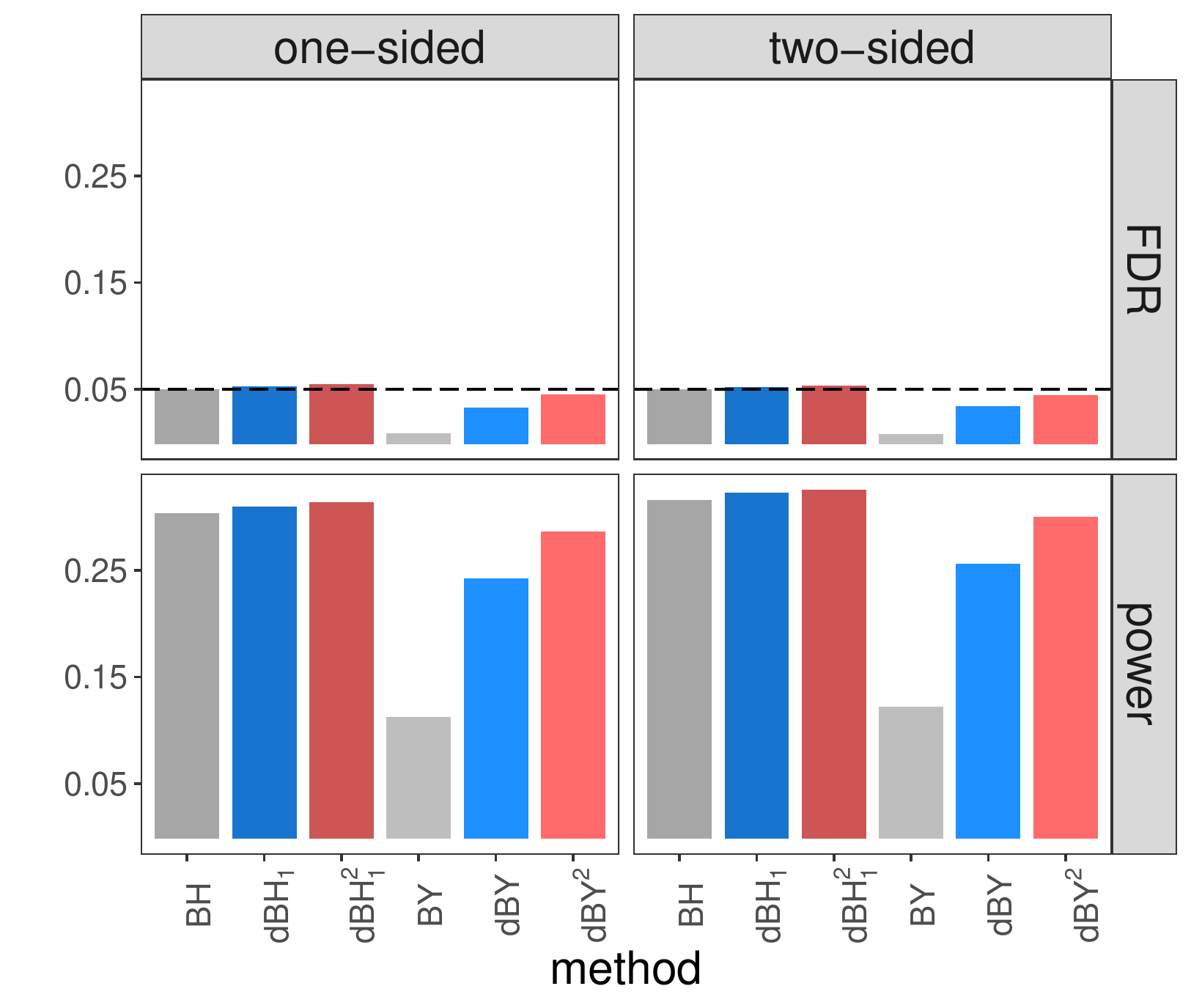}
    \caption{$m = 1000, \,\, n - d = 50$}\label{fig:mvt:large}
  \end{subfigure}
  \caption{Estimated \FDR ~and power for multivariate t-statistics.}\label{fig:mvt}
\end{figure}

Figure~\ref{fig:mvt} shows results for uncorrelated multivariate t-statistics with either $m = 100, n - m = 5$ (Figure~\ref{fig:mvt:small}) or $m = 1000, n - m = 50$ (Figure~\ref{fig:mvt:large}). In the first case, the marginal null distribution of each test statistic is heavy-tailed, and very large values tend to be observed together due to the common variance estimate. In both cases we set the first $10$ hypotheses as alternatives with an equal signal strength, tuned so that $\BH(0.05)$ has approximately $30\%$ power. We evaluate the same six methods as in the multivariate Gaussian case, except that $\gamma$ is taken as $1$ for both one- and two-sided testing because both are CPRD. The results are qualitatively similar to the multivariate Gaussian results. Notably, the power gains of $\dBH_1$ and $\dBH^2_1$ over $\BH$ are more pronounced for heavier-tailed t-statistics.

Finally, we consider two linear modeling scenarios, for which we evaluate the fixed-X knockoff method as an extra competitor. To apply the knockoff method, we always consider two-sided testing problems with $n > 2d$. In this section we simulate the fixed design matrix $\bX$ as one realization of a random matrix with i.i.d. Gaussian entries with $n = 3000$ and $d = 1000$, and simulate $1000$ independent copies of homoscedastic Gaussian error vectors with $\sigma^{2} = 1$, each generating an outcome vector $Y = \bX\beta + \epsilon$ with $\beta_{1} = \cdots = \beta_{10} = \beta^*$ and $\beta_{11} = \cdots = \beta_{1000} = 0$. Again, $\beta^*$ is tuned so that $\BH(0.05)$ has approximately $30\%$ power. For all $\dBH_{\gamma}$ procedures, we choose $\gamma = 0.9$ and find that the randomized pruning step is never invoked for any method in the $1000$ simulations. For the knockoff method, we generate the knockoff matrix via the default semidefinite programming procedure and choose the knockoff statistics as the maximum penalty level at which the variable is selected \citep{barber15}. We use the knockoff+ method in all cases to ensure FDR control at the advertised level. The estimated \FDR ~and power with $\alpha = 0.05$ and $\alpha = 0.2$ are shown in Figure \ref{fig:lm}. The comparison between the $\dBH$ ($\dBY$) procedures and $\BH$ ($\BY$) procedure is qualitatively similar to the previous examples. The fixed-X knockoff has much higher power than all other methods when $\alpha = 0.2$, but has near-zero power when $\alpha = 0.05$. The former may result from the knockoff method's use of the lasso for variable selection \citep{tibshirani1996regression}, while the latter is due to the small-sample issue discussed in Section~\ref{sec:knockoffs}. Appendix~\ref{app:examples} gives several more linear modeling examples showing the same qualitative pattern for random design matrices with different correlation structures.

\begin{figure}[H]
  \centering
  \begin{subfigure}[t]{0.48\textwidth}
    \includegraphics[width = 1\textwidth]{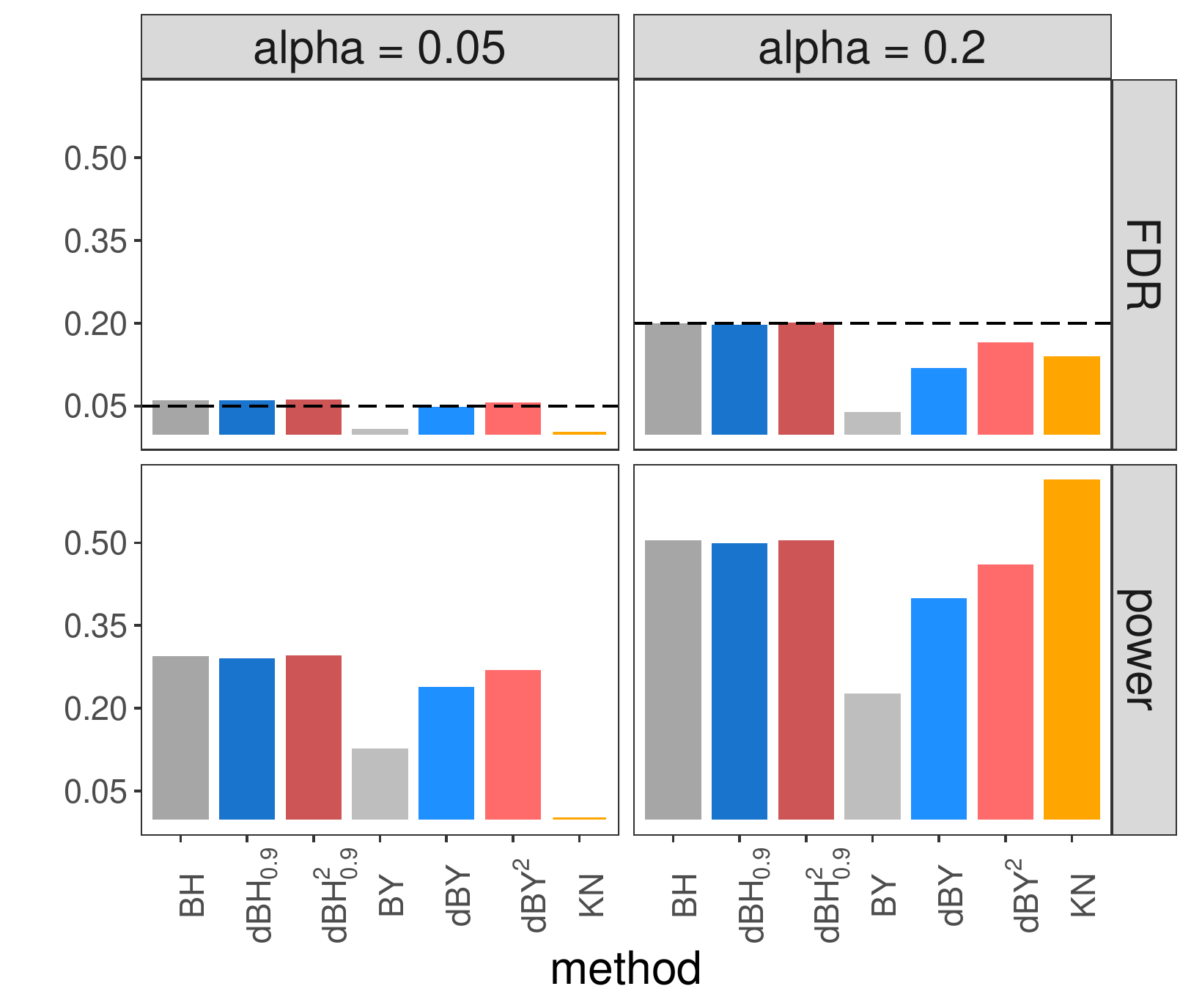}
    \caption{$X$ as a realization of a random matrix}\label{fig:lm}
  \end{subfigure}
  \begin{subfigure}[t]{0.48\textwidth}
    \includegraphics[width = 1\textwidth]{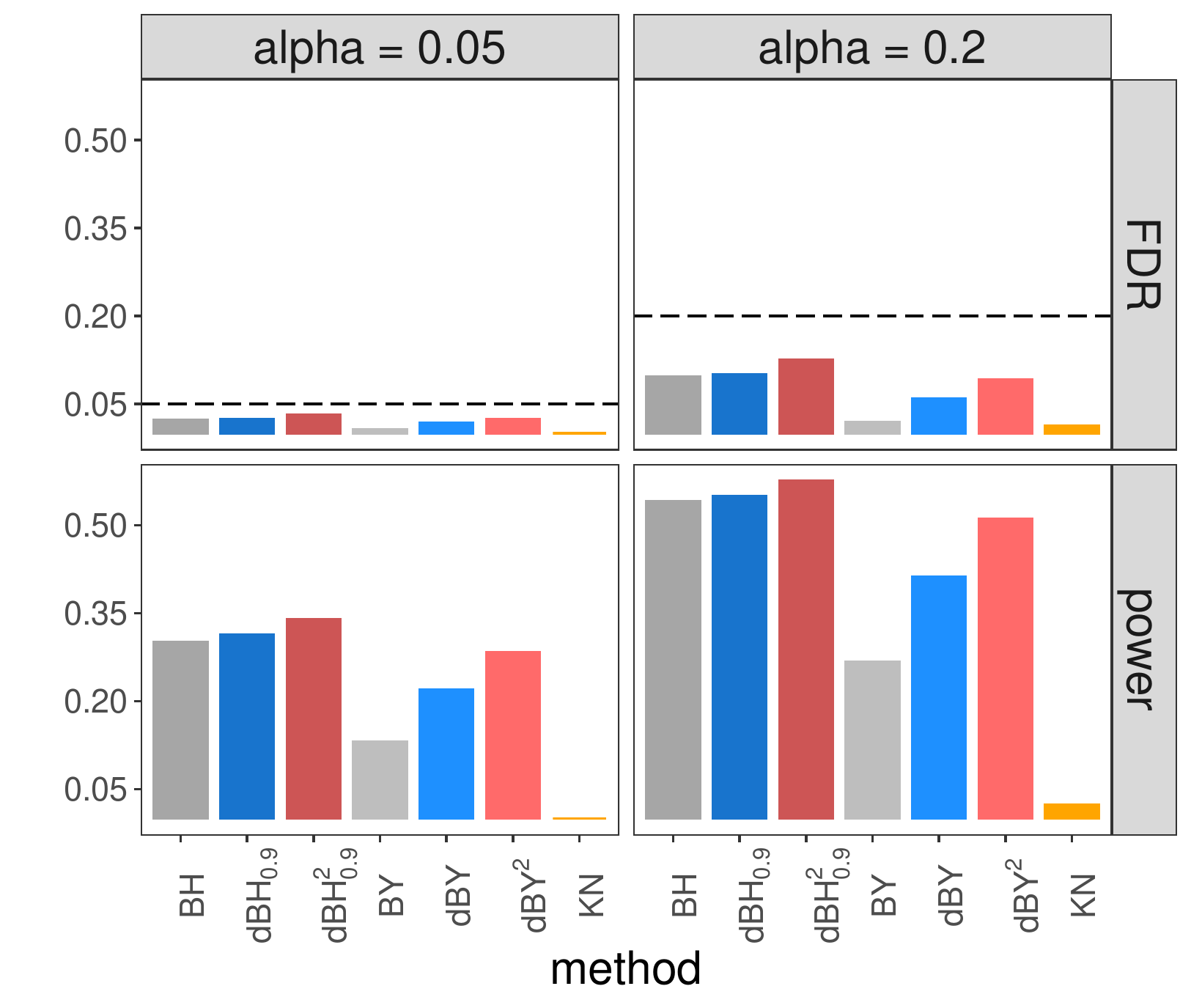}
    \caption{Multiple comparison to control}\label{fig:mcc}
  \end{subfigure}
  \caption{Estimated \FDR ~and power for two types of linear models.}\label{fig:linearmodel}
\end{figure}

While the knockoffs method often outperform the others when $\alpha$ is large enough, the reverse can also occur, as we illustrate in a second linear modeling example: the problem of multiple comparisons to control (MCC) in a one-way layout. For each of $100$ treatment groups, we sample $30$ independent replicates from $N(\mu_i, \sigma^2)$, and a control group with $30$ independent replicates sampled from $N(\mu_0, \sigma^2)$. We then test $H_i: \mu_i = \mu_0$ based on the two-sample t-statistics. In this case, the test statistics are positively equi-correlated. By coding dummy variables this MCC problem is equivalent to a homoscedastic Gaussian linear model with a design matrix $\bX \in \{0,1\}^{3030 \times 101}$ and a coefficient vector $(\mu_1 - \mu_0, \ldots, \mu_{100} - \mu_0, \mu_0)$.  To ameliorate the small-sample issue of the knockoff method, we set the first $30$ hypotheses to be non-nulls with an equal $\mu_i$ that is tuned so that $\BH(\alpha)$ has approximately $30\%$ power. All of the coefficients are inferential targets except $\mu_{0}$, which is effectively an intercept term. For fixed-X knockoffs, we follow \cite{barber15} to generate a knockoff matrix that is orthogonal to the column corresponding to $\mu_{0}$, which is a vector with all entries $1$ in this case. The results are presented in Figure \ref{fig:mcc}. In this case, the fixed-X knockoff is almost powerless for either $\alpha = 0.05$ or $\alpha = 0.2$. In contrast to the previous case, the lack of power is caused by the huge amount of noise generated by knockoffs to handle the equi-correlated covariance structure. The other six methods are less sensitive to this correlation. 



\section{HIV drug resistance data}\label{sec:hiv}

This section compares our method's performance against the BH, BY, and knockoff procedures on the Human Immunodeficiency Virus (HIV) drug resistance data of \citet{rhee2006genotypic}, reproducing and extending the analysis of \citet{barber15}. In each of three separate data sets, we test for associations between mutations present in different HIV samples and resistance to each of 16 different drugs. The data come from three experiments, each for a different drug category: protease inhibitors (PIs), nucleoside reverse transcriptase inhibitors (NRTIs), and  nonnucleoside reverse transcriptase inhibitors (NNRTIs). 

Following \citet{barber15}, we encode mutations as binary with $x_{ij}=1$ if the $j$th mutation is present in the $i$th sample, discard mutations that occur fewer than three times, and remove duplicated columns in the resulting design matrix $\bX \in \{0,1\}^{n \times d}$. For each drug there is a different response vector $Y \in \RR^n$ representing a measure of drug resistance. As in \citet{barber15} we do not include an intercept in the model. We also evaluate replicability in the same way as \citet{barber15}, by comparing the rejection set to the set of mutations identified in the treatment-selected mutation (TSM) panel of \citet{rhee2005hiv}. We refer to Section 4 of \cite{barber15} for further details.

Figure~\ref{fig:hiv-twenty} shows results comparing results for the fixed-X knockoffs, $\BH$, $\dBH_{0.9}^2$, and $\dBY^2$, at significance level $\alpha = 0.2$, as used in \citet{barber15}. For the knockoff method, we generate equi-correlated knockoff copies and use as the knockoff statistic the maximum penalty level at which the variable is selected, following the vingette of the \texttt{knockoff} package \citep{knockoff}. 
The latter three have similar power for all seven responses, with the behavior of $\dBH_{0.9}^2$ nearly identical to $\BH$ and $\dBY^2$ very slightly less powerful. By contrast, the knockoffs method makes somewhat fewer rejections overall than the other methods, but the differences are modest for most drugs. Knockoffs appears to have a higher replicability rate for the TSM panel, possibly because the method is achieving a better tradeoff between Type I and Type II error by using the lasso algorithm to select variables. Alternatively, it may be that the other methods are better able to detect weak signals which are less likely to be replicated in an independent experiment. The $\dBH_{0.9}^2$ method does not require randomization for any of the 16 drugs.

Figure~\ref{fig:hiv-five} shows the same results at the more conservative significance level $\alpha = 0.05$, where knockoffs suffers from the small-sample issues discussed in Section \ref{sec:knockoffs}. The relationships between the other three methods are qualitatively the same. Again, $\dBH_{0.9}^2$ does not require randomization for any of the drugs.

\begin{figure} 
  \centering
  \includegraphics[width=0.9\columnwidth]{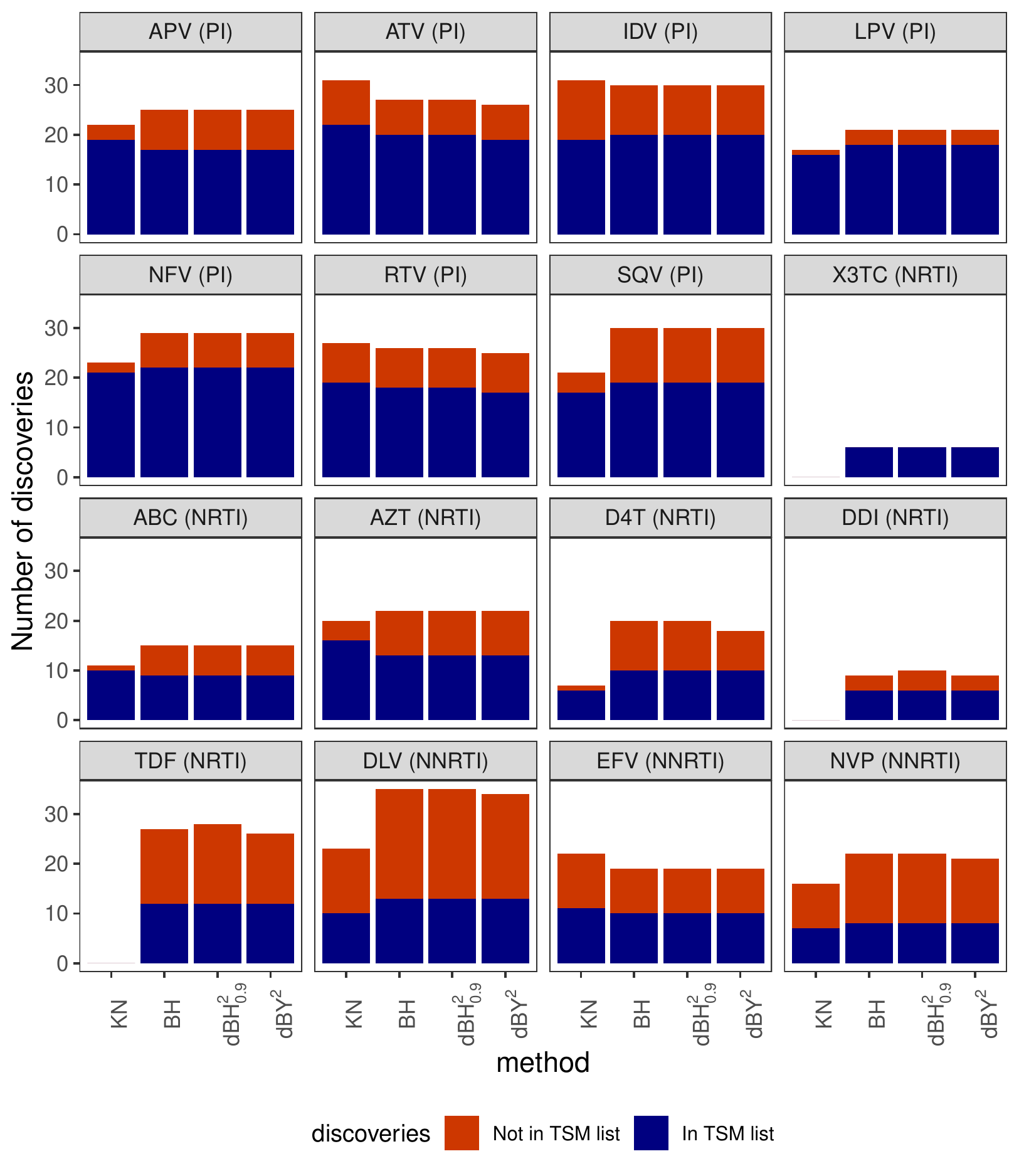}

  \caption{Results on the HIV drug resistance data with $\alpha = 0.2$. The blue segments represent the number of discoveries that were replicated in the TSM panel, while the orange segments represent the number that were not. Results are shown for the fixed-X knockoffs, $\BH$, $\dBH_{0.9}^2$, and $\dBY$.}
 \label{fig:hiv-twenty}
\end{figure}

\begin{figure}
  \centering
  \includegraphics[width=0.9\columnwidth]{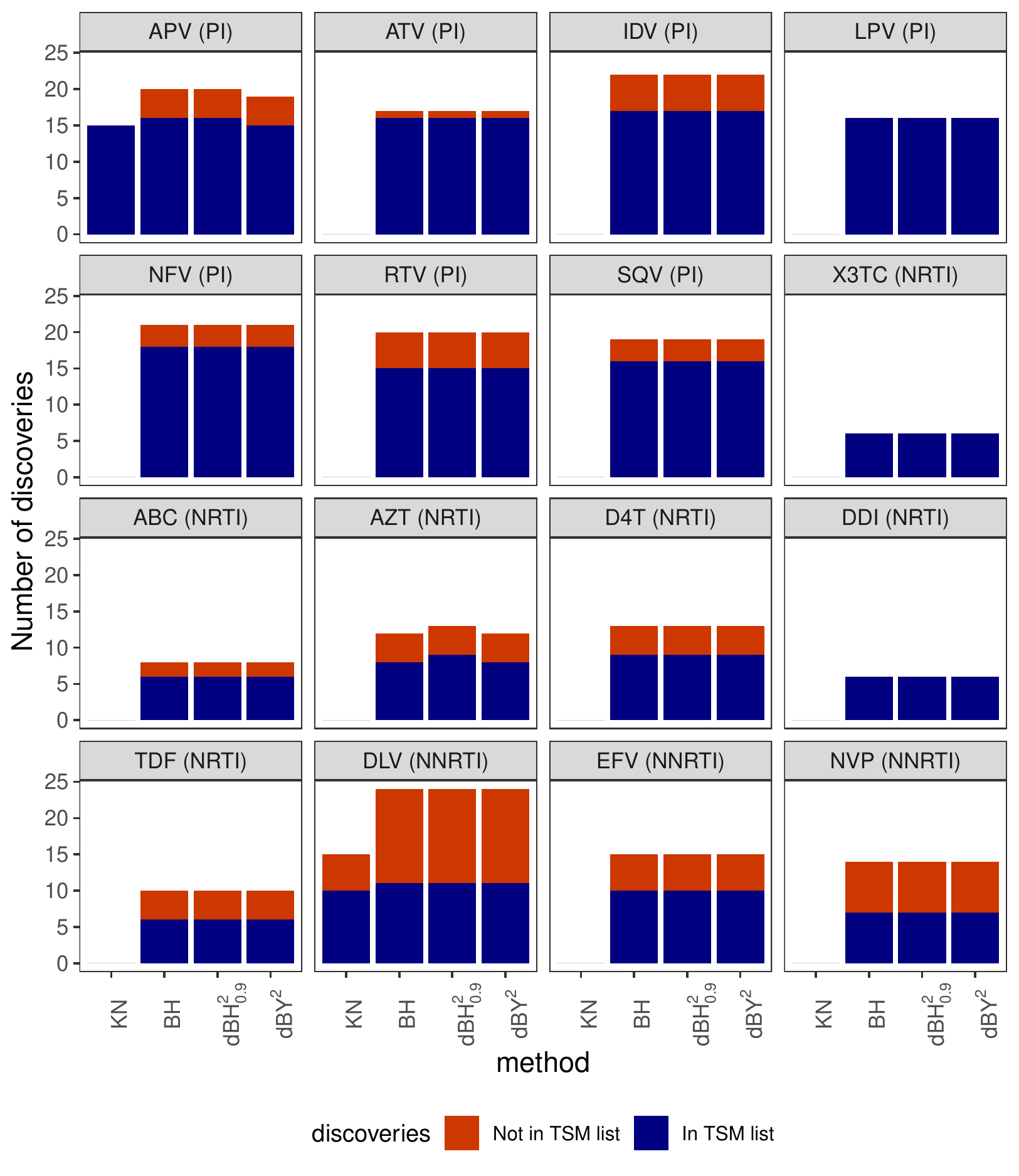}
  \caption{Results on the HIV drug resistance data with $\alpha = 0.05$. The blue segments represent the number of discoveries that were replicated in the TSM panel, while the orange segments represent the number that were not. Results are shown for the fixed-X knockoffs, $\BH$, $\dBH_{0.9}^2$, and $\dBY$.}
  \label{fig:hiv-five}
\end{figure}




\section{Discussion}\label{sec:discussion}

We have presented a new approach for controlling FDR in dependent settings, and proposed new dependence-adjusted step-up methods including the $\dBH_\gamma$, $\dBY$, and $\dSU_\Delta$ procedures. The $\dBH_1$ procedure uniformly improves on the $\BH$ procedure under (conditional) positive dependence, while the $\dBY$ procedure uniformly improves on the $\BY$ procedure. Likewise, our $\dSU$ method can uniformly improve on any shape function method in the style of \citet{blanchard2008two}. 

Practically speaking, our methods offer an alternative to the BH and BY procedures in applications where theoretical FDR control guarantees are attractive. In particular, $\dBY^2$ improves dramatically on the BY procedure and is often competitive even with BH. The $\dBH_{0.9}^2$ procedure offers a balanced approach that is commonly more powerful than BH, and requires randomization only very rarely.

More generally, conditional calibration as proposed here is a general-purpose technical device that may prove useful for supplying FDR control proofs in other contexts like grouped, hierarchical, multilayer, or partial conjunction hypothesis testing \citep[e.g.][]{benjamini14,barber16,lynch16,benjamini2008screening}.

Numerous challenges remain for future work, including investigation into models that constrain the dependence graph, as in Example~\ref{ex:nonpar-graph}. In addition, further development is needed to produce algorithms and software for some of the models we did not implement in this work. Finally, the next sections suggest directions of further methodological innovation.

\subsection{Comparison with knockoffs}\label{sec:knockoffs}

Both the $\dBH$ method and the knockoff filter offer finite-sample FDR control for linear models, but with very different statistical tools and methods of proof. In our simulation experiments neither method is a clear overall winner, but some qualitative trends emerge. First, as expected, the dBH procedure performs similarly to BH in power comparisons, so any comparison between knockoffs and dBH is also a comparison between knockoffs and BH. Second, the dBH and BH procedures consistently enjoy better power in experiments where the total number of rejections is relatively small, either because there are very few non-null coefficients to find, the signals are weak, or the FDR significance level is small. This pattern has a clear theoretical explanation: to make rejections, the knockoff+ method requires $(1 + A_t)/R_t \leq \alpha$, where $A_t$ is the count of $W_j$-values smaller than $-t$ and $R_t$ is the count above $t$. As a result the method cannot make any rejections unless it makes at least $1/\alpha$, and it can be unstable if the number of rejections is on the order of $1/\alpha$. 

Apart from small-sample issues, it remains unclear in which contexts we should expect one method to outperform the other, and this is an interesting question for future research. Because the knockoffs framework is very general and allows the analyst to bring a great deal of prior knowledge to bear, we expect it can enjoy substantial advantages over BH and dBH in problems where the $j$th $t$-statistic carries only a small fraction of the total evidence against $H_j:\; \beta_j = 0$. In particular, $U_j = \hat\beta_{-j} - \hat\beta_j \bb^j$ is independent of the $t$-statistic and $p$-value for $H_j$, but it may hold a wealth of information about $\beta_j$, especially if we reasonably expect that $\beta$ is approximately sparse. In our view, it is an important open problem to develop methods that can likewise exploit this kind of information, for example by using adaptive weights as proposed below, while avoiding the randomization and binarization inherent to knockoff methods.

Our method is also extensible to many settings where no knockoff method has been proposed, for example edge testing in the Gaussian graphical model and the discrete and nonparametric examples discussed in Appendix~\ref{app:examples}. Because our method operates directly on $p$-values it is easily extensible to testing composite hypotheses about parameters, for example to test $H_j:\; |\beta_j| \leq \delta$ for a fixed $\delta > 0$. Finally, in regression problems, there is no requirement that $n \geq 2d$; we require only that $n \geq d + 1$, the same dimension required to test individual regression coefficients. Conversely, we have not extended our framework to conditional randomization tests as proposed in \citet{candes2018panning}, and this may be very challenging in general. Computationally, our method is more scalable for some problems because it avoids solving a semidefinite program or eigendecomposing a large matrix, but the recursively refined variants of our method pose substantial computational challenges of their own. Nevertheless, computational efficiency depends on problem specifics.

\subsection{Extension: adapting to the non-null proportion}

One arguable weakness of the present work is its conservative control of the FDR at level $\alpha m_0/m$. In some problems it would be very useful to correct for this conservatism; for example, in post-screening or other post-selection inference, we may expect $\pi_0 = m_0/m$ to be substantially smaller than 1. For independent $p$-values, various plug-in methods apply a standard method such as BH at an adjusted level $\alpha \hat\pi_0^{-1}$, for some estimator $\hat\pi_0(X)$. \citep[e.g.][]{genovese2002operating, storey02, storey04, benjamini2006adaptive, blanchard2009adaptive}.

Inspired by this approach, we can modify our calibration procedure as discussed in Remark~\ref{rem:rhs} to use the calibration constraint $\kappa_i(S_i) = \alpha \hat\pi_{0,i}(S_i)^{-1} / m$ in place of $\alpha/m$ in \eqref{eq:def-hc}. The resulting method would control FDR at level $\alpha$ provided that
\[
\sum_{i\in \cH_0} \EE \hat\pi_{0,i}(S_i)^{-1} \leq \pi_{0,i}^{-1} m_0 = m,
\]
for which a sufficient condition is that $\EE_{H_i} \hat\pi_0(S_i)^{-1} \leq \pi_0^{-1}$ for each $i = 1,\ldots, m$. Given any pre-existing estimator $\tilde\pi_0(X)$ for which $\EE \tilde\pi_0(X)^{-1} \leq \pi_0^{-1}$, we can construct such an estimator by Rao-Blackwellization:
\[
\hat\pi_{0,i}(S_i) = \left( \EE_{H_i} \left[\,\tilde\pi_0(X)^{-1} \mid S_i\,\right] \right)^{-1},
\]
which amounts to a simple calculation if $H_i$ is conditionally simple.

\subsection{Extension: adaptive weights}

Another promising extension, the full exploration of which is outside the scope of this work, is to use adaptive weights that exploit side information about the hypotheses. There are a variety of setting where $p$-value weights $w_i$ can substantially improve the power of multiple-testing methods \citep{benjamini1997multiple,genovese2006false, dobriban2015optimal}. For fixed weights $w_1,\ldots,w_m$ that sum to one, it is straightforward to generalize our framework by replacing $\alpha/m$ by $\alpha w_i$ in the right-hand side of \eqref{eq:def-hc}.

More interestingly, however, we might wish to use learn the weights from the data, adaptively allowing for some hypotheses to contribute more to the FDR than others. While there is a robust literature on adaptive $p$-value weighting for independent hypotheses \citep[e.g.][]{ignatiadis2016data, boca2017regression, li2019multiple, lei2018adapt, ignatiadis2017covariate, xia2017neuralfdr, tansey2018black}, there is very little work on adaptive weighting for dependent $p$-values. This is a major gap in the literature, since true independence between $p$-values is rare in applied problems

Similarly to the strategy described above for estimating $\pi_0$, we can accommodate data-adaptive weights by using $\kappa_i(S_i) = \alpha w_i(S_i)$, provided that $\sum_{i \in \cH_0} \EE w_i(S_i) \leq 1$. As above, we can Rao-Blackwellize initial weights $\widetilde w$ by setting $w_i(S_i) = \EE_{H_i}\left[\, \widetilde w_i \mid S_i\,\right]$. If $\sum_i \widetilde w_i(X) \leq 1$ almost surely, then
\[
\sum_{i \in \cH_0} \EE \kappa_i(S_i) = \alpha \sum_{i \in \cH_0} \EE \widetilde{w}_i \leq \alpha.
\]

We defer exploration of this idea to future work.

\section*{Reproducibility}

Our \texttt{R} package \texttt{dbh} is available to download at \url{https://github.com/lihualei71/dbh}. A public github repo accompanying the paper with code to reproduce the figures herein can be found at \url{https://github.com/lihualei71/dbhPaper}.

\section*{Acknowledgments}

William Fithian is supported in part by the NSF DMS-1916220 and a Hellman Fellowship from Berkeley. We are grateful to Patrick Chao and Jonathan Taylor for helpful feedback on a draft of this paper.

\bibliographystyle{plainnat}
\bibliography{biblio.bib}

\begin{thebibliography}{45}
\providecommand{\natexlab}[1]{#1}
\providecommand{\url}[1]{\texttt{#1}}
\expandafter\ifx\csname urlstyle\endcsname\relax
  \providecommand{\doi}[1]{doi: #1}\else
  \providecommand{\doi}{doi: \begingroup \urlstyle{rm}\Url}\fi

\bibitem[Barber and Cand{\`e}s(2015)]{barber15}
Rina~Foygel Barber and Emmanuel~J Cand{\`e}s.
\newblock Controlling the false discovery rate via knockoffs.
\newblock \emph{The Annals of Statistics}, 43\penalty0 (5):\penalty0
  2055--2085, 2015.

\bibitem[Barber and Ramdas(2016)]{barber16}
Rina~Foygel Barber and Aaditya Ramdas.
\newblock The p-filter: multilayer false discovery rate control for grouped
  hypotheses.
\newblock \emph{Journal of the Royal Statistical Society: Series B (Statistical
  Methodology)}, 2016.

\bibitem[Benjamini and Bogomolov(2014)]{benjamini14}
Yoav Benjamini and Marina Bogomolov.
\newblock Selective inference on multiple families of hypotheses.
\newblock \emph{Journal of the Royal Statistical Society: Series B (Statistical
  Methodology)}, 76\penalty0 (1):\penalty0 297--318, 2014.

\bibitem[Benjamini and Heller(2008)]{benjamini2008screening}
Yoav Benjamini and Ruth Heller.
\newblock Screening for partial conjunction hypotheses.
\newblock \emph{Biometrics}, 64\penalty0 (4):\penalty0 1215--1222, 2008.

\bibitem[Benjamini and Hochberg(1995)]{bh95}
Yoav Benjamini and Yosef Hochberg.
\newblock Controlling the false discovery rate: a practical and powerful
  approach to multiple testing.
\newblock \emph{Journal of the Royal Statistical Society. Series B
  (Methodological)}, pages 289--300, 1995.

\bibitem[Benjamini and Hochberg(1997)]{benjamini1997multiple}
Yoav Benjamini and Yosef Hochberg.
\newblock Multiple hypotheses testing with weights.
\newblock \emph{Scandinavian Journal of Statistics}, 24\penalty0 (3):\penalty0
  407--418, 1997.

\bibitem[Benjamini and Yekutieli(2001)]{benjamini2001control}
Yoav Benjamini and Daniel Yekutieli.
\newblock The control of the false discovery rate in multiple testing under
  dependency.
\newblock \emph{Annals of statistics}, pages 1165--1188, 2001.

\bibitem[Benjamini et~al.(2006)Benjamini, Krieger, and
  Yekutieli]{benjamini2006adaptive}
Yoav Benjamini, Abba~M Krieger, and Daniel Yekutieli.
\newblock Adaptive linear step-up procedures that control the false discovery
  rate.
\newblock \emph{Biometrika}, 93\penalty0 (3):\penalty0 491--507, 2006.

\bibitem[Blanchard and Roquain(2008)]{blanchard2008two}
Gilles Blanchard and Etienne Roquain.
\newblock Two simple sufficient conditions for fdr control.
\newblock \emph{Electronic journal of Statistics}, 2:\penalty0 963--992, 2008.

\bibitem[Blanchard and Roquain(2009)]{blanchard2009adaptive}
Gilles Blanchard and {\'E}tienne Roquain.
\newblock Adaptive false discovery rate control under independence and
  dependence.
\newblock \emph{Journal of Machine Learning Research}, 10\penalty0
  (Dec):\penalty0 2837--2871, 2009.

\bibitem[Boca and Leek(2017)]{boca2017regression}
Simina~M Boca and Jeffrey~T Leek.
\newblock A regression framework for the proportion of true null hypotheses.
\newblock \emph{Preprint bioRxiv}, 35675, 2017.

\bibitem[Cand\`{e}s et~al.(2018)Cand\`{e}s, Fan, Janson, and
  Lv]{candes2018panning}
Emmanuel Cand\`{e}s, Yingying Fan, Lucas Janson, and Jinchi Lv.
\newblock Panning for gold: 'model-x' knockoffs for high dimensional controlled
  variable selection.
\newblock \emph{Journal of the Royal Statistical Society: Series B (Statistical
  Methodology)}, 80\penalty0 (3):\penalty0 551--577, 2018.

\bibitem[Dobriban et~al.(2015)Dobriban, Fortney, Kim, and
  Owen]{dobriban2015optimal}
Edgar Dobriban, Kristen Fortney, Stuart~K Kim, and Art~B Owen.
\newblock Optimal multiple testing under a gaussian prior on the effect sizes.
\newblock \emph{Biometrika}, 102\penalty0 (4):\penalty0 753--766, 2015.

\bibitem[Farcomeni(2006)]{farcomeni2006more}
Alessio Farcomeni.
\newblock More powerful control of the false discovery rate under dependence.
\newblock \emph{Statistical Methods and Applications}, 15\penalty0
  (1):\penalty0 43--73, 2006.

\bibitem[Farcomeni(2007)]{farcomeni2007some}
Alessio Farcomeni.
\newblock Some results on the control of the false discovery rate under
  dependence.
\newblock \emph{Scandinavian Journal of Statistics}, 34\penalty0 (2):\penalty0
  275--297, 2007.

\bibitem[Ferreira and Zwinderman(2006)]{ferreira2006benjamini}
JA~Ferreira and AH~Zwinderman.
\newblock On the benjamini--hochberg method.
\newblock \emph{The Annals of Statistics}, 34\penalty0 (4):\penalty0
  1827--1849, 2006.

\bibitem[Finner(1999)]{finner1999stepwise}
Helmut Finner.
\newblock Stepwise multiple test procedures and control of directional errors.
\newblock \emph{The Annals of Statistics}, 27\penalty0 (1):\penalty0 274--289,
  1999.

\bibitem[Fithian et~al.(2014)Fithian, Sun, and Taylor]{fithian2014optimal}
William Fithian, Dennis Sun, and Jonathan Taylor.
\newblock Optimal inference after model selection.
\newblock \emph{arXiv preprint arXiv:1410.2597}, 2014.

\bibitem[Genovese and Wasserman(2002)]{genovese2002operating}
Christopher Genovese and Larry Wasserman.
\newblock Operating characteristics and extensions of the false discovery rate
  procedure.
\newblock \emph{Journal of the Royal Statistical Society: Series B (Statistical
  Methodology)}, 64\penalty0 (3):\penalty0 499--517, 2002.

\bibitem[Genovese and Wasserman(2004)]{genovese2004stochastic}
Christopher Genovese and Larry Wasserman.
\newblock A stochastic process approach to false discovery control.
\newblock \emph{The Annals of Statistics}, 32\penalty0 (3):\penalty0
  1035--1061, 2004.

\bibitem[Genovese et~al.(2006)Genovese, Roeder, and
  Wasserman]{genovese2006false}
Christopher~R Genovese, Kathryn Roeder, and Larry Wasserman.
\newblock False discovery control with p-value weighting.
\newblock \emph{Biometrika}, 93\penalty0 (3):\penalty0 509--524, 2006.

\bibitem[Ignatiadis and Huber(2017)]{ignatiadis2017covariate}
Nikolaos Ignatiadis and Wolfgang Huber.
\newblock Covariate-powered weighted multiple testing with false discovery rate
  control.
\newblock \emph{arXiv preprint arXiv:1701.05179}, 2017.

\bibitem[Ignatiadis et~al.(2016)Ignatiadis, Klaus, Zaugg, and
  Huber]{ignatiadis2016data}
Nikolaos Ignatiadis, Bernd Klaus, Judith~B Zaugg, and Wolfgang Huber.
\newblock Data-driven hypothesis weighting increases detection power in
  genome-scale multiple testing.
\newblock \emph{Nature methods}, 2016.

\bibitem[Kim and van~de Wiel(2008)]{kim2008effects}
Kyung~In Kim and Mark~A van~de Wiel.
\newblock Effects of dependence in high-dimensional multiple testing problems.
\newblock \emph{BMC bioinformatics}, 9\penalty0 (1):\penalty0 114, 2008.

\bibitem[Lee et~al.(2016)Lee, Sun, Sun, and Taylor]{lee2016exact}
Jason~D Lee, Dennis~L Sun, Yuekai Sun, and Jonathan~E Taylor.
\newblock Exact post-selection inference, with application to the lasso.
\newblock \emph{The Annals of Statistics}, 44\penalty0 (3):\penalty0 907--927,
  2016.

\bibitem[Lehmann and Romano(2005)]{lehmann2005testing}
EL~Lehmann and Joseph~P Romano.
\newblock \emph{Testing statistical hypotheses}.
\newblock New York:. Springer, 2005.

\bibitem[Lehmann and Scheff{\'e}(1955)]{lehmann1955completeness}
EL~Lehmann and Henry Scheff{\'e}.
\newblock Completeness, similar regions, and unbiased estimation: Part ii.
\newblock \emph{Sankhy{\=a}: The Indian Journal of Statistics (1933-1960)},
  15\penalty0 (3):\penalty0 219--236, 1955.

\bibitem[Lei and Fithian(2018)]{lei2018adapt}
Lihua Lei and William Fithian.
\newblock Adapt: an interactive procedure for multiple testing with side
  information.
\newblock \emph{Journal of the Royal Statistical Society: Series B (Statistical
  Methodology)}, 80\penalty0 (4):\penalty0 649--679, 2018.

\bibitem[Li and Barber(2019)]{li2019multiple}
Ang Li and Rina~Foygel Barber.
\newblock Multiple testing with the structure-adaptive benjamini--hochberg
  algorithm.
\newblock \emph{Journal of the Royal Statistical Society: Series B (Statistical
  Methodology)}, 81\penalty0 (1):\penalty0 45--74, 2019.

\bibitem[Lynch and Guo(2016)]{lynch16}
Gavin Lynch and Wenge Guo.
\newblock On procedures controlling the {FDR} for testing hierarchically
  ordered hypotheses.
\newblock \emph{arXiv preprint arXiv:1612.04467}, 2016.

\bibitem[Patterson and Sesia(2018)]{knockoff}
Evan Patterson and Matteo Sesia.
\newblock \emph{knockoff: The Knockoff Filter for Controlled Variable
  Selection}, 2018.
\newblock URL \url{https://CRAN.R-project.org/package=knockoff}.
\newblock R package version 0.3.2.

\bibitem[Rhee et~al.(2005)Rhee, Fessel, Zolopa, Hurley, Liu, Taylor, Nguyen,
  Slome, Klein, Horberg, et~al.]{rhee2005hiv}
Soo-Yon Rhee, W~Jeffrey Fessel, Andrew~R Zolopa, Leo Hurley, Tommy Liu,
  Jonathan Taylor, Dong~Phuong Nguyen, Sally Slome, Daniel Klein, Michael
  Horberg, et~al.
\newblock Hiv-1 protease and reverse-transcriptase mutations: correlations with
  antiretroviral therapy in subtype b isolates and implications for
  drug-resistance surveillance.
\newblock \emph{The Journal of infectious diseases}, 192\penalty0 (3):\penalty0
  456--465, 2005.

\bibitem[Rhee et~al.(2006)Rhee, Taylor, Wadhera, Ben-Hur, Brutlag, and
  Shafer]{rhee2006genotypic}
Soo-Yon Rhee, Jonathan Taylor, Gauhar Wadhera, Asa Ben-Hur, Douglas~L Brutlag,
  and Robert~W Shafer.
\newblock Genotypic predictors of human immunodeficiency virus type 1 drug
  resistance.
\newblock \emph{Proceedings of the National Academy of Sciences}, 103\penalty0
  (46):\penalty0 17355--17360, 2006.

\bibitem[Romano et~al.(2008)Romano, Shaikh, and Wolf]{romano2008control}
Joseph~P Romano, Azeem~M Shaikh, and Michael Wolf.
\newblock Control of the false discovery rate under dependence using the
  bootstrap and subsampling.
\newblock \emph{Test}, 17\penalty0 (3):\penalty0 417, 2008.

\bibitem[Shaffer(1980)]{shaffer1980control}
Juliet~Popper Shaffer.
\newblock Control of directional errors with stagewise multiple test
  procedures.
\newblock \emph{The Annals of Statistics}, pages 1342--1347, 1980.

\bibitem[Storey(2002)]{storey02}
John~D Storey.
\newblock A direct approach to false discovery rates.
\newblock \emph{Journal of the Royal Statistical Society: Series B (Statistical
  Methodology)}, 64\penalty0 (3):\penalty0 479--498, 2002.

\bibitem[Storey(2003)]{storey2003positive}
John~D Storey.
\newblock The positive false discovery rate: a bayesian interpretation and the
  q-value.
\newblock \emph{The Annals of Statistics}, 31\penalty0 (6):\penalty0
  2013--2035, 2003.

\bibitem[Storey et~al.(2004)Storey, Taylor, and Siegmund]{storey04}
John~D Storey, Jonathan~E Taylor, and David Siegmund.
\newblock Strong control, conservative point estimation and simultaneous
  conservative consistency of false discovery rates: a unified approach.
\newblock \emph{Journal of the Royal Statistical Society: Series B (Statistical
  Methodology)}, 66\penalty0 (1):\penalty0 187--205, 2004.

\bibitem[Tansey et~al.(2018)Tansey, Wang, Blei, and Rabadan]{tansey2018black}
Wesley Tansey, Yixin Wang, David~M Blei, and Raul Rabadan.
\newblock Black box fdr.
\newblock \emph{arXiv preprint arXiv:1806.03143}, 2018.

\bibitem[Tian and Taylor(2018)]{tian2018selective}
Xiaoying Tian and Jonathan Taylor.
\newblock Selective inference with a randomized response.
\newblock \emph{The Annals of Statistics}, 46\penalty0 (2):\penalty0 679--710,
  2018.

\bibitem[Tibshirani(1996)]{tibshirani1996regression}
Robert Tibshirani.
\newblock Regression shrinkage and selection via the lasso.
\newblock \emph{Journal of the Royal Statistical Society: Series B
  (Methodological)}, 58\penalty0 (1):\penalty0 267--288, 1996.

\bibitem[Tibshirani et~al.(2016)Tibshirani, Taylor, Lockhart, and
  Tibshirani]{tibshirani2016exact}
Ryan~J Tibshirani, Jonathan Taylor, Richard Lockhart, and Robert Tibshirani.
\newblock Exact post-selection inference for sequential regression procedures.
\newblock \emph{Journal of the American Statistical Association}, 111\penalty0
  (514):\penalty0 600--620, 2016.

\bibitem[Troendle(2000)]{troendle2000stepwise}
James~F Troendle.
\newblock Stepwise normal theory multiple test procedures controlling the false
  discovery rate.
\newblock \emph{Journal of Statistical Planning and Inference}, 84\penalty0
  (1-2):\penalty0 139--158, 2000.

\bibitem[Weinstein et~al.(2013)Weinstein, Fithian, and
  Benjamini]{weinstein2013selection}
Asaf Weinstein, William Fithian, and Yoav Benjamini.
\newblock Selection adjusted confidence intervals with more power to determine
  the sign.
\newblock \emph{Journal of the American Statistical Association}, 108\penalty0
  (501):\penalty0 165--176, 2013.

\bibitem[Xia et~al.(2017)Xia, Zhang, Zou, and Tse]{xia2017neuralfdr}
Fei Xia, Martin~J Zhang, James~Y Zou, and David Tse.
\newblock Neuralfdr: Learning discovery thresholds from hypothesis features.
\newblock In \emph{Advances in Neural Information Processing Systems}, pages
  1541--1550, 2017.

\end{thebibliography}

\appendix


\section{Proofs}\label{app:proofs}

We restate and prove several of the technical results from the paper.

\lempmi*
\begin{proof}
Let $p^0 = \pmi$, $R^0 = R(\pmi)$, $R = R(p)$. By the properties of step-up procedures, we always have $p_{(R^0)}^0 \leq \Delta(R^0)$ and $p_{(R)} \leq \Delta(R)$, and $j \in \cR(p)$ iff $p_j \leq \Delta(R)$. By monotonicity, we have $R^0 \geq R$ and $\Delta(R^0) \geq \Delta(R)$.
  
  $(1 \Rightarrow 2)$. Assume $p_i \leq \Delta(R^0)$. Because $p_{(R^0)}^0 \leq \Delta(R^0)$, and $p$ and $p^0$ only differ in their $i$th coordinate, we have $p_{(R^0)} \leq p_i \vee p_{(R^0)}^0 \leq \Delta(R^0)$. As a result, $R \geq R^0$ and $\Delta(R) \geq \Delta(R^0) \geq p_i$.

  $(2 \Rightarrow 3)$. Assume $i \in \cR(p)$. Then $p_i \leq p_{(R)} \leq \Delta(R)$, and $p_{(r)} > \Delta(r)$ for $r > R$. Because $p_{(r)}^0 = p_{(r)}$ for all order statistics with $p_{(r)} \geq p_i$, we have $R^0 = R$ as well. As a result,
\[
\cR(p) = \{j:\; p_j \leq \Delta(R)\} = \{j:\; p_j \leq \Delta(R^0)\} = \cR(p^0)
\]

  $(3 \Rightarrow 1,2)$. Assume $\cR(p) = \cR(\pmi)$. Because $p^0_i = 0$, $i \in \cR(\pmi) = \cR(p)$, so we must have $p_i \leq \Delta(R) = \Delta(R^0)$.
\end{proof}

\thmsafe*
\begin{proof}

  ~\\{\bf \em Claim 1.} If $R_i^0 = R^{\BH(\alpha)}(p^{(i \gets 0)})$, then by Lemma~\ref{lem:p-minus-i} we have
  \[
  p_i \leq \frac{\alpha R_i^0}{m} \iff i \in \cR^{\BH(\alpha)},
  \]
  with $\hR_i = R^{\BH(\alpha)} = R_i^0$ on the same set. As a result,
  \[
  g_i^*(\alpha; p_{-i}) = \EE_{H_i}\left[\frac{1\{p_i \leq \alpha R_i^0/m\}}{R_i^0} \mid p_{-i}\right] = \frac{\alpha}{m}.
  \]
  For $c > \alpha$, we have $c R_i^0/m \leq \tau^{\BH}(c)$, so
  \[
  g_i^*(c; p_{-i}) \geq \frac{\alpha}{m} + \EE_{H_i}\left[\frac{1\{\alpha R_i^0/m < p_i \leq c R_i^0/m\}}{m} \mid p_{-i}\right] > \frac{\alpha}{m}
  \]
  As a result, $\hc_i = 1$ for all $i = 1,\ldots, m$, so $\cR = \cR_+ = \cR^{\BH(\alpha)}$.

  ~\\{\bf \em Claim 2.} Because $\{\hR_i \leq r\}$ is a non-decreasing set for all $r \leq m$ we have for  $P \in H_i$,
  \begin{align}
    \label{eq:cprds-1}
    \PP_P&\left[\;\hR_i \leq r \;\;\big|\;\; p_i \leq \frac{\alpha r}{m}, \;S_i\;\right]
         \;+\; \PP_P\left[\;\hR_i = r+1 \;\;\big|\;\; p_i \leq \frac{\alpha (r+1)}{m}, \;S_i\;\right]\\[10pt]
    \label{eq:cprds-2}
       &\leq\; \PP_P\left[\;\hR_i \leq r \;\;\big|\;\; p_i \leq \frac{\alpha (r+1)}{m}, \;S_i\;\right]
         \;+\; \PP_P\left[\;\hR_i = r+1 \;\;\big|\;\; p_i \leq \frac{\alpha (r+1)}{m}, \;S_i\;\right]\\[10pt]
    \label{eq:cprds-3}
       &=\; \PP_P\left[\;\hR_i \leq r+1 \;\;\big|\;\; p_i \leq \frac{\alpha (r+1)}{m}, \;S_i\;\right].
  \end{align}
  Beginning with $\{\hR_i \leq 1\} = \{\hR_i = 1\}$ and then iteratively applying the above inequality, we obtain
  \begin{equation}\label{eq:telescope}
    \sum_{r=1}^m \PP_P\left[\;\hR_i = r \;\big|\; p_i \leq \frac{\alpha r}{m}, \;S_i\;\right] 
    \leq \PP_P\left[\;\hR_i \leq m \;\big|\; p_i \leq \alpha, \;S_i\;\right] = 1,
  \end{equation}
  and
  \begin{align}
    \EE_P\left[\; \frac{1\left\{p_i \leq \alpha \hR_i/m\right\}}{\hR_i} \;\big|\; S_i\;\right]
    &= \sum_{r=1}^m \frac{1}{r} \;\PP_P\left[\; \hR_i = r, \; p_i \leq \frac{\alpha r}{m} \;\big|\; S_i\;\right]\\
    \label{eq:cond-prob-p}
    &\leq \sum_{r=1}^m \frac{\alpha}{m} \;\PP_P\left[\; \hR_i = r \;\big|\; p_i \leq \frac{\alpha r}{m} , \;S_i\;\right]\\
    &\leq \frac{\alpha}{m},
  \end{align}
  so $g_i^*(\alpha \semic S_i) \leq \alpha / m$ and $\hc_i \geq \alpha$. As a result, $\cR_+ \supseteq \cR^{\BH(\alpha)}$, so $R_+ \geq \hR_i$ for all $i \in \cR_+$.

  ~\\{\bf \em Claim 4.} Define $R^\alpha = R^{\SU_\Delta(\alpha)}$, $\Delta_\alpha(0)=0$, and the intervals $I_k = \left(\Delta_{\alpha}(k-1), \Delta_{\alpha}(k)\right]$ for $k=1,\ldots,m$. By the nature of step-up procedures, $i\in\cR^{\SU_\Delta(\alpha)}$ and $R^\alpha = \hR_i$ on the set $\{p_i \leq \Delta_\alpha(R^\alpha)\}$. Then we have for $P \in H_i$,
  \begin{align}
    \EE_P\left[\; \frac{1\left\{p_i \leq \tau_i(\alpha)\right\}}{\hR_i} \mid S_i\;\right]
    &\;=\; \EE_P\left[\; \frac{1\left\{p_i \leq \alpha \beta\left(R^{\alpha}\right)/m\right\}}{R^\alpha} \;\big|\; S_i\;\right]\\
    &\;=\; \sum_{k=1}^m \,\EE_P\left[\; \frac{1\left\{p_i \leq \alpha \beta\left(R^{\alpha}\right)/m\right\}}{R^\alpha} \,\cdot\, 1\{p_i \in I_k\} \;\big|\; S_i\;\right]\\
    \label{eq:cLm-2}
    &\;\leq\; \sum_{k=1}^m  \frac{1}{k}\, \PP_P\left[\; p_i \in I_k \;\big|\; S_i\;\right]\\
    \label{eq:cLm-3}
    &\;\leq\; \sum_{k=1}^m\frac{1}{k}\,\cdot\,\left(\Delta_\alpha(k)-\Delta_\alpha(k-1)\right)\\
    &\;=\; \frac{\alpha}{m} \sum_{k=1}^m \nu(\{k\}) \;=\; \frac{\alpha}{m}.
  \end{align}
  The inequality in \eqref{eq:cLm-2} follows from the fact that
  \[
  p_i \in I_k \;\text{and}\; p_i \leq \frac{\alpha \beta(R^\alpha)}{m} \;\;\Longrightarrow\;\; R^\alpha \geq k.
  \]
  Because $1/k$ is decreasing in $k$, the uniform distribution maximizes the sum in \eqref{eq:cLm-2} among all superuniform distributions, leading to the inequality in \eqref{eq:cLm-3}.

  ~\\{\bf \em Claim 3.} If we define $\tau^{\BY}(c) = \tau^{\BH}(c/L_m)$, a rescaled version of the effective BH threshold, then the claim follows as a special case of Claim 4, with  $\nu(\{k\}) = (k L_m)^{-1}$.
\end{proof}


\section{Further examples}\label{app:examples}

\subsection{Edge testing in Gaussian graphical models}

In the usual Gaussian graphical model (GGM) setting, we observe $X_1,\ldots,X_n \simiid N_d(\mu,\Sigma)$, and attempt to reconstruct the partial dependence graph $\cG$ of pairs $(j,k)$ for which $X_{ij}$ and $X_{ik}$ are not independent conditional on $X_{i,-(jk)}$. As a hypothesis testing problem, the edge $(j,k)$ is not present if and only if $\Theta_{ij} = 0$, where $\Theta = \Sigma^{-1}$; thus
\[
H_{jk}:\; \Theta_{ij} = 0 \Longleftrightarrow (i,j) \notin \cG.
\]

We begin by constructing the sample covariance matrix
\[
V(X) = \frac{1}{n-1} \sum_i (X_i-\overline X)(X_i - \overline X)', \quad \text{with } (n-1)V \sim \text{Wishart}(\Sigma,n-1).
\]

The Wishart distribution is an exponential family with complete sufficient statistic $T(X) = (\overline X, V)$, and the standard test for $H_{jk}$ is simply the $t$-test for the coefficient of $X_j$ in a multiple regression of $X_k$ on the other $d-1$ variables.

A homotopy algorithm for the Wishart problem is more complex than the homotopy algorithm for the other Gaussian-derived problems discussed above.

\subsection{Post-selection $z$- and $t$-tests}

Another potentially interesting application of our work is in post-selection multiple testing of regression coefficients after some . The post-selection distribution of regression coefficients follows a truncated multivariate Gaussian leading to post-selection $z$- and $t$-tests depending on whether the error variance is known or unknown, as investigated in various works including \citet{tibshirani2016exact,lee2016exact, fithian2014optimal, tian2018selective}. Because the post-selection distribution is a continuous exponential family, methods closely related to the ones discussed above may be used, with more computational effort. We leave full investigation of these examples to future work.

\subsection{Multiple comparisons to control for binary outcomes}\label{app:binomial_mcc}

Our methods extend to discrete as well as continuous models. For example, consider a clinical trial or A/B test with binary outcomes, in which $m$ different treatments are compared in the same experiment to a common control treatment, known as the {\em multiple comparisons to control} (MCC) problem in multiple testing. Let $n_i$ denote the number of experimental subjects in the $i$th treatment group, and let $X_i \sim \text{Binom}(n_i, \theta_i)$ denote the number whose binary response is positive. In addition, let $X_0 \sim \text{Binom}(n_0, \theta_0)$ denote the number of positive responses in the control group. Assume that $n_0,\ldots,n_m$ are fixed and known.

To test $H_i:\; \theta_i = \theta_0$ (or $H_i:\; \theta_i \leq \theta_0$), we can reject for extreme values (respectively large values) of the statistic $X_i$, whose null distribution is hypergeometric conditional on $S_i = (X_0 + X_i, X_{-i})$. The resulting $p$-values are correlated with each other through their common dependence on $X_0$, since larger values of $X_0$ shift the null distributions of all $X_i$ to the right, increasing the tests' critical values.

Conditional expectations given $S_i$ in this model can be evaluated exactly, since $X_i$ is conditionally supported on 
the finite set $\{0 \vee (X_0 + X_i - n_0),\ldots, n_i \wedge (X_0+X_i)\}$.

\subsection{Nonparametric multiple comparisons to control}
Conditionally simple models can arise in other contexts than exponential families; for example, consider a nonparametric one-way layout problem with real-valued observations:
\[
X_{ij} \simiid F_i, \quad \text{for } i = 0,\ldots,m, \;\; j=1,\ldots,n_i.
\]

Such a model might arise in an A/B testing context where we wish to compare each $F_i$ for $i \geq 1$ with a common control distribution $F_0$, a nonparametric version of the MCC problem. A complete sufficient statistic for the full model $\cP$ is the set of order statistics for each of the $m+1$ samples, or equivalently the empirical distributions of each sample $T(X) = (\hF_0, \ldots, \hF_m)$, where
\[
\hF_i(x) = \frac{1}{n_i} \sum_j 1\{X_{ij} \leq x\}.
\]

The null hypothesis $H_i:\; F_i = F_0$ defines a submodel with complete sufficient statistic 
\[
S_i = \left(\frac{n_0\hF_0 + n_i\hF_i}{n_0+n_i}, \;\hF_1, \ldots, \hF_{i-1},\hF_{i+1}, \ldots, \hF_m\right),
\]
or equivalently the pooled order statistics of sample $i$ and the control sample, as well as the separate order statistics of each of the other samples. Under $H_i$, every permutation of the order statistics is equally likely, and the $p$-value any two-sample permutation test of $H_i$ will be uniformly distributed on $\left\{\frac{1}{B+1}, \ldots, \frac{B}{B+1}, 1\right\}$ where $B$ is the number of random permutations used (ruling out ties).


\section{Algorithmic Details}\label{app:computation}

In this Appendix we discuss the algorithmic details as well as the computational tricks of both the homotopy algorithm and approximate numerical integration.

\subsection{Useful subclasses of $\dSU_{\gamma, \Delta}$}
For any threshold collection $\{\Delta_\alpha(r): r\in [m]\}$, the homotopy algortihm in Section \ref{sec:homotopy} can be applied to $\dSU_{\gamma, \Delta}$ with slight modification. Here we introduce several subclasses of sparse threshold collections that can reduce computational cost without losing nice theoretical guarantees. Given any integers $1 \le a_1 < a_2 < \cdots < a_L \le m$, define
\begin{equation}
  \label{eq:dSU}
  \Delta_\alpha(r) = \frac{\alpha\beta(r)}{m} = \frac{\alpha a_\ell}{m}, \quad (r \in [a_\ell, a_{\ell+1}), \,\,\ell = 0, 1, \ldots, L),
\end{equation}
where $a_0 = 0$ and $a_{L+1} = m + 1$ for convenience. Note that $\beta(r)\le r$. Using the same argument as in the proof of Theorem \ref{thm:safe} Claim 1 and 2, we can prove that $\dSU_{\gamma, \Delta}$ with \eqref{eq:dSU} controls \FDR ~at level $\alpha$ in finite samples when the p-values are independent or CPRDS. Further, we can define the safe version by setting $\gamma = 1 / L_{m, a}$ where
\[L_{m, a} = \sum_{\ell = 1}^{L}\frac{a_\ell - a_{\ell - 1}}{a_\ell}.\]
This is safe because $\beta(r) / L_{m, a}$ can be rewritten as $\sum_{i=1}^r i\nu (\{i\})$ where
\[\nu(\{a_\ell\}) =  (a_\ell - a_{\ell - 1}) / a_\ell L_{m, a}, \,\, \mbox{and} \,\,\nu(\{i\}) = 0, \,\, (i\not\in \{a_1, \ldots, a_L\}).\]
It is easy to verify that $\nu$ is a density function on $[m]$ and thus the proof of Theorem \ref{thm:safe} Claim 4 guarantees that $\dSU_{1 / L_{m, a}, \Delta}$ is safe. 

Indeed, \eqref{eq:dSU} includes various interesting cases.
\begin{itemize}
\item When $L = 1$ and $a_1 = 1$, \eqref{eq:dSU} recovers the Bonferroni correction. Moreover $L_{m, a} = 1$, implying that the Bonferroni correction is safe without correcting $\alpha$.
\item When $L = m$ and $a_\ell = \ell$, \eqref{eq:dSU} recovers $\BH(\alpha)$. In this case, $L_{m, a} = \sum_{\ell = 1}^{m}(1 / \ell)$ so the safe version recovers the $\BY(\alpha)$ procedure.
\item When $L = \lfloor \log_2 m\rfloor$ and $a_\ell = 2^\ell$, \eqref{eq:dSU} recovers the setting \eqref{eq:sparse_dSU} and $L_{m, a} = L / 2$.
\end{itemize}
Motivated by \eqref{eq:sparse_dSU}, we design a class of ``geometrically increasing'' integer sequences with simple analytical forms as follows:
\begin{equation}
  \label{eq:geom}
  a_\ell = \left\lceil \frac{\beta^{\ell - 1} - 1}{\beta - 1} + 1\right\rceil, \,\, \ell \in [L]\mbox{ where } L = \left\lfloor\frac{\log \{(\beta - 1)(m - 1) + 1\}}{\log \beta}\right\rfloor 
\end{equation}
It is easy to see that \eqref{eq:sparse_dSU} is a special case of \eqref{eq:geom} with $\beta = 2$. Since $a_1 = 1$, $\dSU_{1, \Delta}(\alpha)$ is never less powerful than the Bonferroni correction. 

As for power, $\dSU_{1, \Delta}(\alpha)$ is strictly dominated by $\dBH_{1}(\alpha)$ since $\beta(r)\le r$. However, $\dSU_{1 / L_{m, a}, \Delta}(\alpha)$ may be more powerful than $\dBH_{1 / L_{m}}(\alpha)$ when $L_{m, a}\le L_{m}$. For instance, $L_{m} = \log m + \cO(1)$ while $L_{m, a} = (\log m) ((\beta - 1) / \beta \log \beta) + \cO(1)$ for $\Delta$ defined in \eqref{eq:geom}. As long as $\beta > 1$, $L_{m, a} < L_{m}$ for sufficiently large $m$. Furthermore, it is easy to verify that the mapping $\beta \mapsto (\beta - 1) / \beta \log \beta$ is decreasing in $\beta$. Thus a higher $\beta$ would produce a smaller correction factor $L_{m, a}$.

As for computational efficiency, $\dSU_{\gamma, \Delta}(\alpha)$ may be faster than $\dBH_{\gamma}(\alpha)$, when the number of distinct positive thresholds is $L$ is much smaller than $m$, in which case the number of potential knots is reduced as discussed in Section \ref{sec:computation}. For the class \eqref{eq:geom}, the number of distinct thresholds is $\log m / \log \beta + \cO(1)$.

\subsection{Computation tricks}

\subsubsection{Efficient update of $R_i^{(c)}(t)$}
In principle, $R_{i}^{(c)}(t)$ can be recovered from $\{B_\ell^{(c)}(t): \ell = 0, \ldots, m\}$ by \eqref{eq:Rit_Bt}. However, this naive method involves a search with computational cost up to $m$ for each knot. Indeed, $R_i^{(c)}(t_k)$ can also be updated sequentially as follows:
\[R_i^{(c)}(t_k) = \left\{
    \begin{array}{ll}
      R_i^{(c)}(t_{k-1}) & (\mbox{if }r_k < R_i^{(c)}(t_{k-1}))\\
      R_i^{(c)}(t_{k-1}) & (\mbox{if }r_k \ge R_i^{(c)}(t_{k-1})\mbox{ and }B_{r_k}^{(c)}(t_k) \not = 0)\\
      r_k & (\mbox{if }r_k > R_i^{(c)}(t_{k-1})\mbox{ and }B_{r_k}^{(c)}(t_k) = 0)\\
      \max\left\{\ell < r_k: B^{(c)}_\ell(t) = 0\right\} & (\mbox{if }r_k = R_i^{(c)}(t_{k-1})\mbox{ and }B_{r_k}^{(c)}(t_k) = 0)
    \end{array}
  \right..\]
In all but the last scenario, the search cost is zero. Only when the value of $R_i^{(c)}(t)$ decreases, the search cost is nonzero and equal to $R_i^{(c)}(t^{-}) - R_i^{(c)}(t)$. For the independent and CPRDS case, $R_i^{(c)}(t)$ is strictly increasing and thus no search is needed at all. For other cases, we observe that in most cases $R_i^{(c)}(t)$ is strictly increasing or has occasional drops by a small amount. Therefore, the search cost of updating $R_i^{(c)}(t)$ is negligible.

\subsubsection{Q-value capping to reduce $\sum_i$}\label{app:qcap}
As discussed in Section \ref{sec:homotopy}, the first determinant of the computational cost is the size of $\sum_i$ in \eqref{eq:sumijr}, namely the number of hypotheses for which $g_i^{*}(q_i \semic S_i)$ needs to be evaluated. Intuitively, the maximal $\hat{c}_i$ cannot be much larger than $\alpha$. For $\dBH_1(\alpha)$ on independence p-values, $g_i^*(\alpha\mid S_i) = \alpha / m$ and thus $\hat{c}_i = \alpha$. For CPRDS cases, although $\hat{c}_i \ge \alpha$, we observed that it is always below $2\alpha$ in all our pilot numerical studies. Of course there is no theoretical guarantee that $\hat{c}_i \le 2\alpha$. Nonetheless, if we cap $\hat{c}_i$ at $2\alpha$, the $\dBH$ procedures still control \FDR ~in finite samples because this operation is equivalent to modifying $\tau_i(c; X)$ as $\tau_i(c\wedge 2\alpha; X)$ which is still non-decreasing in $c$ for all $X$. This trick excludes all hypotheses with q-values above $2\alpha$, without the need to compute $g_i^*$. As a result, the size of $\sum_i$ is reduced to $R_{\BH}(2\alpha)$, which is usually a few orders of magnitude smaller than $m$. Meanwhile, as we observed, it is typical that  $\hat{c}_i < 2\alpha$ and so this capping step does not lose power. 

\subsubsection{Screening to reduce $\sum_j$ and $\sum_r$}\label{app:screening}
For a given hypothesis $H_i$, we need to find $\cK_{i, j, r}$ defined in \eqref{eq:cK}. As discussed in Section \ref{sec:numerical_integration}, for one-sided testing, we can also reduce the range of the integral \eqref{eq:integral} from $\RR$ to a finite interval $[t_\lo, t_\hi]$ with a tiny approximation error $\alpha \epsilon / m$. Similarly, the range can be reduced into $[-t_\hi, -t_\lo]\cup [t_\lo, t_\hi]$ for two-sided testing considered in Section \ref{sec:examples}. For simplicity, we only discuss one-sided testing in this subsection and discuss a shortcut to handle two-sided testing in the next subsection. For this reason, we only need to find knots lying in this interval. Our goal is to find an efficient way to identify pairs $(j, r)$ for which $\cK_{i, j, r}\cap [t_\lo, t_\hi]$ is empty and to ignore them in the computation.

The idea is to compute the minimum $p_{j, \min}$ and maximum $p_{j, \max}$ of $p_j(t)$ over $[t_\lo, t_\hi]$ and to find all thresholds between $[p_{j, \min}, p_{j, \max}]$. As a result, those with no thresholds in the interval can be excluded directly, thereby reducing the size of $\sum_j$, while given $j$, the thresholds outside the interval can be excluded, thereby reducing the size of $\sum_r$. It can be implemented efficiently if $p_{j, \min}$ and $p_{j, \max}$ have analytical forms.

For the one-sided multivariate Gaussian testing problem, since $\eta_i$ is increasing and $\xi_{ij}$ is linear, the minimum and maximum are achieved at $t_\lo$ (resp. $t_\hi$) and $t_\hi$ (resp. $t_\lo$) if $\Sigma_{j, i} > 0$ (resp. $\Sigma_{j, i} < 0$). For short-ranged covariance structures like in the AR process, $\Sigma_{j, i}$ is tiny for most $j$'s. For such a $j$, $p_{j, \min}$ is very close to $p_{j, \max}$ and it is likely that $\cK_{i, j, r}\cap [t_\lo, t_\hi]$ is empty. So the screening step can adaptively remove the hypotheses with low correlation with $H_i$.

\subsubsection{A shortcut for two-sided testing}\label{app:two-sided}
For all examples considered in Section \ref{sec:examples}, $\xi_{ij}(t)$ is identical for one- and two-sided testing. For the latter, $f_j(t) = 2(1 - F(|t|))$ where $F$ is the marginal distribution function. In addition, since the effective range of the integral \eqref{eq:integral} can be reduced to $[-t_\hi, -t_\lo]\cup [t_\lo, t_\hi]$, it remains to compute
\begin{align*}
  &\cK_{i, j, r}\cap \lb [-t_\hi, -t_\lo]\cup [t_\lo, t_\hi]\rb = \left\{t\in [-t_\hi, -t_\lo]\cup [t_\lo, t_\hi]: \xi_{ij}(t) = \pm F^{-1}\lb 1 - \frac{cr}{2m}\rb\right\}.
\end{align*}
This can be written as the union of four sets $\cK_{i, j, r}^{++}\cup \cK_{i, j, r}^{+-}\cup \cK_{i, j, r}^{-+} \cup \cK_{i, j, r}^{--}$ where
\begin{align*}
  &\cK_{i, j, r}^{++} = \left\{t: t\in [t_\lo, t_\hi], \xi_{ij}(t) = F^{-1}\lb 1 - \frac{cr}{2m}\rb\right\}\\
  &\cK_{i, j, r}^{+-} = \left\{t: t\in [t_\lo, t_\hi], -\xi_{ij}(t) = F^{-1}\lb 1 - \frac{cr}{2m}\rb\right\}\\
  &\cK_{i, j, r}^{-+} = \left\{-t: t\in [t_\lo, t_\hi], \xi_{ij}(-t) = F^{-1}\lb 1 - \frac{cr}{2m}\rb\right\}\\
  &\cK_{i, j, r}^{--} = \left\{-t: t\in [t_\lo, t_\hi], -\xi_{ij}(-t) = F^{-1}\lb 1 - \frac{cr}{2m}\rb\right\}.
\end{align*}
Each of them has the same structure as in the one-sided testing counterpart, with $\xi_{ij}(t)$ replaced by $\xi_{ij}(t), \xi_{ij}(-t), -\xi_{ij}(t), -\xi_{ij}(-t)$, and we compute each of the four sets separately. 

\subsubsection{Screening for $\dBH^2_\gamma(\alpha)$}
Denote by $g_i^{(1)}$ and $g_i^{(2)}$ the conditional expectations $g_i^*(q_i \semic S_i)$ in $\dBH_\gamma(\alpha)$ and $\dBH^2_\gamma(\alpha)$, respectively. Note that $g_i^{(2)}$ is much more expensive to compute than $g_i^{(1)}$. When the procedure is safe, Theorem \ref{thm:recursive} guarantees that $g_i^{(2)} \le g_i^{(1)}$ almost surely. As a consequence, $g_i^{(2)}$ is below $\alpha / m$ whenever $g_i^{(1)}$ is and thus we can avoid computing $g_i^{(2)}$ for all rejected hypotheses by $\dBH_\gamma(\alpha)$.

When the procedure is not safe, through extensive numerical studies, we observed that $g_i^{(1)}$ and $g_i^{(2)}$ are typically not significantly different for any given $i$. Thus we can view $g_i^{(1)}$ as a proxy for $g_i^{(2)}$. Since it is only necessary to decide whether $g_i^{(2)}\le \alpha / m$, we avoid the computation if $g_i^{(1)}\le C\alpha / m$ for some constant $C < 1$. In particular, we choose $C = 0.9$ in our implementation as the default.

\subsection{Finding knots for multivariate t-statistics}\label{app:knots_t}
Recall the definitions of $\eta_i$ and $\xi_{ij}$ from Section \ref{sec:homotopy}. For multivariate t-statistics, as shown in Section \ref{sec:tstats},
\[\xi_{ij}(t) = a_{ij}\sqrt{n - d + t^2} + b_{ij}t, \,\, \mbox{where } a_{ij} = \frac{U_i}{\sqrt{V_i \Psi_{j, j}}}, \mbox{ and } b_{ij} = \frac{\Psi_{j,i}}{\Psi_{i, i}}.\]
Note that $\xi_{ij}(t), \xi_{ij}(-t), -\xi_{ij}(t), -\xi_{ij}(-t)$ all have this form. Recalling the definition of $\cK_{i, j, r}$ and the discussion in Appendix \ref{app:two-sided}, we need to solve equations in the following form (after transforming $t$ to $t\sqrt{n - d}$):
\begin{equation}
  \label{eq:quasi-quadratic}
  a\sqrt{1 + t^2} + bt = c.
\end{equation}
Moreover, if $\alpha < 0.5$, we only need to find \emph{positive} solutions in $[t_\lo, t_\hi]$ and we know that $c > 0$. Nonetheless, the solution of \eqref{eq:quasi-quadratic} is more complicated than it appears to be. Although it is attempting to solve the induced quadratic equation $a^2(1 + t^2) = (c - bt)^2$, the solution of the latter may not satisfy \eqref{eq:quasi-quadratic} since we need $\sign(c - bt) = \sign(a)$. Moreover, as shown in \eqref{eq:update_B} in Section
 \ref{sec:homotopy}, we also need to compute $\sign(p_j'(t))$. In this case,
 \[\sign(p_j'(t)) = \sign(\eta_j'(\xi_{ij}(t))\xi_{ij}'(t)) = \sign(\xi_{ij}'(t)).\]
With a generic form \eqref{eq:quasi-quadratic}, $\xi_{ij}'(t) = at /\sqrt{1 + t^2} + b$. 

In order to apply the screening step discussed in Appendix \ref{app:screening}, we need analytical formulae for the minimum and maximum of the function $t\mapsto a\sqrt{1 + t^2} + bt$.

\begin{proposition}\label{prop:screen_t}
  Write $m(t)$ for $a\sqrt{1 + t^2} + bt$. Given any $0 < t_\lo < t_\hi$, let $m_- = \min\{m(t_\lo), m(t_\hi)\}$ and $m_+ = \max\{m(t_\lo), m(t_\hi)\}$. 
  \begin{itemize}
  \item If $|a| \le |b|$ or $\sign(a) = \sign(b)$,
    \[\min_{t\in [t_\lo, t_\hi]}m(t) = m_-, \quad \max_{t\in [t_\lo, t_\hi]}m(t) = m_+.\]
  \item If $|a| > |b|$ and $\sign(a) = \sign(b)$. Let $t^* = \sqrt{b^2 / (a^2 - b^2)}$.
    \begin{itemize}
    \item If $t^* \not\in [t_\lo, t_\hi]$,
      \[\min_{t\in [t_\lo, t_\hi]}m(t) = m_-, \quad \max_{t\in [t_\lo, t_\hi]}m(t) = m_+.\]
    \item If $t^* \in [t_\lo, t_\hi]$,
      \[\min_{t\in [t_\lo, t_\hi]}m(t) = \min\{m_-, m(t^*)\}, \quad \max_{t\in [t_\lo, t_\hi]}m(t) = \max\{m_+, m(t^*)\}.\]
    \end{itemize}
  \end{itemize}
\end{proposition}

\begin{proof}
  Note that $m'(t) = at / \sqrt{1 + t^2} + b$. If $|a| \le |b|$,
  \[m'(t)\sign(b) \ge |b|\lb 1 - \frac{t}{\sqrt{1 + t^2}}\rb \ge 0.\]
  As a result, $m(t)$ is either non-decreasing or non-increasing. Thus the extremes are achieved at the boundaries. Similarly, if $\sign(a) = \sign(b)$, $m'(t)$ has the same sign with $b$ on $[0, \infty)$, implying that the extremes are also achieved at the boundaries.

  If $|a| > |b|$ and $\sign(a) \not= \sign(b)$, $m'(t^*) = 0$. Without loss of generality we assume $b > 0$ and $a < 0$. Then $m'(t) > 0$ for $t < t^*$ and $m'(t) < 0$ for $t > t^*$. Thus $m(t)$ is increasing on $[0, t^*]$ and decreasing on $[t^*, \infty)$. This proves the second case. The case with $b < 0$ and $a > 0$ can be proved similarly.
\end{proof}

The screening step guarantees that each equation of concern in the form of \eqref{eq:quasi-quadratic} has at least one positive solution. The following proposition provides neat analytical formulae for the solutions of \eqref{eq:quasi-quadratic} as well as $m'(t)$ for each solution. Albeit straightforward, it avoids unnecessary algebraic operations and thus is important for an efficient implementation of the homotopy algorithm. 

\begin{proposition}\label{prop:quasi-quadratic}
Assume that $c > 0$ and $m(t) = c$ has at least one positive root. 
 \begin{enumerate}[(1)]
  \item If $a = 0$, $m(t) = c$ has only one positive root $t_1 = \frac{c}{b}$ with $m'(t_1) > 0$.
  \item If $b = 0$, $m(t) = c$ has only one positive root $t_1 = \sqrt{\frac{c - a}{c}}$ with $m'(t_1) > 0$.
  \item If $a > 0$ and $b = \pm a$, $m(t) = c$ has only one positive root $t_1 = \frac{c^2 - b^2}{2bc}$ with $\sign(m'(t)) = \sign(b)$.    
  \item If $a < |b|$ or $a > b > 0$, $m(t) = c$ has only one positive root $t_1 = \frac{bc - \sign(b)a\sqrt{b^2 + c^2 - a^2})}{b^2 - a^2}$ with $\sign(m'(t_1)) = \sign(b)$.
  \item If $c > a > -b > 0$, $m(t) = c$ has only one positive root $t_1 = \frac{bc - a\sqrt{b^2 + c^2 - a^2}}{b^2 - a^2}$ with $m'(t_1) > 0$.
  \item If $a > -b > 0$ and $a\ge c$, $m(t) = c$ has two positive roots $t_1 = \frac{bc - a\sqrt{b^2 + c^2 - a^2}}{b^2 - a^2}$ and $t_2 = \frac{bc + a\sqrt{b^2 + c^2 - a^2}}{b^2 - a^2}$ with $m'(t_1) > 0$ and $m'(t_2) < 0$.
  \end{enumerate}
  Furthermore, these are all settings in which $m(t) = c$ can have at least one positive root.
\end{proposition}

\begin{proof}
  In case (1), $m(t) = c$ reduces to a linear function with root $c / b$. Since we assume $m(t) = c$ has at least one positive solution, it must be the positive solution and thus $b > 0$ and $m'(t_1) = b > 0$. The case (2) can be proved similarly. For the remaining cases, note that
  \begin{align}
    m(t) = c &\Longleftrightarrow a^2(1 + t^2) = (c - bt)^2 \mbox{ and }\sign(c - bt) = \sign(a)\nonumber\\
    & \Longleftrightarrow (b^2 - a^2)t^2 - 2bct + (c^2 - a^2) = 0 \mbox{ and }\frac{c - bt}{a}\ge 0.\label{eq:solution}
  \end{align}
  When $b^2 - a^2 = 0$, \eqref{eq:solution} reduces to a linear equation with solution $t = (c^2 - b^2) / 2bc$. Since we assume $m(t) = c$ has at least one positive solution, it must be the positive solution. Moreover, as shown in the proof of Proposition \ref{prop:screen_t}, when $|b| = |a|$, $\sign(m'(t)) = \sign(b)$. Thus, (3) is proved. 

  When $b^2 - a^2 \not = 0$, the first equation in \eqref{eq:solution} is quadratic. Thus, 
  \begin{align}
    m(t) = c &\Longleftrightarrow t = \frac{bc \pm a\sqrt{b^2 + c^2 - a^2}}{b^2 - a^2}  \mbox{ and }\frac{c - bt}{a}\ge 0.\nonumber
  \end{align}
  For each of the two potential solutions,
  \begin{align*}
    \frac{c - bt}{a} &= \frac{c(b^2 - a^2) - b(bc\pm a\sqrt{b^2 + c^2 - a^2})}{(b^2 - a^2)a} = \frac{-ca \mp b\sqrt{b^2 + c^2 - a^2}}{b^2 - a^2}\\
    & = \frac{b^2 + c^2}{ca \mp b\sqrt{b^2 + c^2 - a^2}}.
  \end{align*}
  Therefore,
  \begin{equation}
    \label{eq:solution2}
    m(t) = c \Longleftrightarrow t = \frac{bc \pm a\sqrt{b^2 + c^2 - a^2}}{b^2 - a^2}  \mbox{ and }ca \mp b\sqrt{b^2 + c^2 - a^2}\ge 0.
  \end{equation}
  First we prove the part for $m'(t)$ in case (4) -- (6). If $t = \frac{bc + a\sqrt{b^2 + c^2 - a^2}}{b^2 - a^2} > 0$ is the solution of \eqref{eq:solution2}, then $c - bt = a\sqrt{1 + t^2}$ and
  \begin{align*}
    \sign(m'(t)) & = \sign\lb \frac{at}{\sqrt{1 + t^2}} + b\rb = \sign(at + b\sqrt{1 + t^2})\\
                 & = \sign\lb at + \frac{b(c - bt)}{a}\rb = \sign(a) \sign((a^2 - b^2)t + bc)\\
    & = \sign(a)\sign(-a\sqrt{b^2 + c^2 - a^2}) = -1.
  \end{align*}
  Similarly, if $t = \frac{bc - a\sqrt{b^2 + c^2 - a^2}}{b^2 - a^2} > 0$ is the solution of \eqref{eq:solution2}, then $\sign(m'(t)) = 1$. This derivation covers case (4) -- (6).

  Next we compute the solutions \eqref{eq:solution2} in case (4) -- (6). We consider each case separately.
  \begin{itemize}
  \item If $a < |b|$,
  \begin{align*}
    (b^2 - a^2)(b^2 + c^2) > 0&\Longrightarrow b^2(b^2 + c^2 - a^2) > a^2c^2 \Longrightarrow ca - |b|\sqrt{b^2 + c^2 - a^2} < 0\\
                              & \Longrightarrow ca - \sign(b)b\sqrt{b^2 + c^2 - a^2} < 0\\
    & \Longrightarrow \frac{b^2 + c^2}{ca - \sign(b)b\sqrt{b^2 + c^2 - a^2}} < 0.
  \end{align*}
  This means $t = \frac{bc + a\sign(b)\sqrt{b^2 + c^2 - a^2}}{b^2 - a^2}$ does not satisfy the second condition of \eqref{eq:solution2}. By the assumption that $m(t) = c$ has at least one positive solution, $t = \frac{bc - a\sign(b)\sqrt{b^2 + c^2 - a^2}}{b^2 - a^2}$ must be the only positive solution. 
\item If $a > b > 0$, $\frac{bc + a\sqrt{b^2 + c^2 - a^2}}{b^2 - a^2} < 0$ and thus cannot be the positive solution of \eqref{eq:solution2}. By the assumption that $m(t) = c$ has at least one positive solution and \eqref{eq:solution2}, $t  = \frac{bc - a\sqrt{b^2 + c^2 - a^2}}{b^2 - a^2} = \frac{bc - a\sign(b)\sqrt{b^2 + c^2 - a^2}}{b^2 - a^2}$ must be the only positive solution. 
\item If $c > a > -b > 0$,
  \[\frac{bc + a\sqrt{b^2 + c^2 - a^2}}{b^2 - a^2} = \frac{c^2 - a^2}{bc - a\sqrt{b^2 + c^2 - a^2}} < 0.\]
  Similar to the last case, $t  = \frac{bc - a\sqrt{b^2 + c^2 - a^2}}{b^2 - a^2}$ must be the only positive solution.
\item If $a > -b > 0$ and $a \ge c$, we can easily verify that $\frac{bc \pm a\sqrt{b^2 + c^2 - a^2}}{b^2 - a^2}$ both satisfy \eqref{eq:solution2}. 
\end{itemize}

Finally, it is not hard to see that the only scenario that is not covered by case (1) -- (6) is that $a < 0$ and $b\le |a|$. In this case, $m(t) < a(\sqrt{1 + t^2} - t) < 0$ and thus $m(t) = c$ cannot have any solution.
\end{proof}



\section{Full simulation results}\label{app:full-simulations}

This section includes a fuller description of our simulation results. All experimental settings are listed below.
\begin{itemize}
\item Multivariate z-statistics drawn from $N(\mu, \Sigma)$ with $m = 1000$ and $10$ non-nulls in the top of the list with equal mean. We consider three types of covariance matrices:
  \begin{itemize}
  \item the AR$(0.8)$ process, i.e. $\Sigma_{ij} = (0.8)^{|i - j|}$. In this case, one-sided p-values are CPRD while two-sided p-values are not.
  \item the AR$(-0.8)$ process, i.e. $\Sigma_{ij} = (-0.8)^{|i - j|}$. In this case, neither one- nor two-sided p-values are not CPRD.
  \item the block dependent structure with $\Sigma_{ii} = 1$ and $\Sigma_{ij} = 0.5 \cdot 1(\lceil i / 20\rceil = \lceil j / 20\rceil)$. In this case, one-sided p-values are CPRD while two-sided p-values are not.
  \end{itemize}
\item Multivariate t-statistics with degree-of-freedom $n - d\in \{5, 50\}$. The null distribution is heavy-tailed for the former and is light-tailed for the latter. The z-statistics are drawn from $N(\mu, \Sigma)$ with $m = 100$ when $n - d = 5$ and $m = 1000$ when $n - d = 50$, and with $10$ non-nulls in the top of the list with equal mean. We consider three types of covariance matrices:
  \begin{itemize}
  \item the AR$(0.8)$ process, i.e. $\Sigma_{ij} = (0.8)^{|i - j|}$. In this case, one-sided p-values are CPRD while two-sided p-values are not.
  \item uncorrelated structure, i.e. $\Sigma_{ij} = 1(i\not = j)$. In this case, both one- and two-sided p-values are CPRD.
  \item the block dependent structure with $\Sigma_{ii} = 1$ and $\Sigma_{ij} = 0.5 \cdot 1(\lceil i / 20\rceil = \lceil j / 20\rceil)$. In this case, one-sided p-values are CPRD while two-sided p-values are not.
  \end{itemize}
\item Fixed-design homoscedastic Gaussian linear models with $n = 3000$ and $d = 1000$. The first $10$ coefficients are set to be non-zero with a equal size and the intercept is set to be $0$. As with Section \ref{sec:simulations}, we only consider two-sided testing with $\alpha\in \{0.05, 0.2\}$ for a fair comparison with the fixed-X knockoffs. The design matrix is generated as a realization of a random Gaussian matrix with i.i.d. rows drawn from $N(0, \Sigma)$, where $\Sigma$ takes one of the three form as in the multivariate-t case. The p-values are not CPRD in any case.
\item Multiple comparisons to control for Gaussian outcomes with $m = 100$ groups with either $3$ or $30$ replicates in each group. In both cases, we set the first $30$ groups as non-nulls with equal effect sizes. The p-values are not CPRD in any case.
\end{itemize}
For all above settings, the signal strength is tuned such that $\BH(0.05)$ has approximately $30\%$ power through a separate Monte-Carlo simulation. As with Section \ref{sec:simulations}, the \FDR ~and power are estimated on $1000$ independent simulations for each setting. Apart from the methods considered in Section \ref{sec:simulations}, we also include their sparse counterparts with threshold collection defined in \eqref{eq:sparse_dSU}, denoted by s-BH, s-dBH, s-dBH$^2$, s-BY, s-dBY, and s-dBY$^2$. For all CPRD cases we take $\gamma = 1$ in dBH, dBH$^2$, s-dBH and s-dBH$^2$, while in all other cases $\gamma$ is set to be $0.9$. For the knockoffs method, we generate the knockoff matrix using both the equicorrelated and semi-definite programming-based constructions. As in Section \ref{sec:simulations}, the intercept term is not included for the linear models but is included for the multiple comparisons to control for knockoffs. All experimental results are qualitatively the same as those in Section \ref{sec:simulations}. 

\newpage
\subsection{Testing on multivariate z-statistics}
\begin{figure}[H]
  \centering
  \begin{subfigure}[t]{0.48\textwidth}
    \includegraphics[width = 0.9\textwidth]{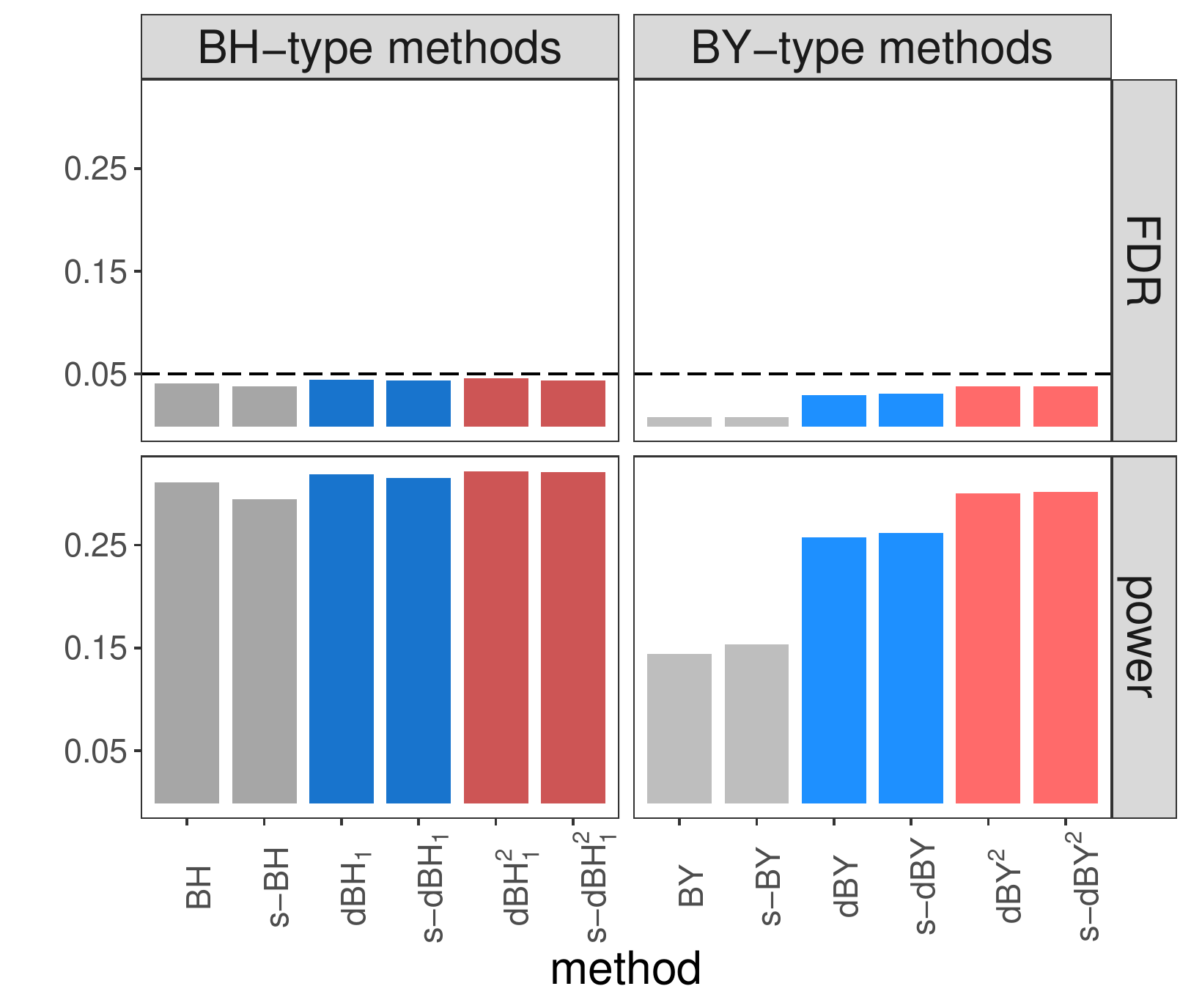}
    \caption{One-sided testing}
  \end{subfigure}
  \begin{subfigure}[t]{0.48\textwidth}
    \includegraphics[width = 0.9\textwidth]{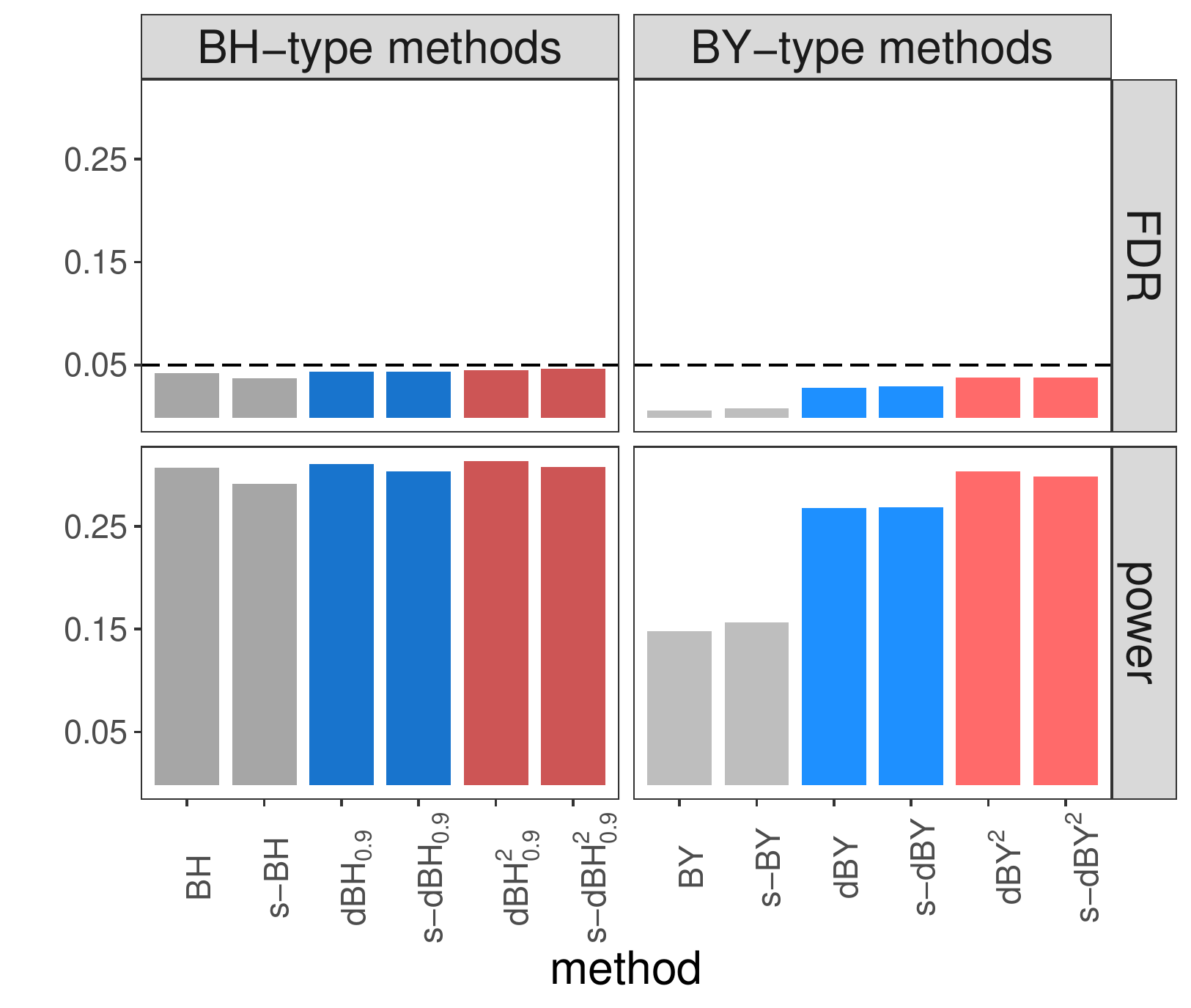}
    \caption{Two-sided testing}
  \end{subfigure}
  \caption{Multivariate z-statistics with AR$(0.8)$ covariance structure.}
\end{figure}

\begin{figure}[H]
  \centering
  \begin{subfigure}[t]{0.48\textwidth}
    \includegraphics[width = 0.9\textwidth]{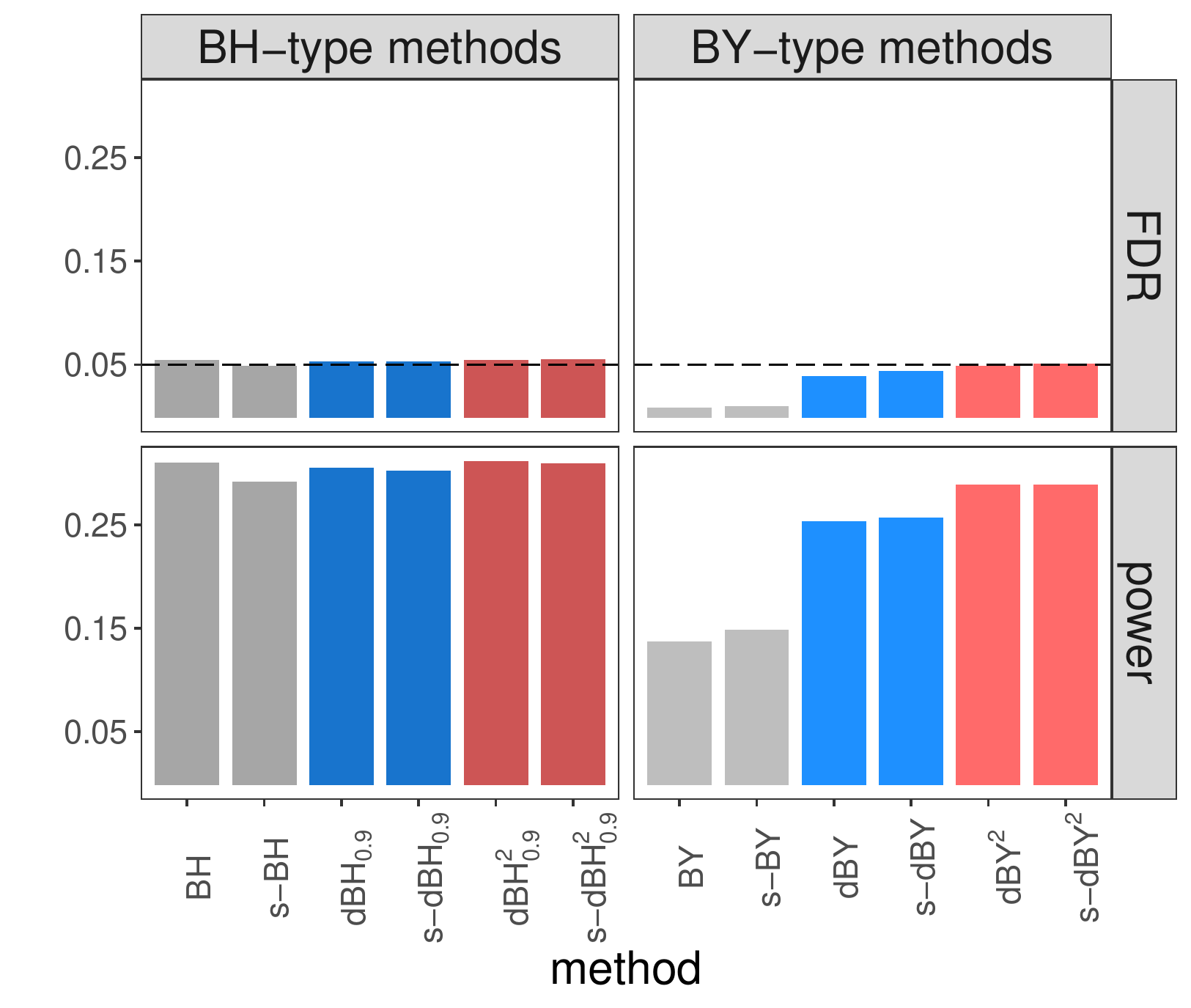}
    \caption{One-sided testing}
  \end{subfigure}
  \begin{subfigure}[t]{0.48\textwidth}
    \includegraphics[width = 0.9\textwidth]{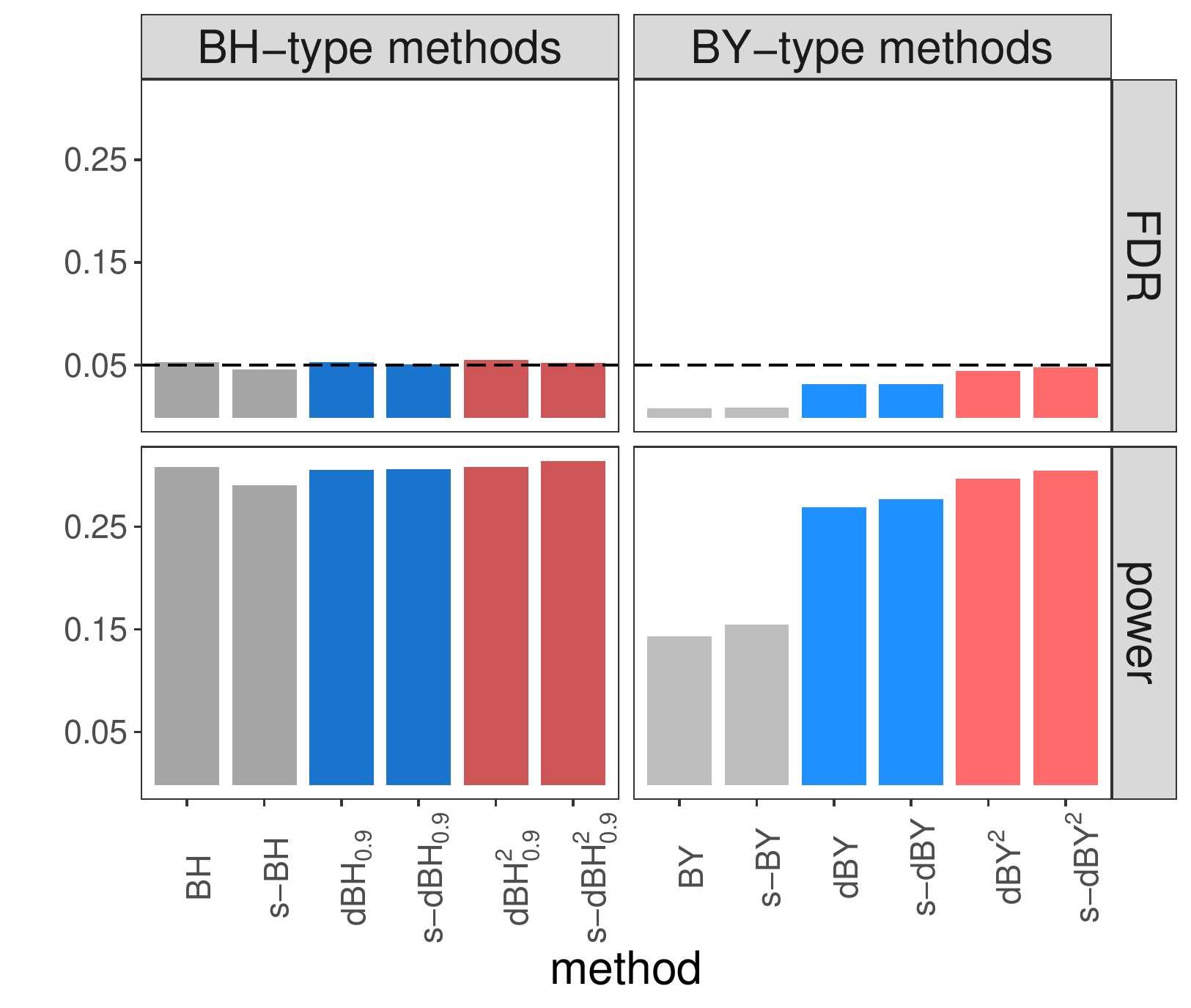}
    \caption{Two-sided testing}
  \end{subfigure}
  \caption{Multivariate z-statistics with AR$(-0.8)$ covariance structure.}
\end{figure}

\begin{figure}[H]
  \centering
  \begin{subfigure}[t]{0.48\textwidth}
    \includegraphics[width = 0.9\textwidth]{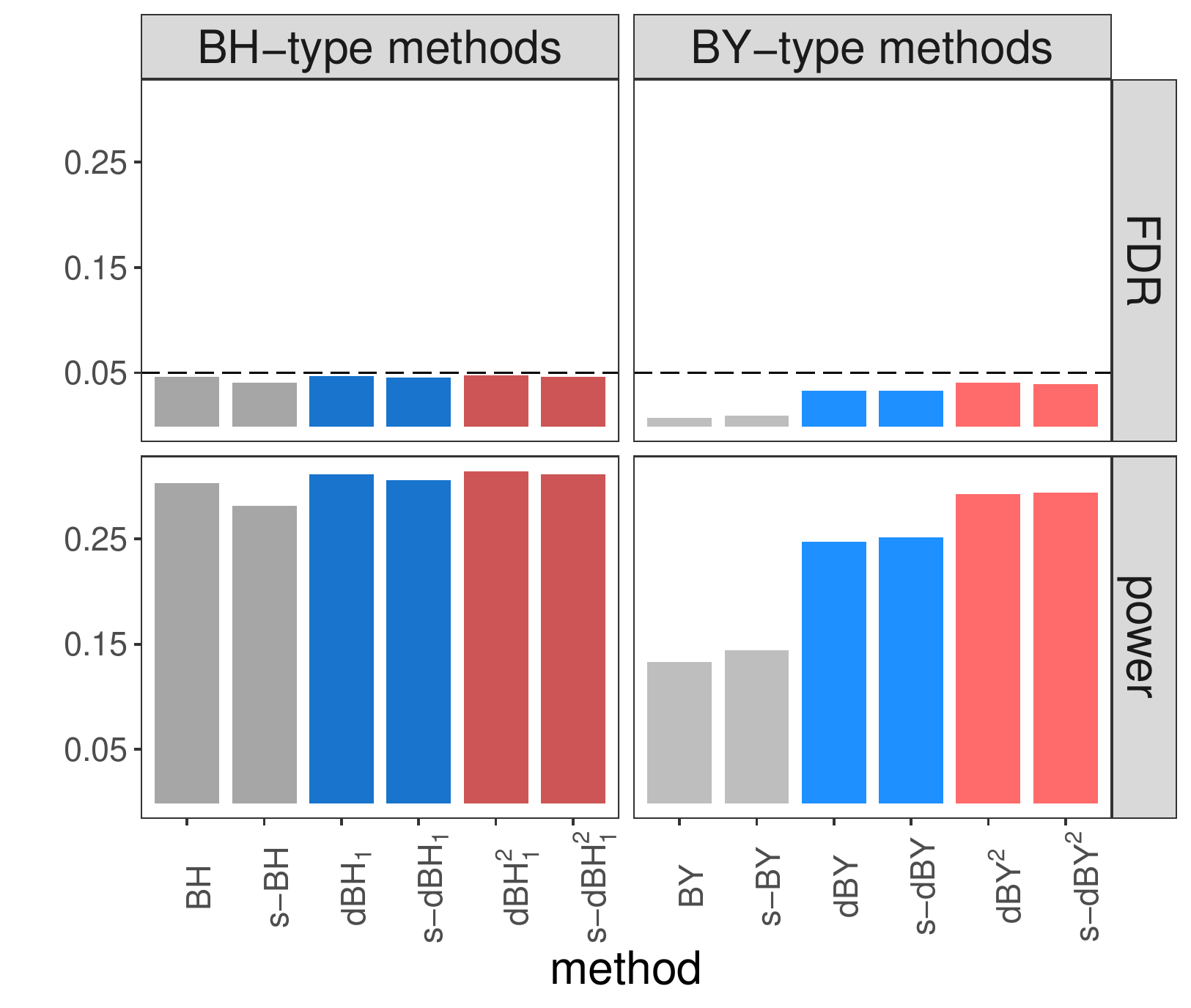}
    \caption{One-sided testing}
  \end{subfigure}
  \begin{subfigure}[t]{0.48\textwidth}
    \includegraphics[width = 0.9\textwidth]{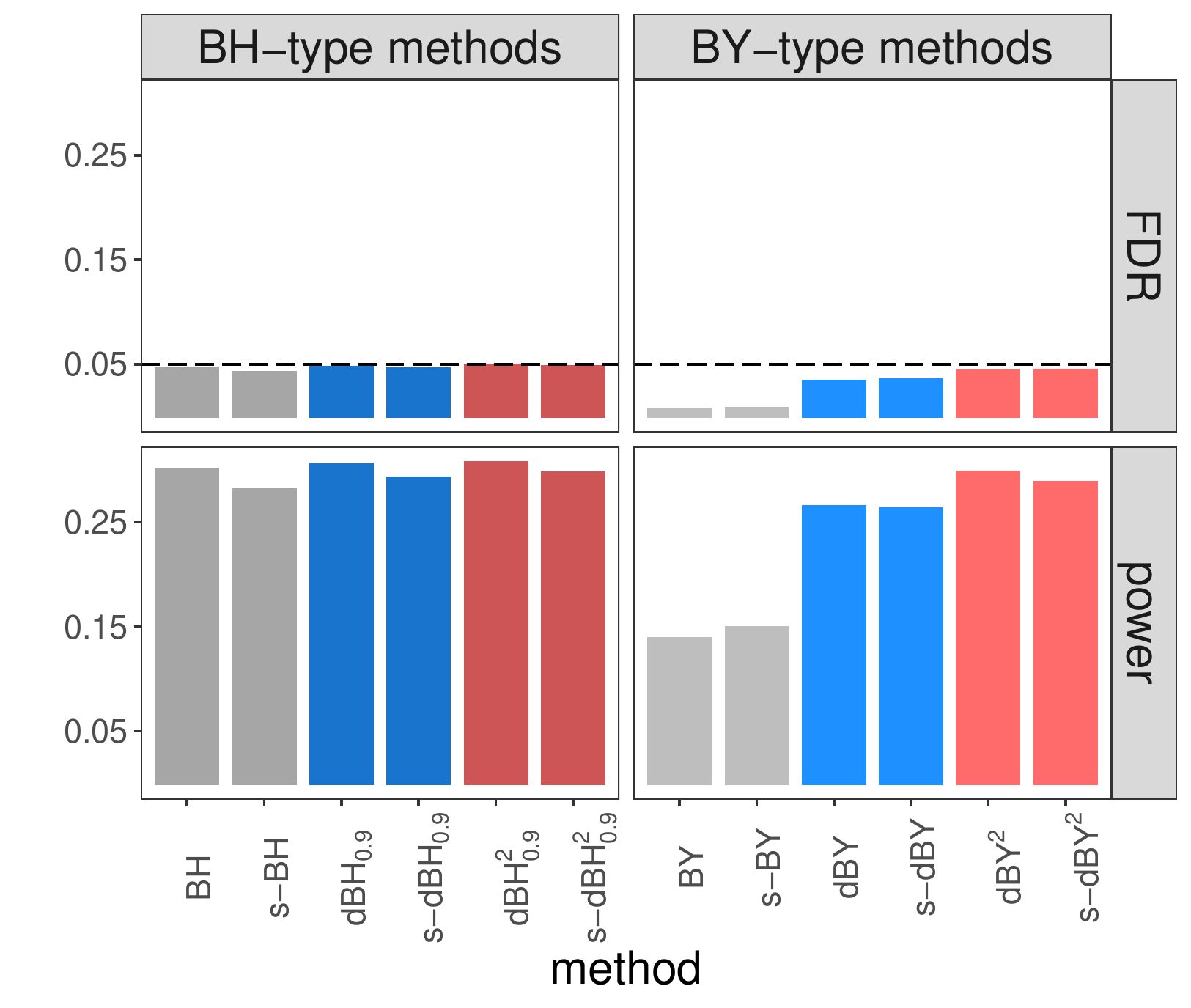}
    \caption{Two-sided testing}
  \end{subfigure}
  \caption{Multivariate z-statistics with block covariance structure.}
\end{figure}

\subsection{Testing on heavy-tailed multivariate t-statistics}
\begin{figure}[H]
  \centering
  \begin{subfigure}[t]{0.48\textwidth}
    \includegraphics[width = 0.9\textwidth]{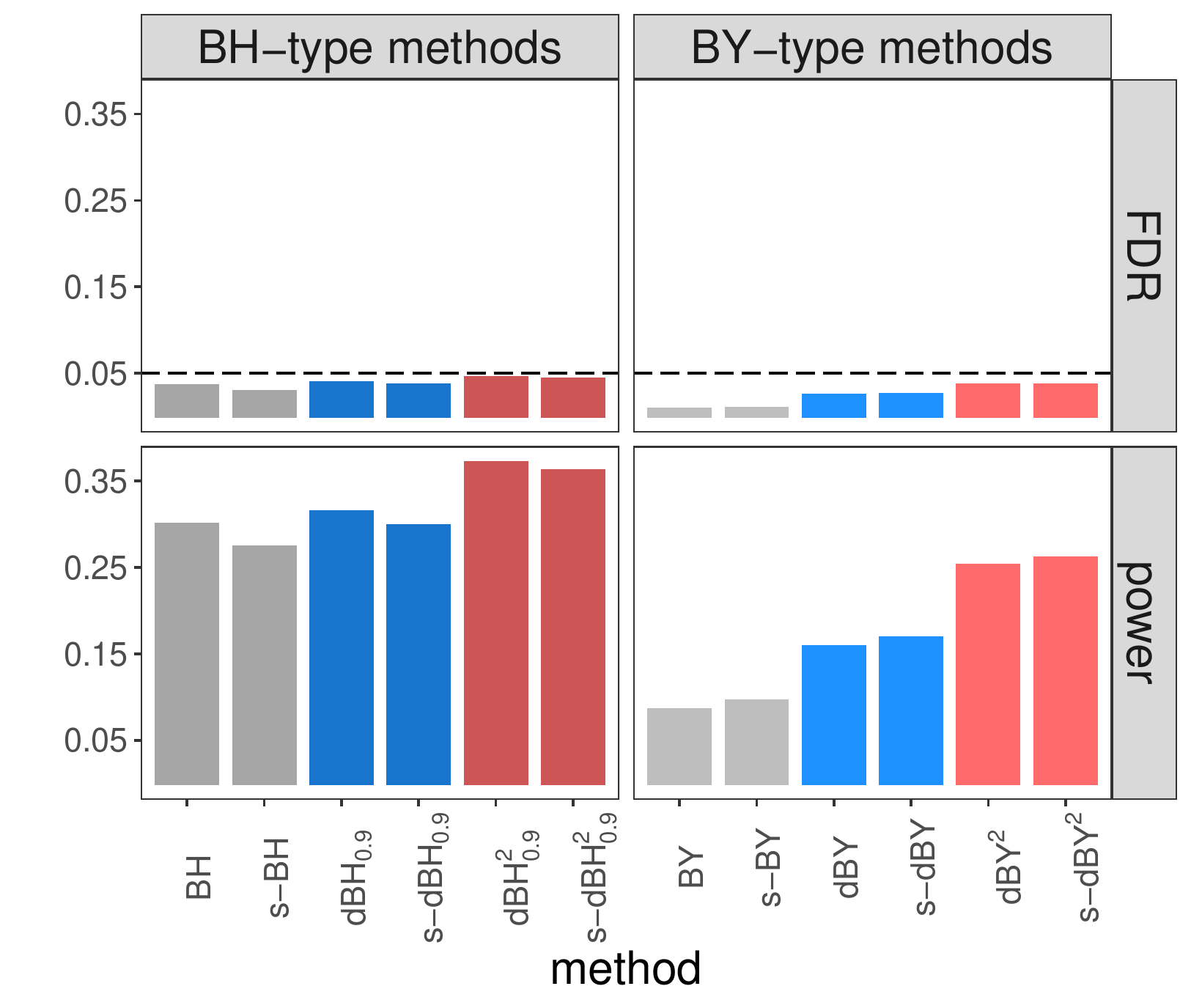}
    \caption{One-sided testing}
  \end{subfigure}
  \begin{subfigure}[t]{0.48\textwidth}
    \includegraphics[width = 0.9\textwidth]{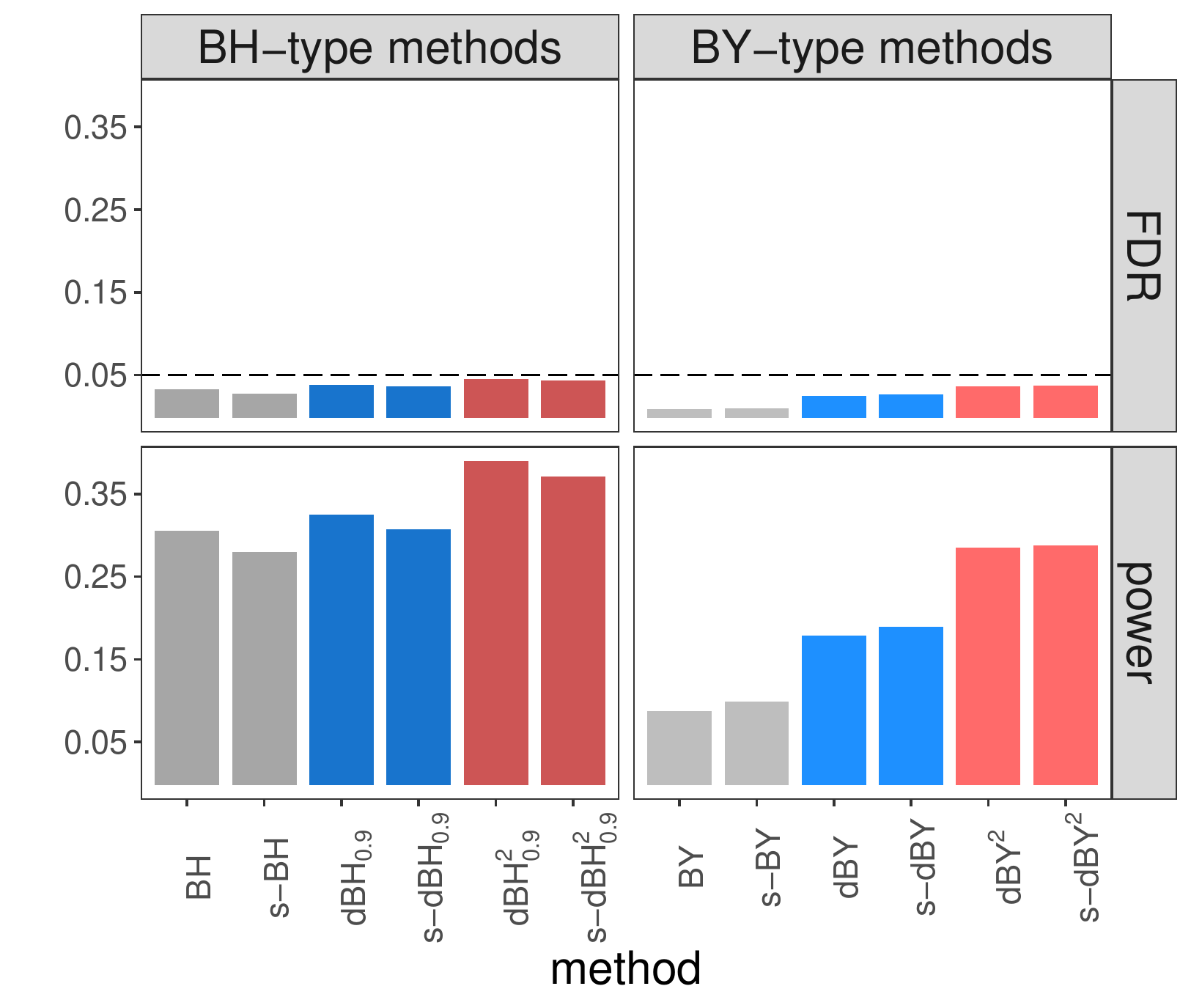}
    \caption{Two-sided testing}
  \end{subfigure}
  \caption{Heavy-tailed multivariate t-statistics with AR$(0.8)$ z-statistics.}
\end{figure}

\begin{figure}[H]
  \centering
  \begin{subfigure}[t]{0.48\textwidth}
    \includegraphics[width = 0.9\textwidth]{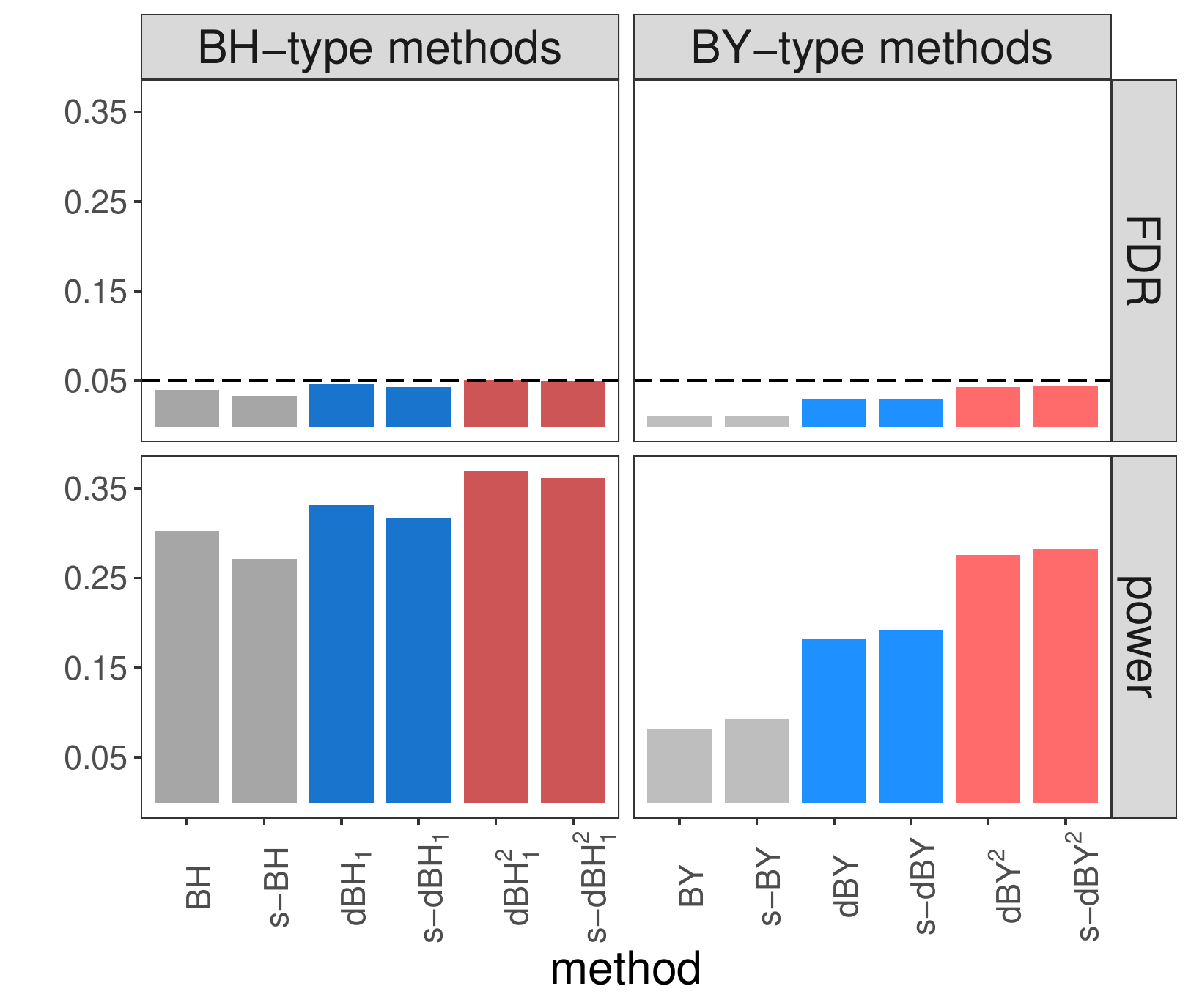}
    \caption{One-sided testing}
  \end{subfigure}
  \begin{subfigure}[t]{0.48\textwidth}
    \includegraphics[width = 0.9\textwidth]{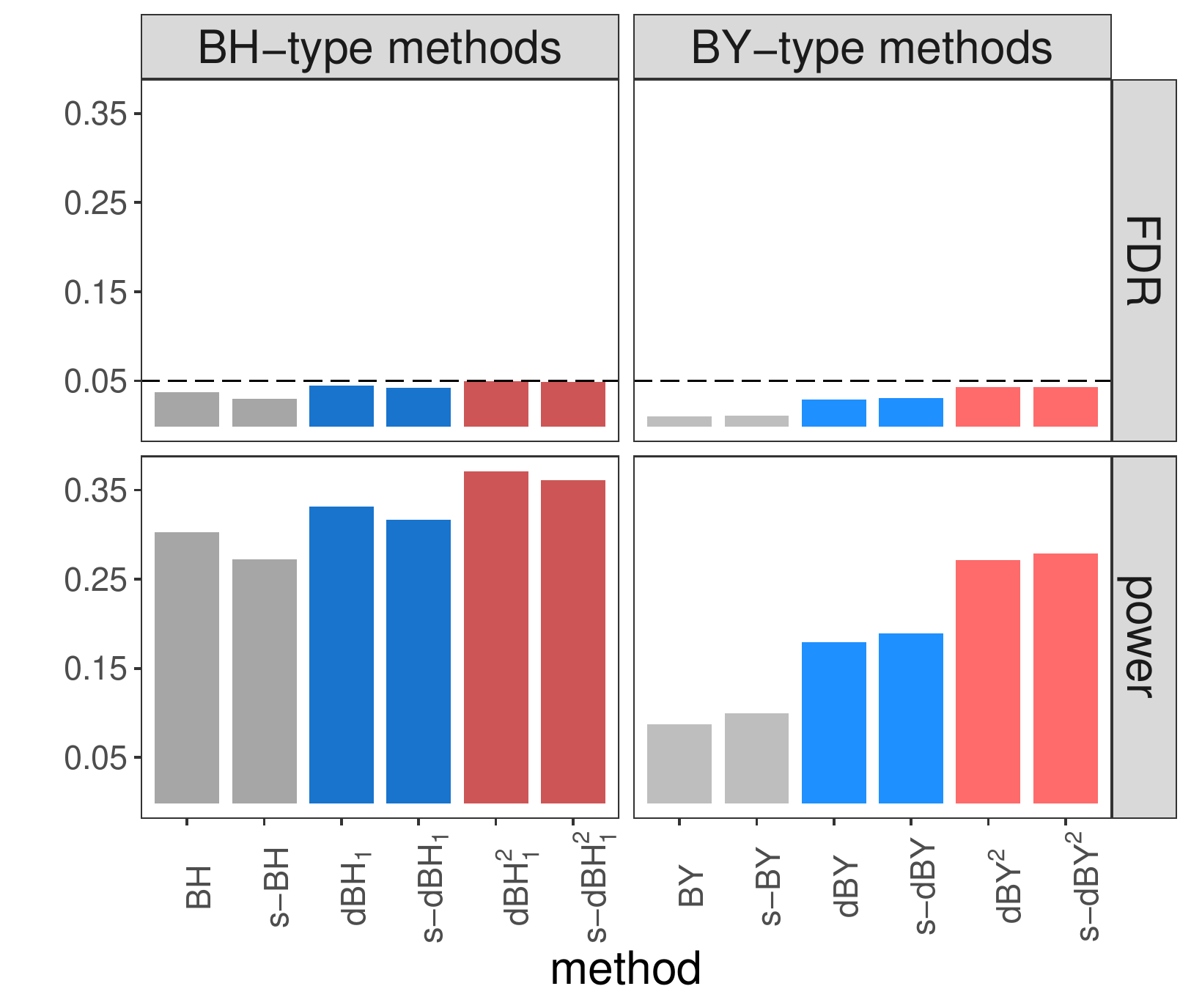}
    \caption{Two-sided testing}
  \end{subfigure}
  \caption{Heavy-tailed multivariate t-statistics with uncorrelated z-statistics. }
\end{figure}

\begin{figure}[H]
  \centering
  \begin{subfigure}[t]{0.48\textwidth}
    \includegraphics[width = 0.9\textwidth]{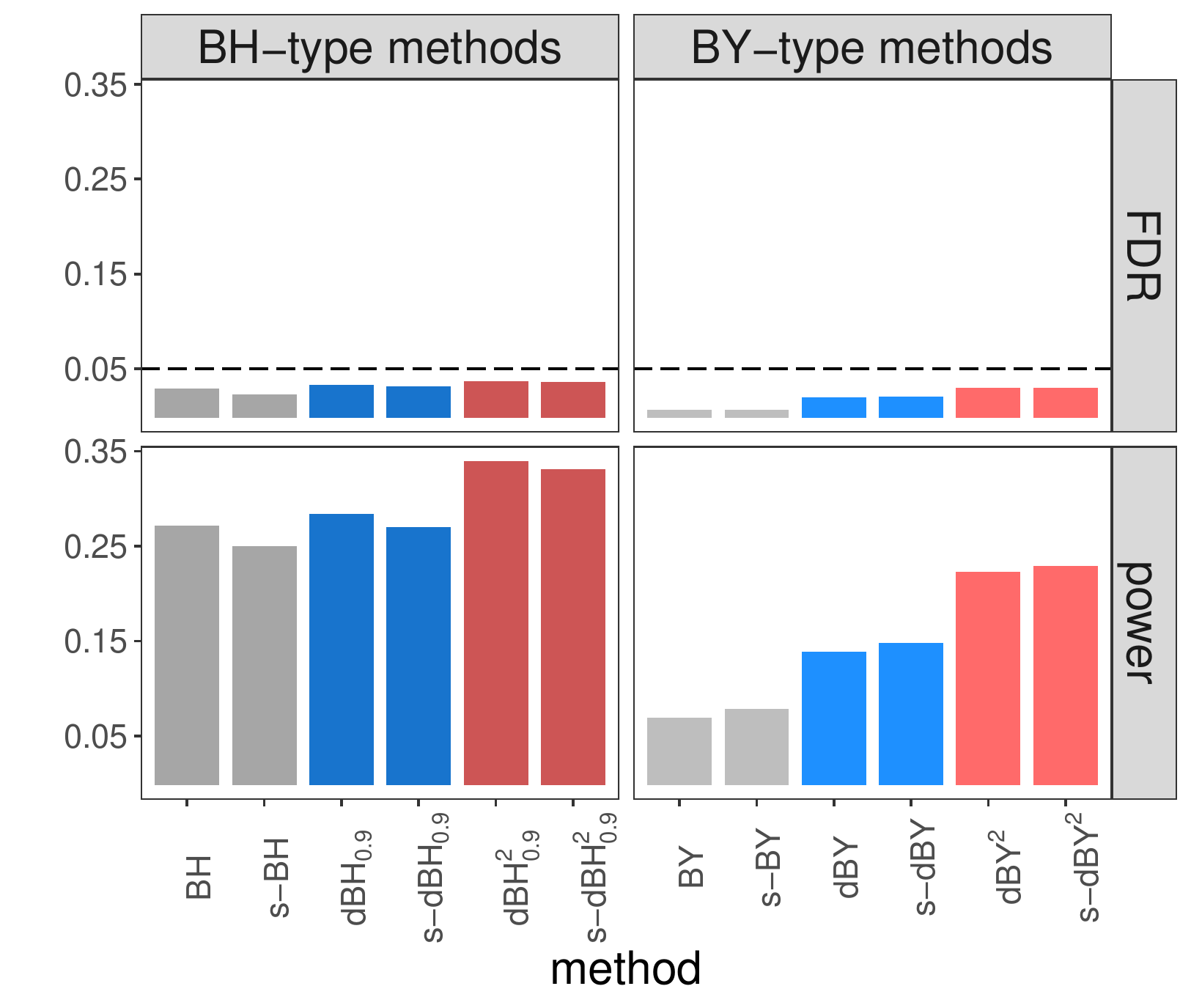}
    \caption{One-sided testing}
  \end{subfigure}
  \begin{subfigure}[t]{0.48\textwidth}
    \includegraphics[width = 0.9\textwidth]{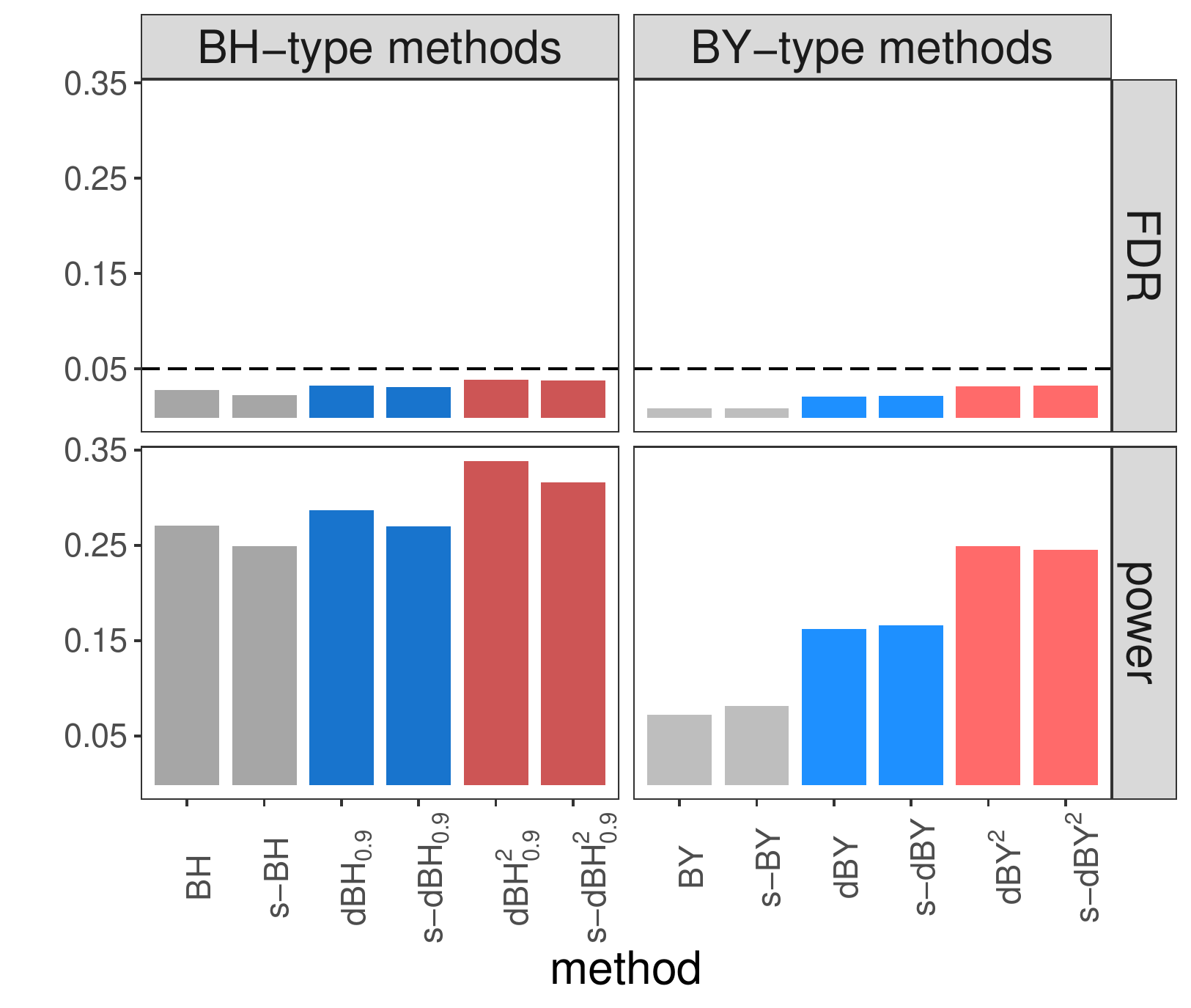}
    \caption{Two-sided testing}
  \end{subfigure}
  \caption{Heavy-tailed multivariate t-statistics with block dependent z-statistics.}
\end{figure}

\subsection{Testing on light-tailed multivariate t-statistics}
\begin{figure}[H]
  \centering
  \begin{subfigure}[t]{0.48\textwidth}
    \includegraphics[width = 0.9\textwidth]{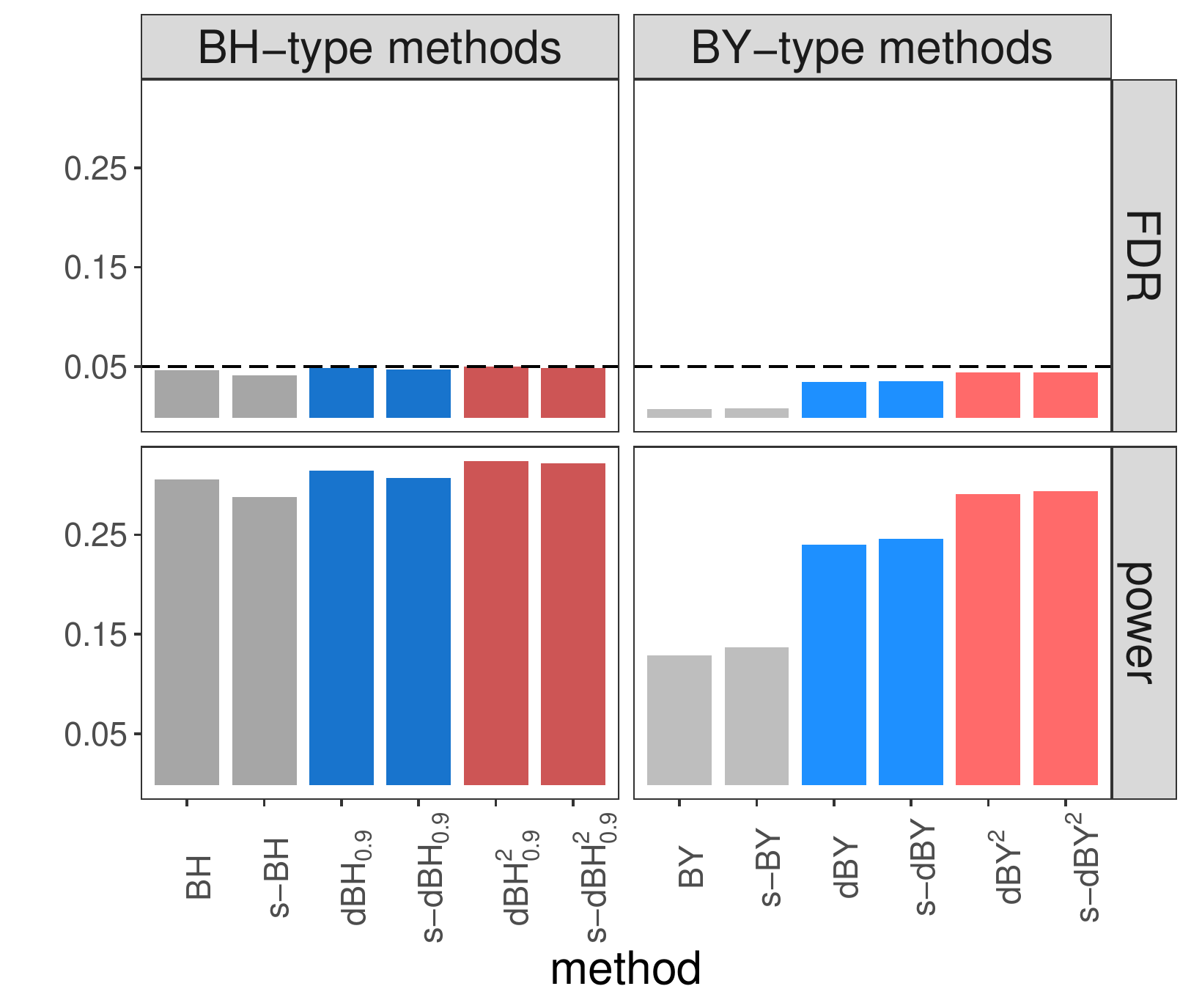}
    \caption{One-sided testing}
  \end{subfigure}
  \begin{subfigure}[t]{0.48\textwidth}
    \includegraphics[width = 0.9\textwidth]{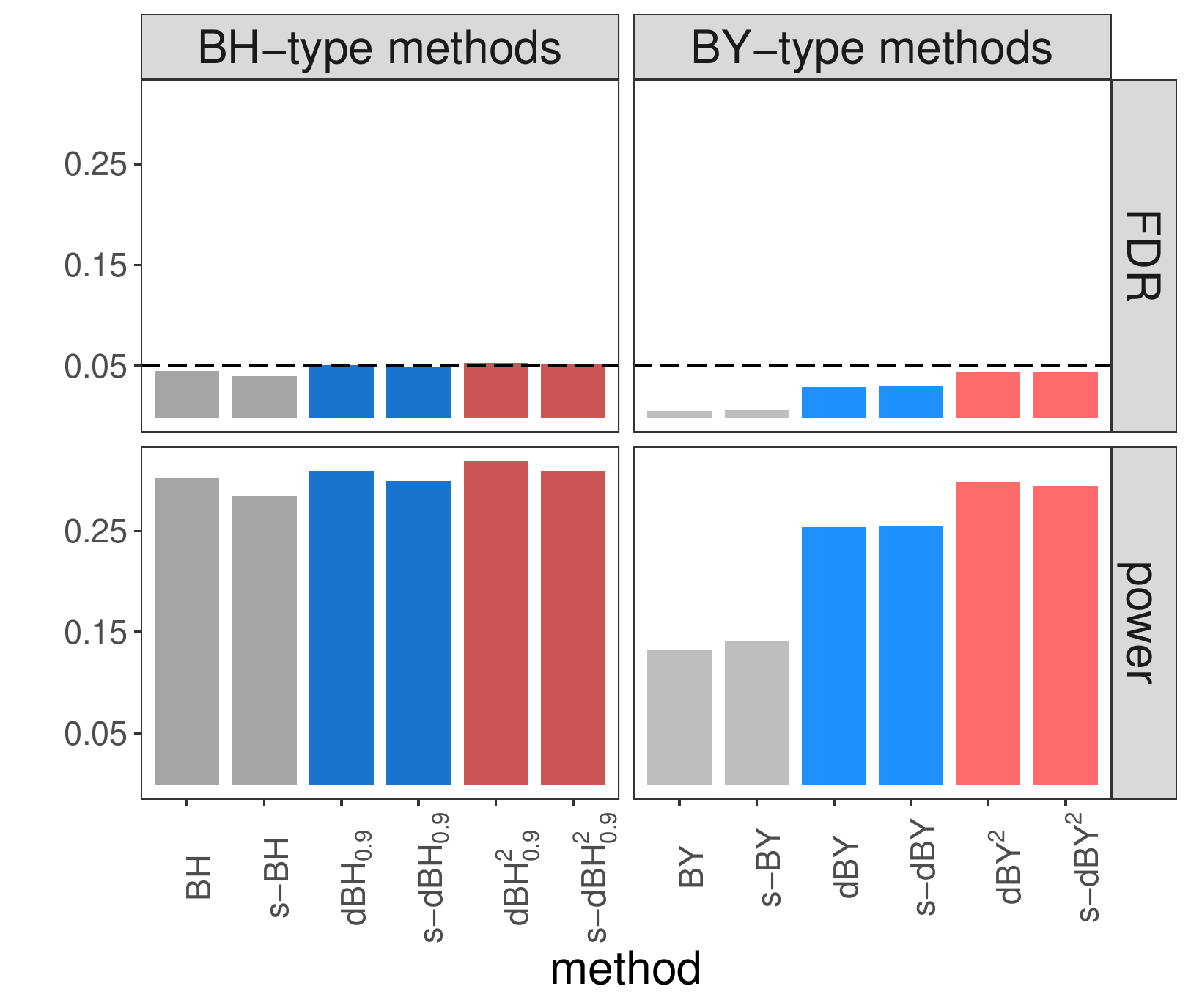}
    \caption{Two-sided testing}
  \end{subfigure}
  \caption{Light-tailed multivariate t-statistics with AR$(0.8)$ z-statistics}
\end{figure}

\begin{figure}[H]
  \centering
  \begin{subfigure}[t]{0.48\textwidth}
    \includegraphics[width = 0.9\textwidth]{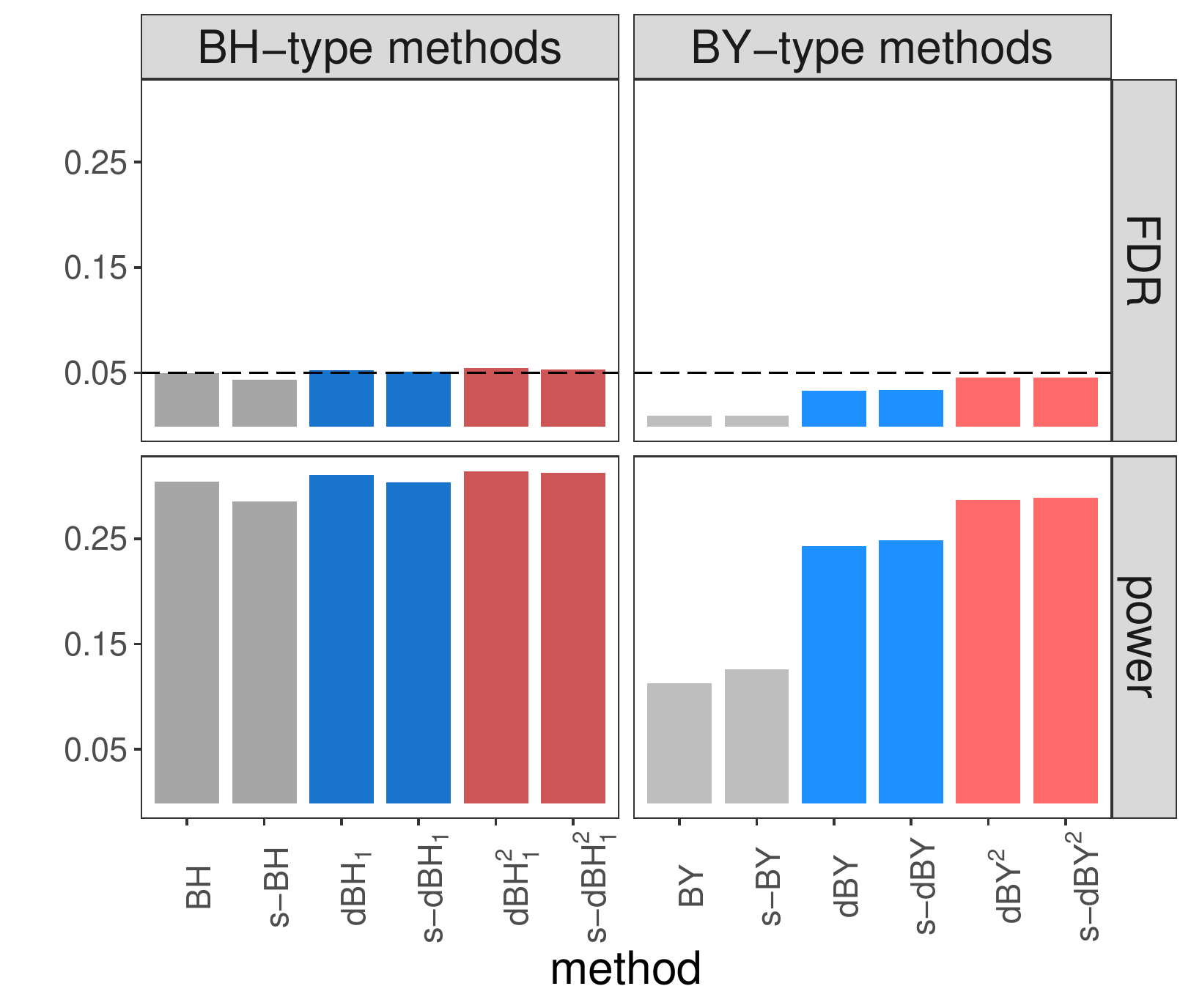}
    \caption{One-sided testing}
  \end{subfigure}
  \begin{subfigure}[t]{0.48\textwidth}
    \includegraphics[width = 0.9\textwidth]{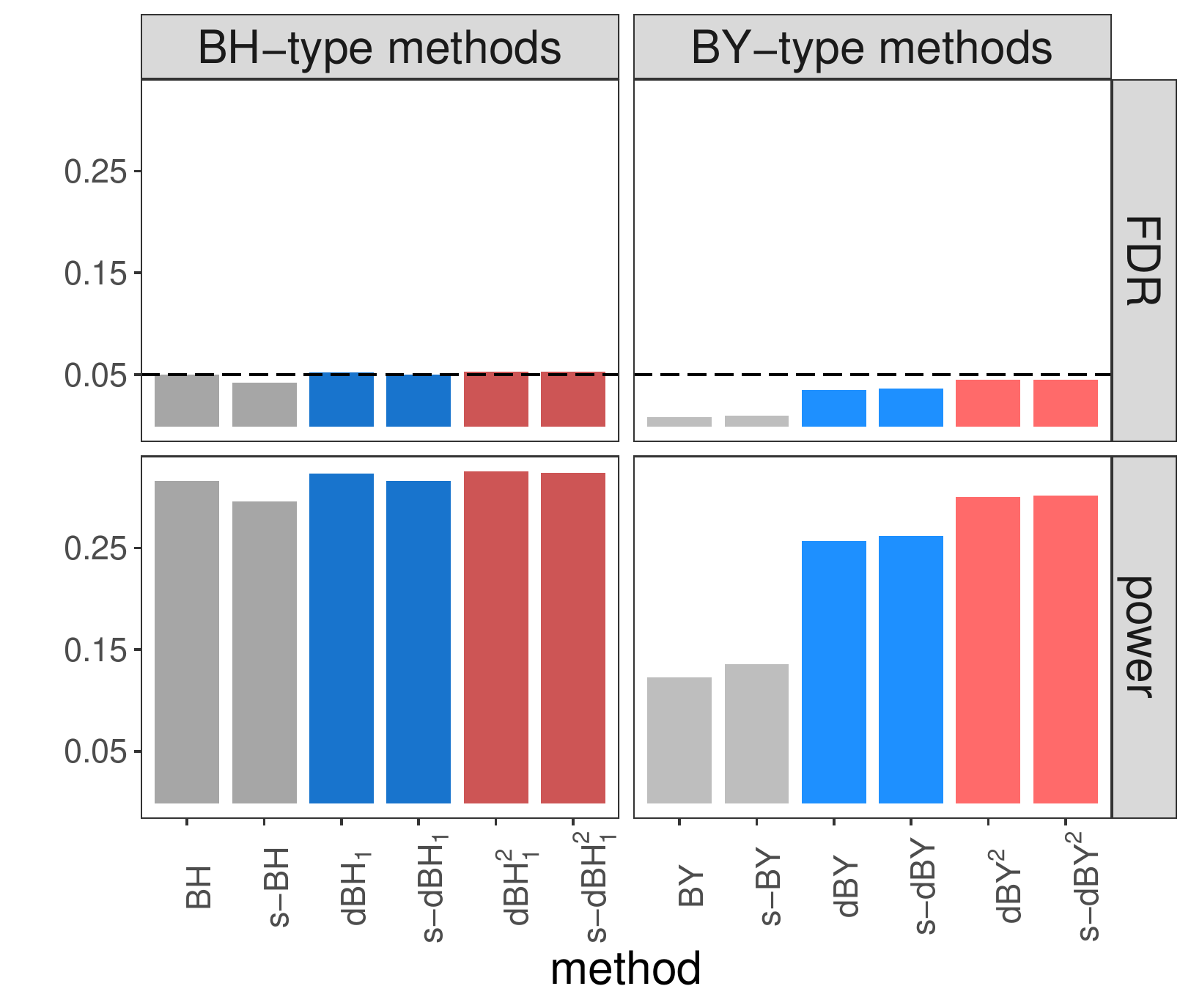}
    \caption{Two-sided testing}
  \end{subfigure}
  \caption{Light-tailed multivariate t-statistics with uncorrelated z-statistics}
\end{figure}

\begin{figure}[H]
  \centering
  \begin{subfigure}[t]{0.48\textwidth}
    \includegraphics[width = 0.9\textwidth]{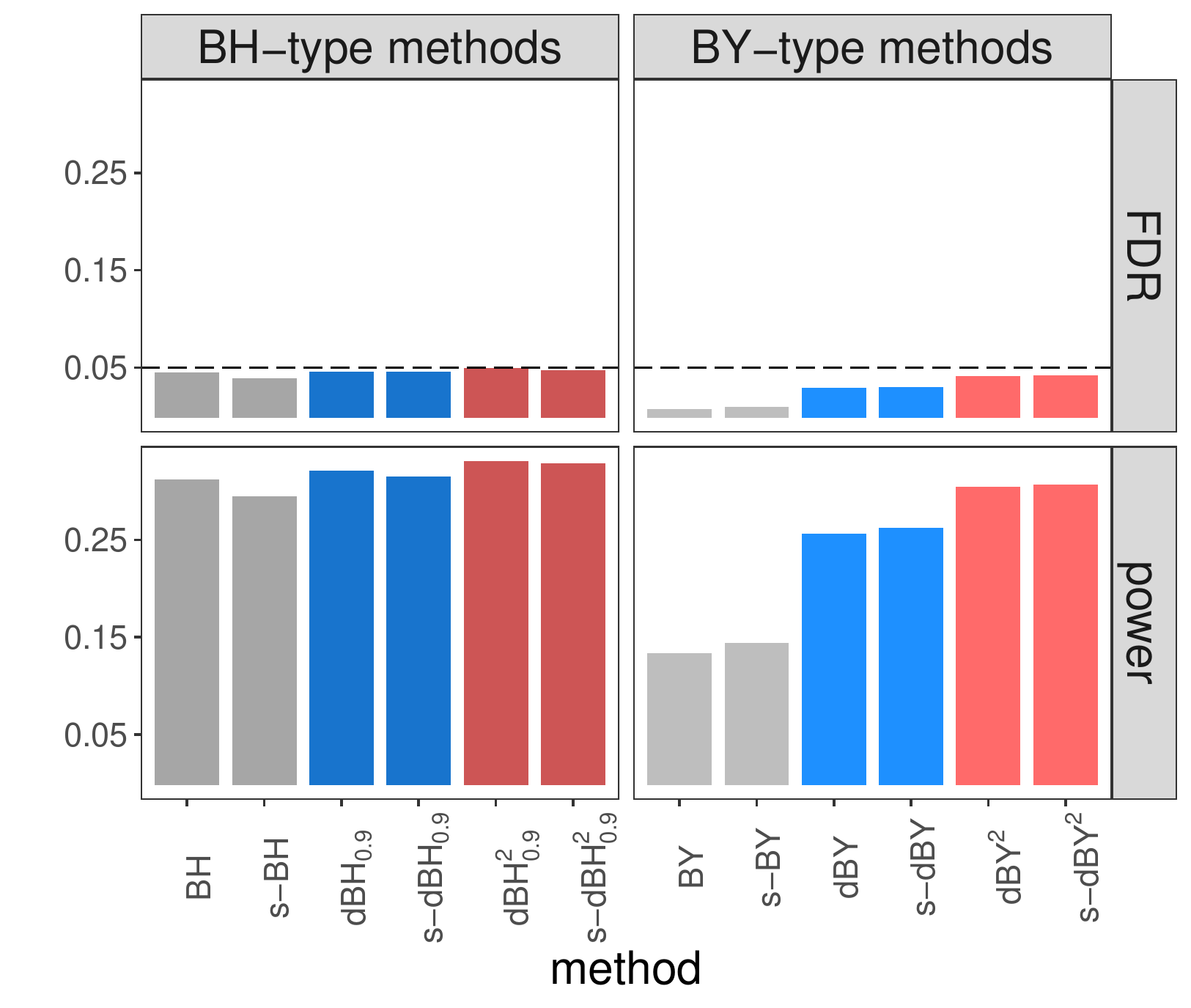}
    \caption{One-sided testing}
  \end{subfigure}
  \begin{subfigure}[t]{0.48\textwidth}
    \includegraphics[width = 0.9\textwidth]{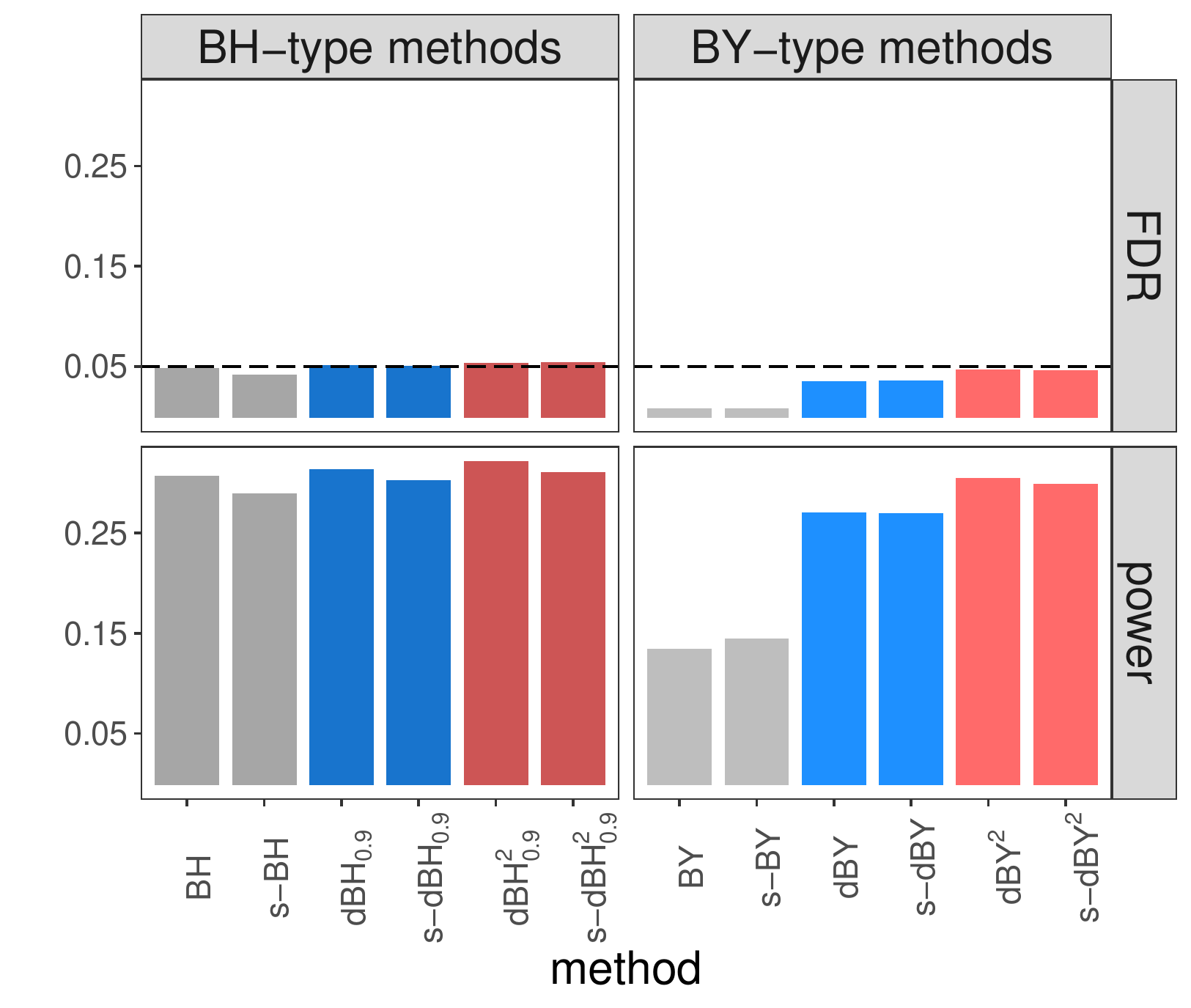}
    \caption{Two-sided testing}
  \end{subfigure}
  \caption{Light-tailed multivariate t-statistics with block dependent z-statistics}
\end{figure}

\subsection{Testing on fixed-design homoscedastic Gaussian linear models}
\begin{figure}[H]
  \centering
  \begin{subfigure}[t]{0.48\textwidth}
    \includegraphics[width = 0.9\textwidth]{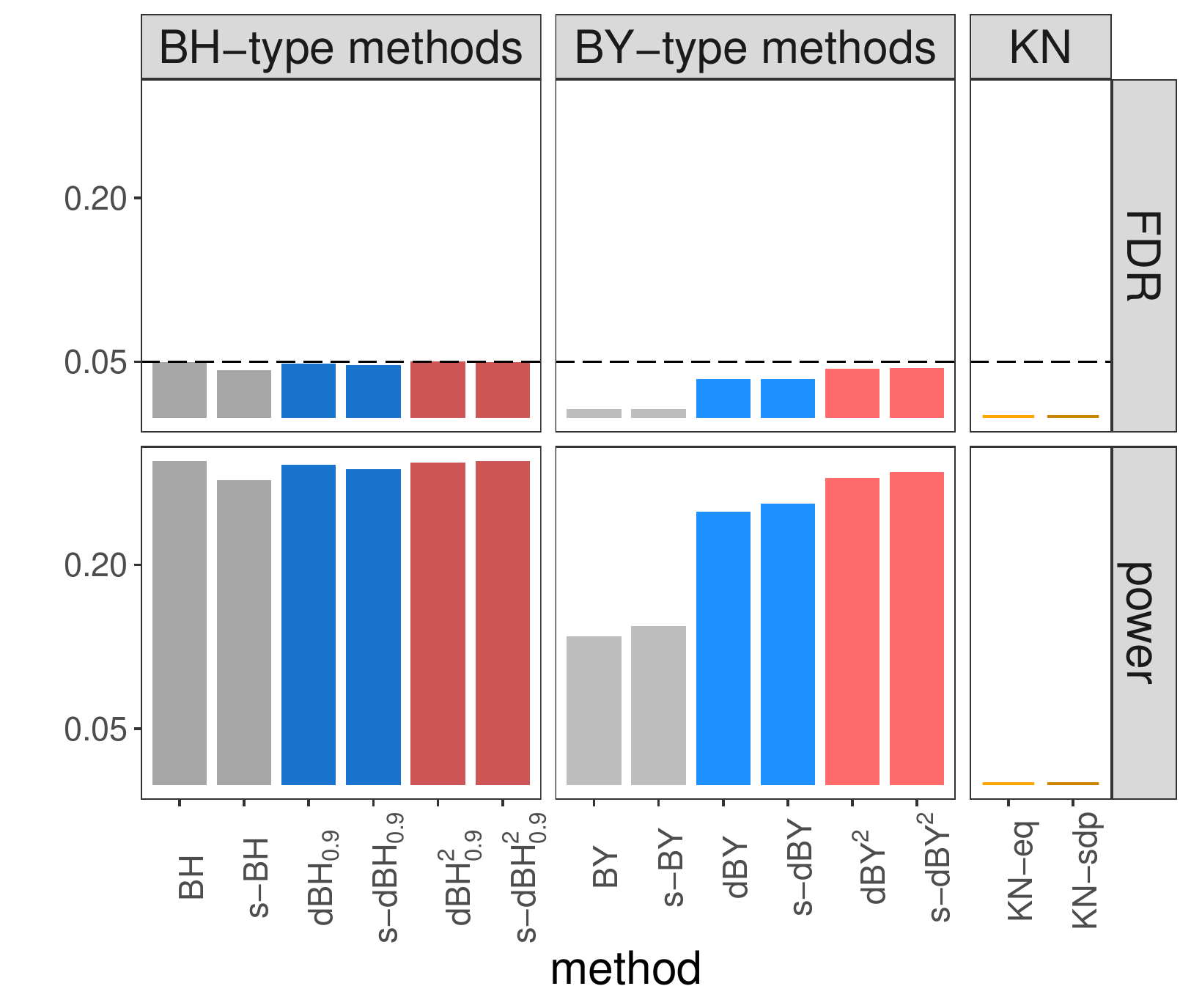}
    \caption{$\alpha = 0.05$}
  \end{subfigure}
  \begin{subfigure}[t]{0.48\textwidth}
    \includegraphics[width = 0.9\textwidth]{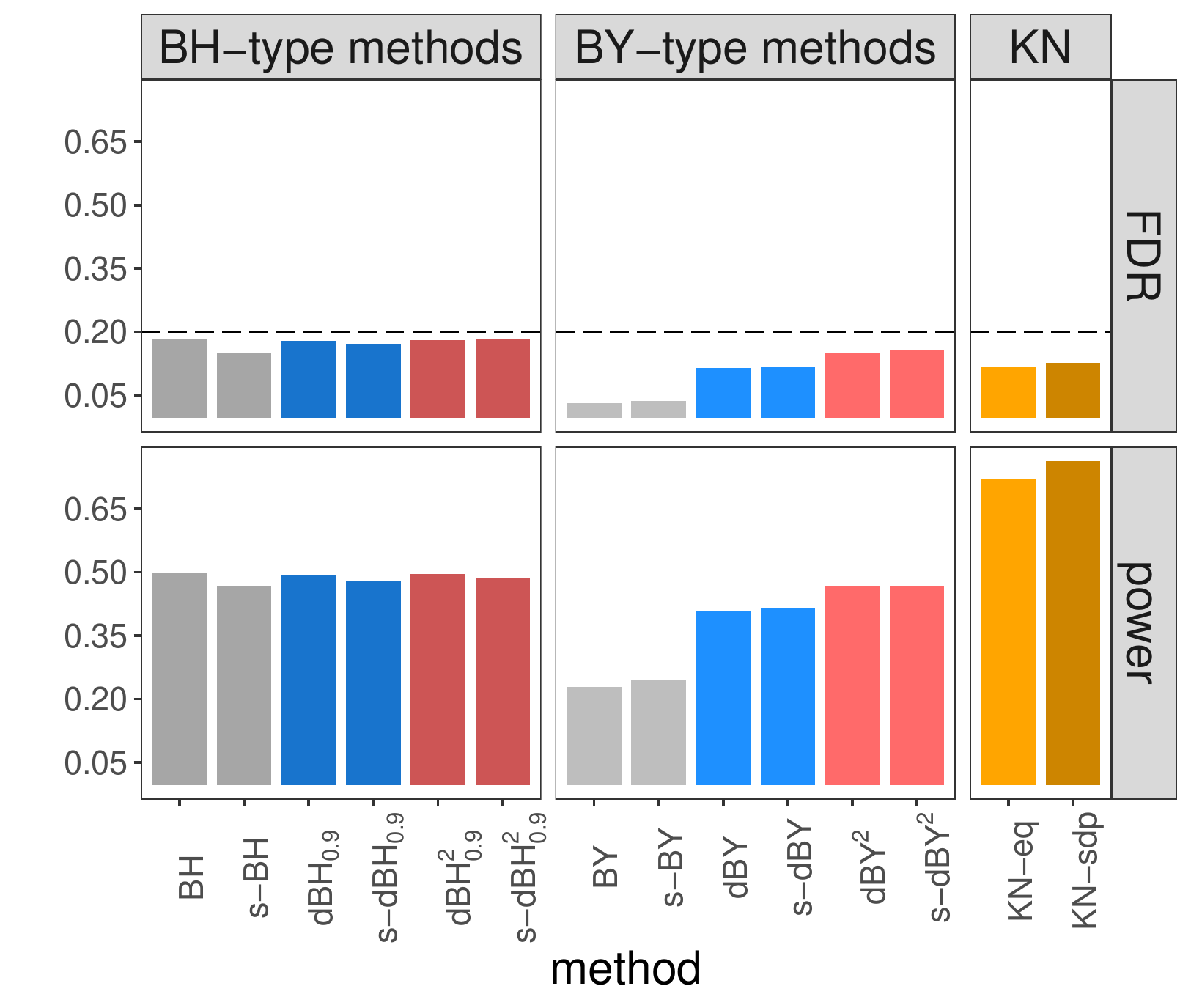}
    \caption{$\alpha = 0.2$}
  \end{subfigure}
  \caption{$X$ as a realization of a random Gaussian matrix with AR$(0.8)$ rows}
\end{figure}

\begin{figure}[H]
  \centering
  \begin{subfigure}[t]{0.48\textwidth}
    \includegraphics[width = 0.9\textwidth]{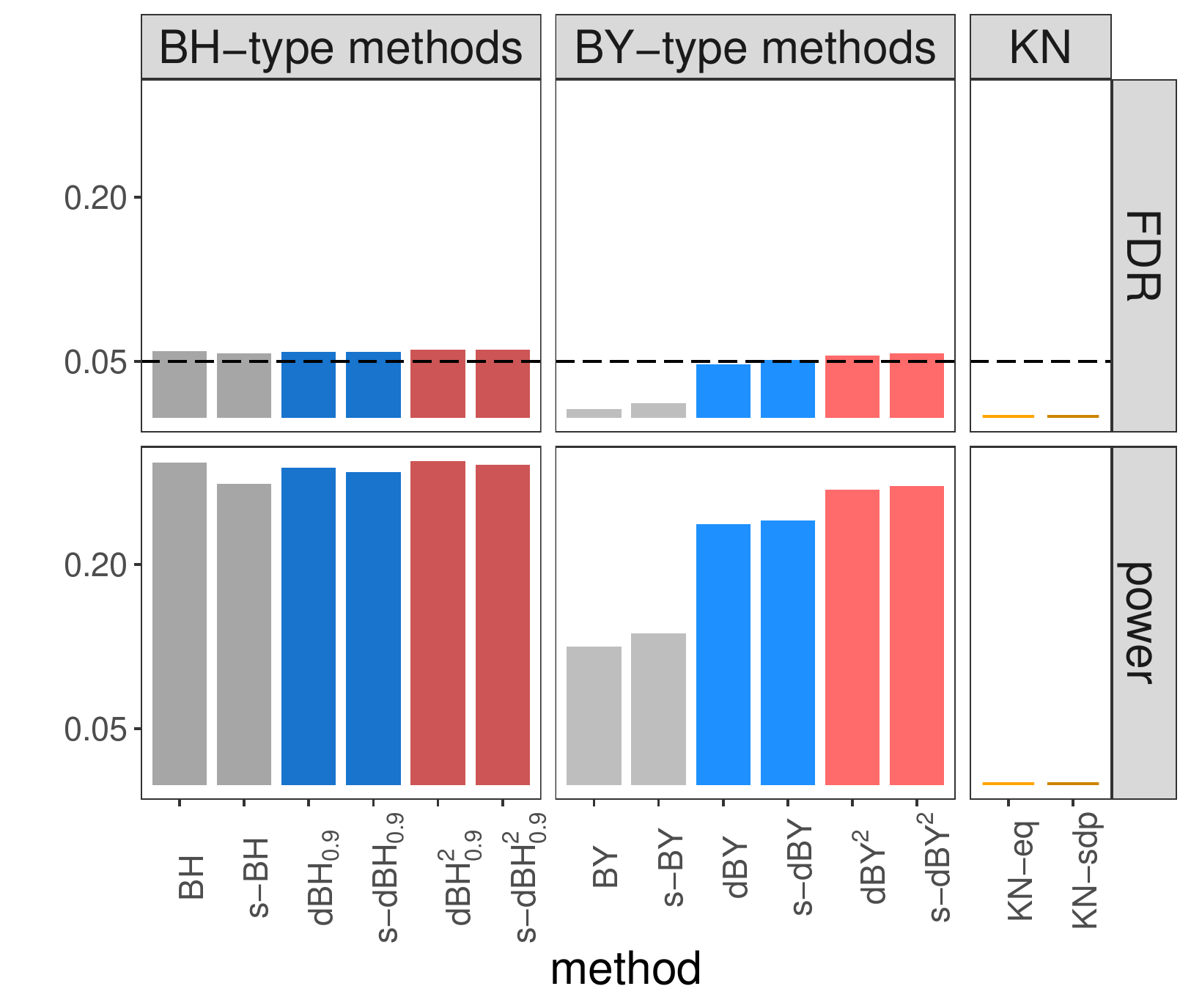}
    \caption{$\alpha = 0.05$}
  \end{subfigure}
  \begin{subfigure}[t]{0.48\textwidth}
    \includegraphics[width = 0.9\textwidth]{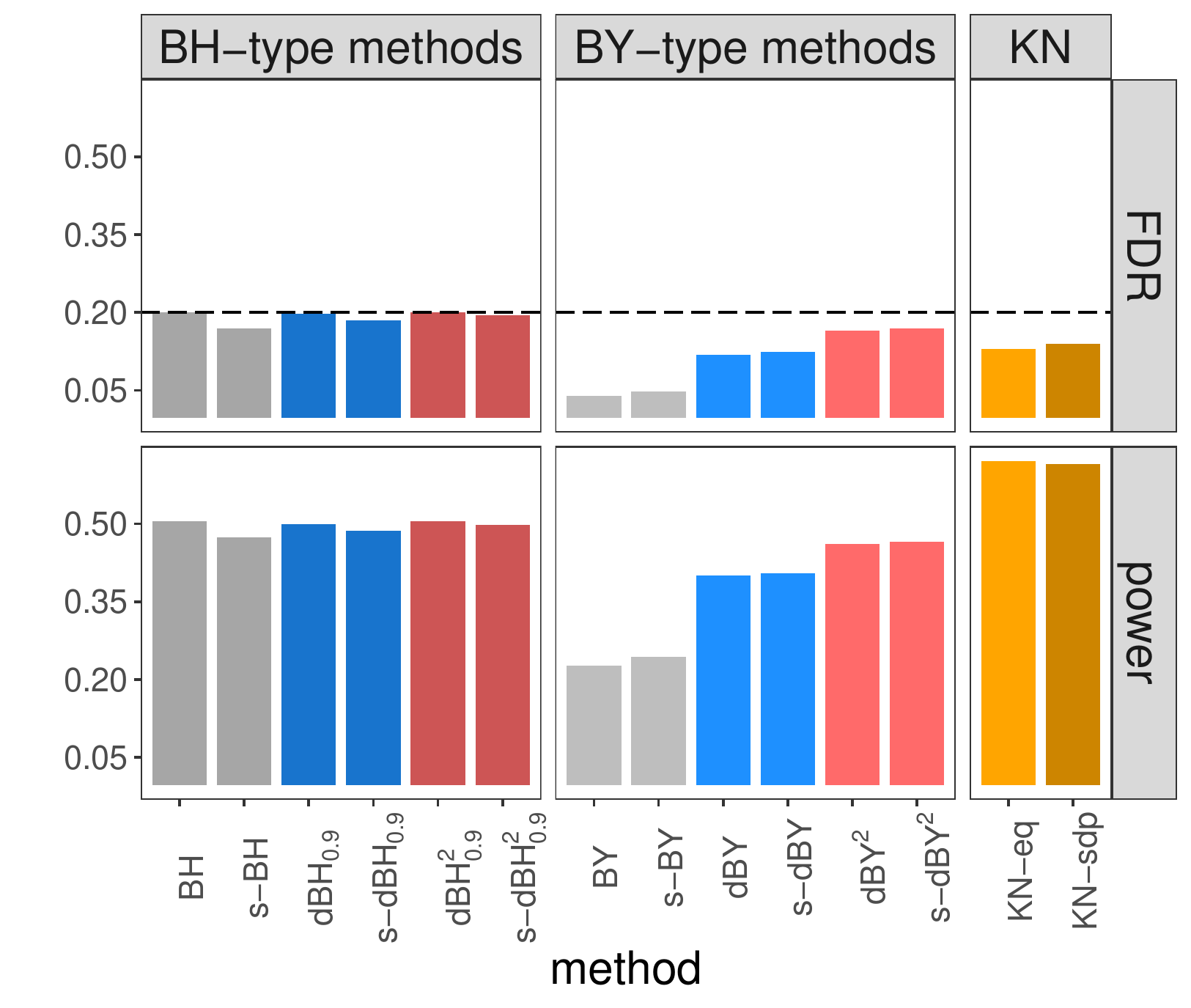}
    \caption{$\alpha = 0.2$}
  \end{subfigure}
  \caption{$X$ as a realization of a random Gaussian matrix with i.i.d. entries}
\end{figure}

\begin{figure}[H]
  \centering
  \begin{subfigure}[t]{0.48\textwidth}
    \includegraphics[width = 0.9\textwidth]{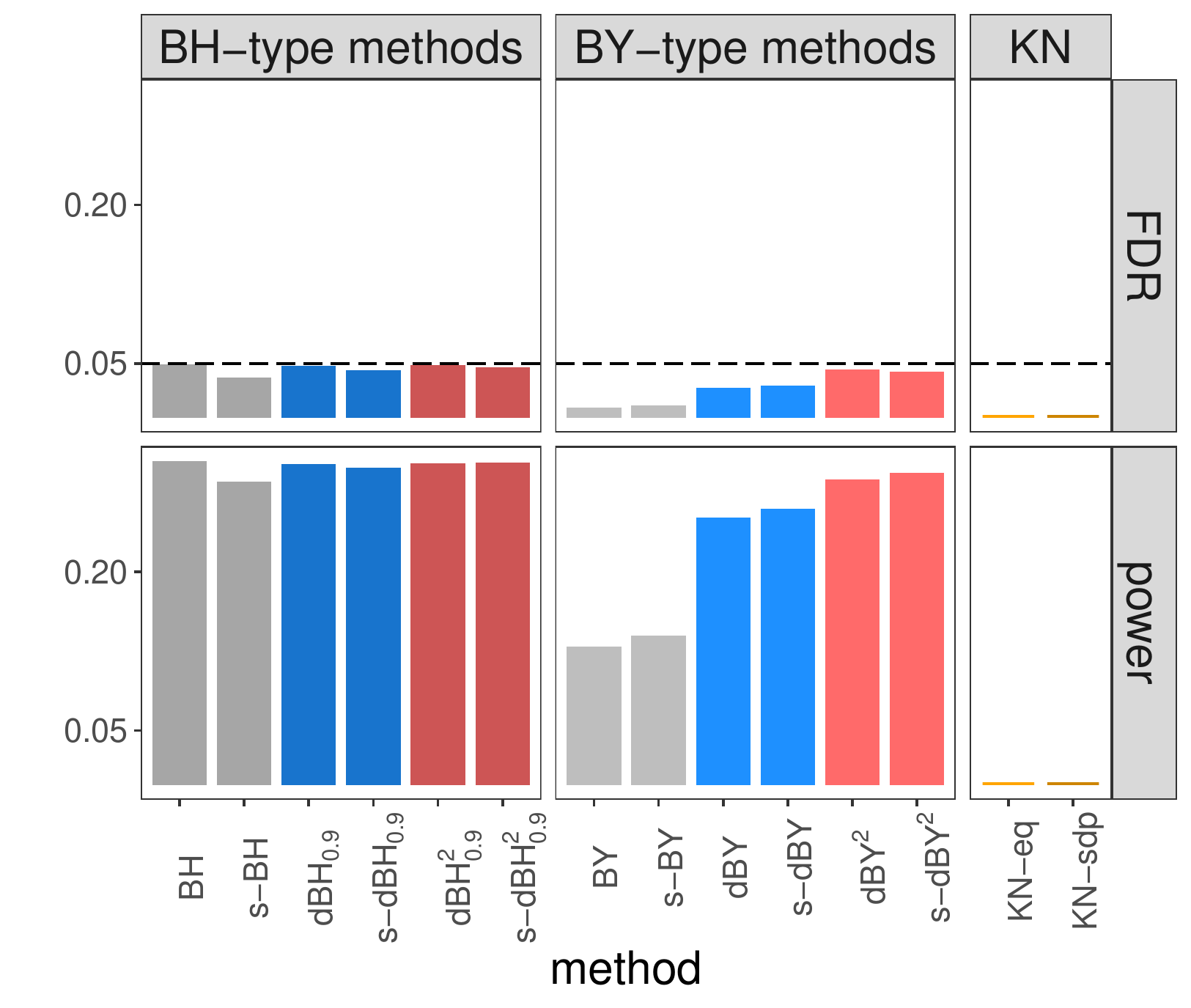}
    \caption{$\alpha = 0.05$}
  \end{subfigure}
  \begin{subfigure}[t]{0.48\textwidth}
    \includegraphics[width = 0.9\textwidth]{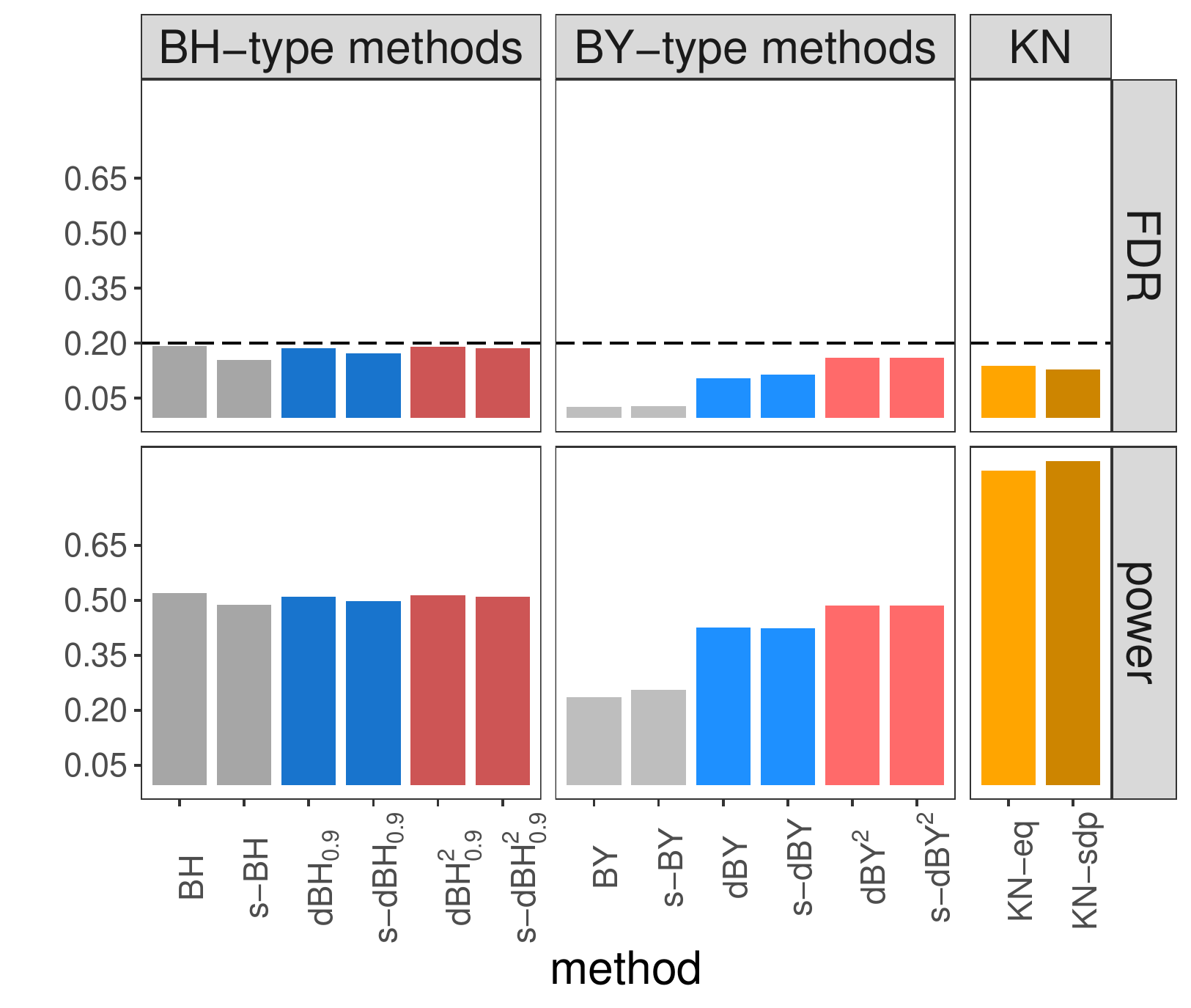}
    \caption{$\alpha = 0.2$}
  \end{subfigure}
  \caption{$X$ as a realization of a random Gaussian matrix with block dependent rows}
\end{figure}

\subsection{Multiple comparisons to control for Gaussian outcomes}
\begin{figure}[H]
  \centering
  \begin{subfigure}[t]{0.48\textwidth}
    \includegraphics[width = 0.9\textwidth]{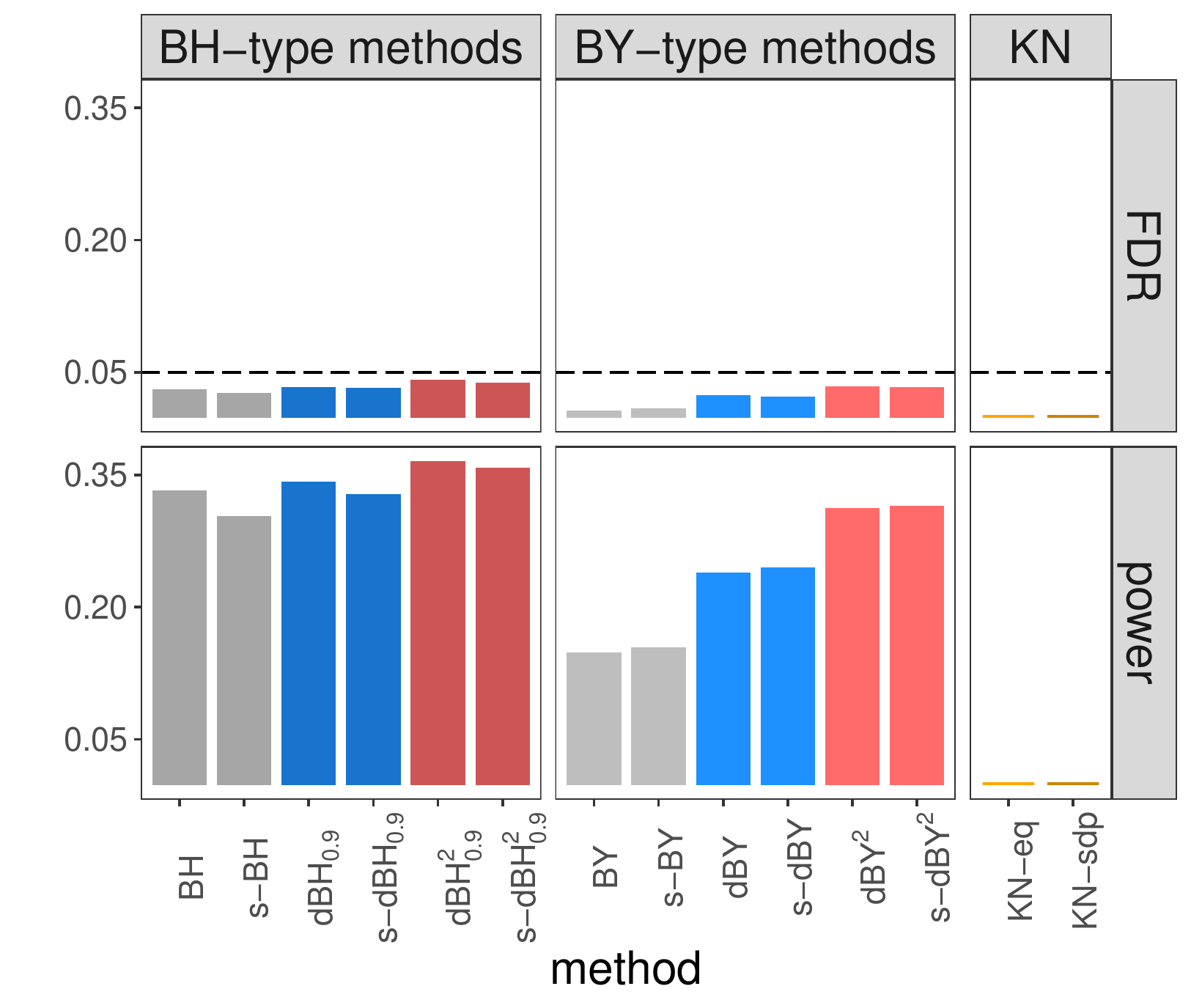}
    \caption{$\alpha = 0.05$}
  \end{subfigure}
  \begin{subfigure}[t]{0.48\textwidth}
    \includegraphics[width = 0.9\textwidth]{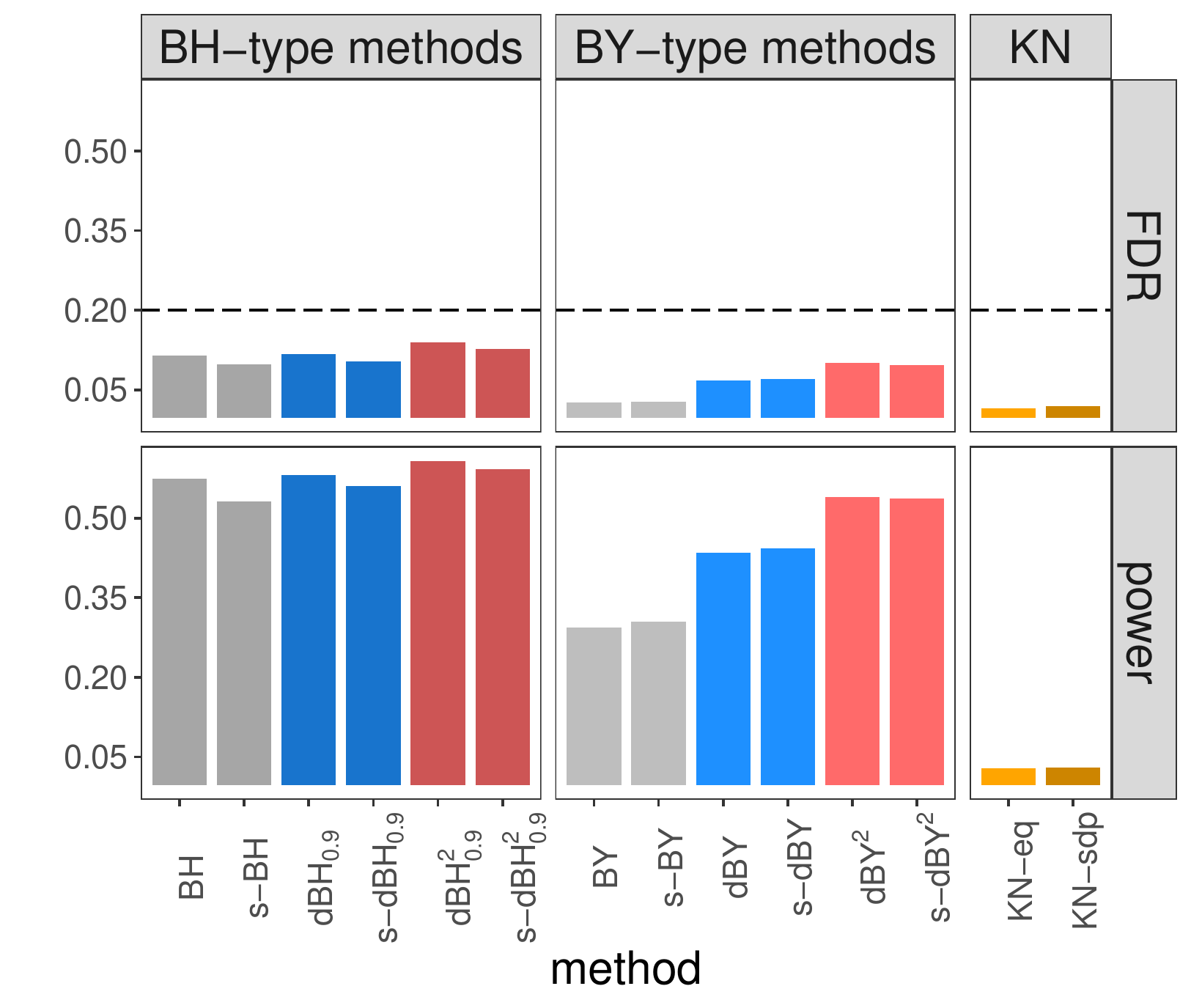}
    \caption{$\alpha = 0.2$}
  \end{subfigure}
  \caption{Multiple comparisons to control with $3$ replicates in each group}
\end{figure}

\begin{figure}[H]
  \centering
  \begin{subfigure}[t]{0.48\textwidth}
    \includegraphics[width = 0.9\textwidth]{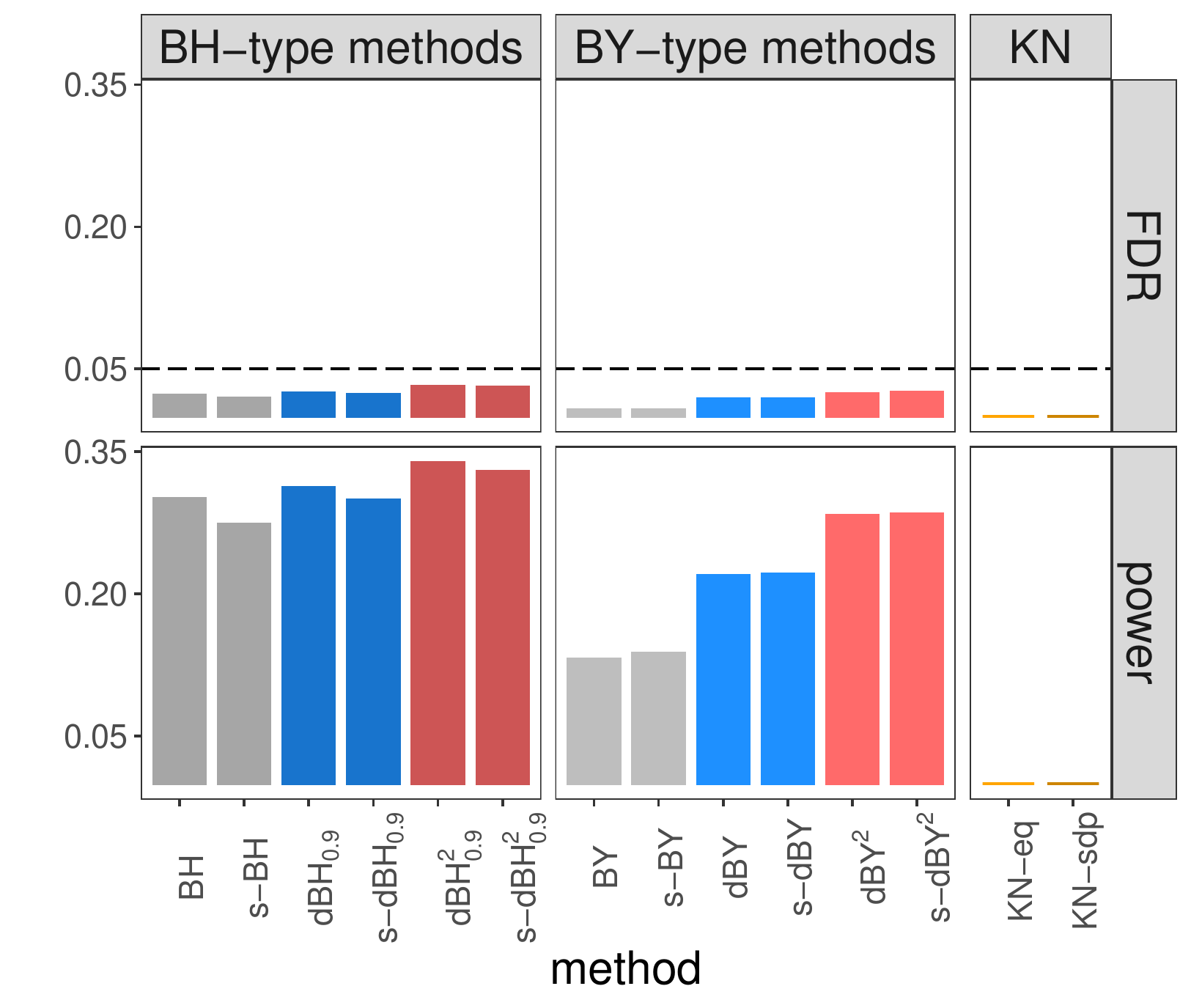}
    \caption{$\alpha = 0.05$}
  \end{subfigure}
  \begin{subfigure}[t]{0.48\textwidth}
    \includegraphics[width = 0.9\textwidth]{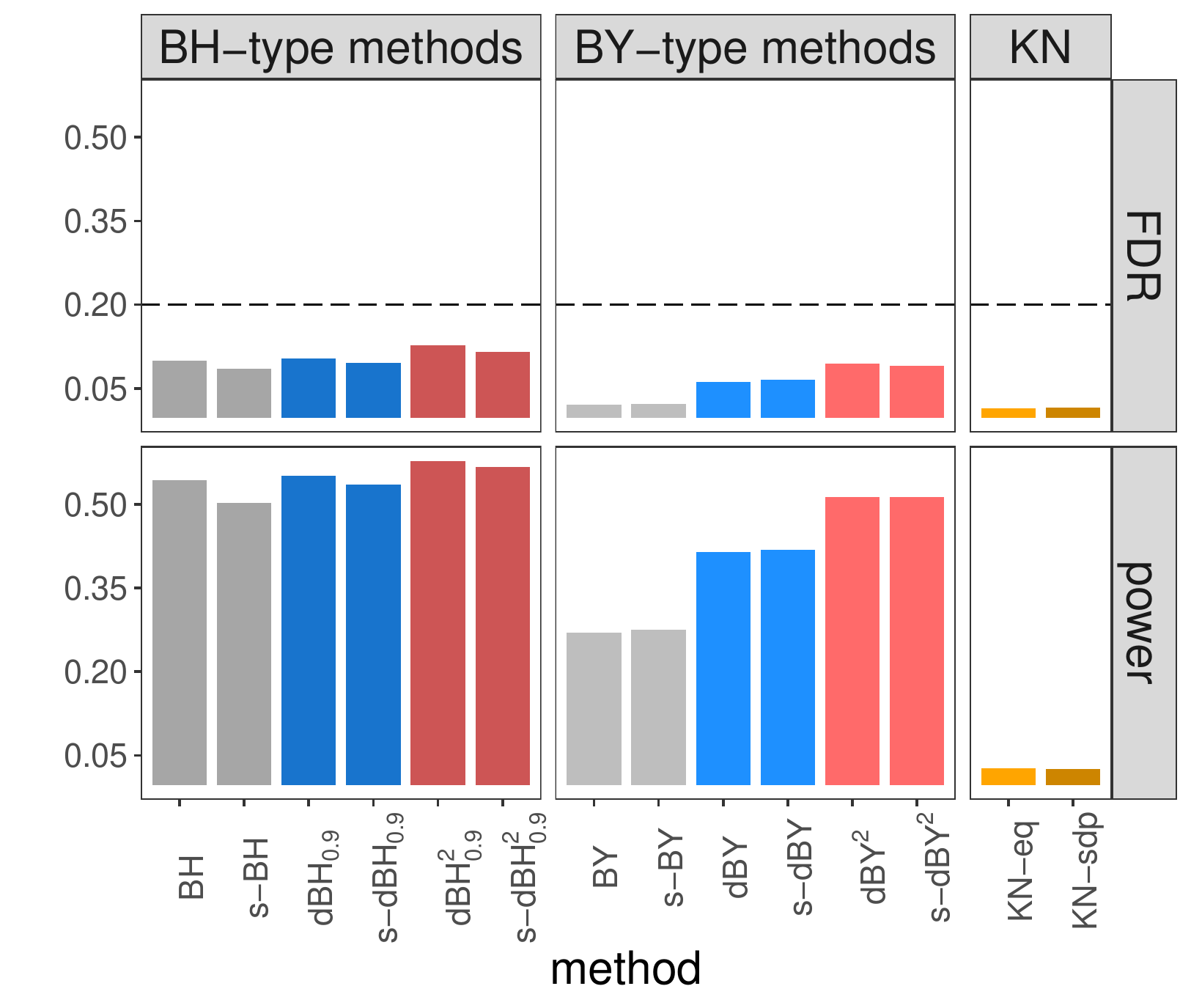}
    \caption{$\alpha = 0.2$}
  \end{subfigure}
  \caption{Multiple comparisons to control with $30$ replicates in each group}
\end{figure}

\newpage
\subsection{Estimated probability of the randomized pruning step}
In principle, the randomized pruning step may be invoked for dBH, s-dBH, dBH$^2$, and s-dBH$^2$, if the procedures are not safe. We summarize the fraction of simulations in which it is invoked for each case below. As desired, the chance of the randomized pruning step is extremely low. 

\begin{table}[H]
  \centering
  \begin{tabular}{c|c|c|cccc}
    \toprule
    \midrule
    & & & $\dBH_{0.9}$ & $\text{s-dBH}_{0.9}$ & $\dBH_{0.9}^2$ & $\text{s-dBH}_{0.9}^2$\\
    \midrule
    \multirow{6}{*}{\makecell{Multivariate \\ z-statistics}} & \multirow{3}{*}{One-sided} & AR$(0.8)$ & $0$ & $0$ & $0$ & $0$\\
    & & AR$(-0.8)$ & $0.012$ & $0.002$ & $0.001$ & $0$\\
    & & block & $0$ & $0$ & $0$ & $0$\\\cline{2-3}
    & \multirow{3}{*}{Two-sided} & AR$(0.8)$ & $0$ & $0$ & $0$ & $0$\\
    & & AR$(-0.8)$ & $0.002$ & $0$ & $0.001$ & $0$\\
    & & block & $0$ & $0$ & $0$ & $0$\\
    \midrule
    \multirow{6}{*}{\makecell{Multivariate \\ t-statistics}} & \multirow{3}{*}{One-sided} & AR$(0.8)$ & $0$ & $0$ & $0$ & $0$\\
    & & uncorrelated & $0$ & $0$ & $0$ & $0$\\
    & & block & $0$ & $0$ & $0$ & $0$\\\cline{2-3}
    & \multirow{3}{*}{Two-sided} & AR$(0.8)$ & $0$ & $0$ & $0$ & $0$\\
    & & uncorrelated & $0$ & $0$ & $0$ & $0$\\
    & & block & $0$ & $0$ & $0$ & $0$\\
    \midrule
    \multirow{6}{*}{\makecell{Linear \\ models}} & \multirow{3}{*}{$\alpha = 0.05$} & AR$(0.8)$ & $0$ & $0$ & $0$ & $0$ \\
    & & uncorrelated & $0$ & $0$ & $0$ & $0$\\
    & & block & $0$ & $0$ & $0$ & $0$\\\cline{2-3}
    & \multirow{3}{*}{$\alpha = 0.2$} & AR$(0.8)$ & $0$ & $0$ & $0$ & $0$\\
    & & uncorrelated & $0$ & $0$ & $0$ & $0$\\
    & & block & $0$ & $0$ & $0$ & $0$\\
    \midrule
    \multirow{4}{*}{\makecell{Multiple \\ comparisons \\ to control}} & \multirow{2}{*}{$\alpha = 0.05$} & 3 replicates & $0$ & $0$ & $0.001$ & $0$\\
    & & 30 replicates & $0$ & $0$ & $0$ & $0$\\\cline{2-3}
    & \multirow{2}{*}{$\alpha = 0.2$} & 3 replicates & $0.012$ & $0.004$ & $0.001$ & $0.001$\\
    & & 30 replicates & $0.016$ & $0.004$ & $0.003$ & $0.003$\\
    \midrule
    \bottomrule
  \end{tabular}
  \caption{Estimated probability of the randomized pruning step.}
\end{table}


\end{document}